\documentclass[12pt]{article}
\usepackage{amsmath}
\usepackage{amssymb}
\usepackage{amsthm}
\usepackage{mathrsfs}
\usepackage{geometry}
\usepackage{diagbox}
\usepackage{subfigure}
\usepackage[colorlinks,linkcolor=blue,citecolor=blue]{hyperref}
\usepackage[toc,page]{appendix}
\newtheorem{thm}{Theorem}[section]

\newtheorem{pro}{Proposition}[section]
\newtheorem{lem}{Lemma}[section]

\theoremstyle{definition}
\newtheorem{Ap}{Assumption}

\newtheorem{remark}{Remark}[section]
\usepackage{graphicx}
\usepackage{epstopdf}
\usepackage{enumerate}
\usepackage{color}
\usepackage{accents}

\usepackage[all]{xy}
\allowdisplaybreaks

\newcommand{\mfv}{\mathfrak{v}}

\newcommand{\bbB}{\boldsymbol{B}}

\newcommand{\bbS}{\boldsymbol{S}}

\newcommand{\bbw}{\boldsymbol{w}}
\newcommand{\bbx}{\boldsymbol{x}}
\newcommand{\bby}{\boldsymbol{y}}
\newcommand{\bbz}{\boldsymbol{z}}

\newcommand{\bbmu}{\boldsymbol{\mu}}

\newcommand{\bbGa}{\boldsymbol{\Gamma}}

\newcommand{\bbSig}{\boldsymbol{\Sigma}}

\newcommand{\mbE}{\mathbb{E}}
\newcommand{\mbP}{\mathbb{P}}

\newcommand{\mbN}{\mathbb{N}}

\newcommand{\mbR}{\mathbb{R}}

\newcommand{\mbI}{\mathbb{I}}

\newcommand{\mrO}{O}
\newcommand{\mro}{o}

\newcommand{\mcN}{\mathcal{N}}

\newcommand{\mcV}{\mathcal{V}}
\newcommand{\mcU}{\mathcal{U}}
\newcommand{\mcE}{\mathcal{E}}
\newcommand{\mcF}{\mathcal{F}}
\newcommand{\mcS}{\mathcal{S}}
\newcommand{\mcT}{\mathcal{T}}
\newcommand{\mcM}{\mathcal{M}}
\newcommand{\mcL}{\mathcal{L}}
\newcommand{\mcA}{\mathcal{A}}
\newcommand{\mcB}{\mathcal{B}}
\newcommand{\mcD}{\mathcal{D}}

\newcommand{\mcI}{\mathcal{I}}
\newcommand{\mcJ}{\mathcal{J}}
\newcommand{\mcP}{\mathcal{P}}

\newcommand{\fM}{\mathcal{M}}

\newcommand{\fy}{\mathfrak{y}}

\newcommand{\tr}{\operatorname{Tr}}
\newcommand{\diag}{\operatorname{diag}}
\newcommand{\Cov}{\operatorname{Cov}}
\newcommand{\Var}{\operatorname{Var}}

\newcommand{\sign}{\operatorname{sign}}

\newcommand{\RomanNumeralCaps}[1]
    {\MakeUppercase{\romannumeral #1}}

\title{Testing the Deviation of a High-dimensional Mean}
\author{Zengjing Chen\thanks{zjchen@sdu.edu.cn},\quad Ruihan Liu\thanks{rhliu@connect.hku.hk},\quad Jianfeng Yao\thanks{jeffyao@cuhk.edu.cn}\\
\small{School of Mathematics, Shandong University$^*$}\\
\small{Department of Statistics and Actuarial Science, The University of Hong Kong$^\dagger$}\\
\small{School of Data Science, The Chinese University of Hong Kong (Shenzhen)$^\ddagger$}}
\date{}
\begin{document}
\maketitle
\begin{abstract}   
    This paper investigates the high-dimensional one-sample equivalence testing problem for a reference mean vector $\bbmu_0$: $H_0:\|\bbmu-\bbmu_0\|>d_0$ versus $H_1: \|\bbmu-\bbmu_0\|\leq d_0$ for a pre-specified length threshold $d_0 > 0$. In high-dimensional regimes, classical equivalence tests suffer from severe asymptotic breakdowns. To resolve this difficulty, we propose a novel test statistic based on a two-armed bandit (TAB) U-process, which leverages positive/negative feedback mechanisms from control theory. We rigorously establish the weak convergence of TAB U-process to a novel two-dimensional stochastic differential equation (SDE), termed the U-bandit SDE. By employing in-depth analysis of this SDE, we characterize the asymptotic size and power of the resulting test statistic. The theoretical framework is extended to the two-sample setting and validated through finite-sample simulations and an empirical analysis of microbiome compositions. 
\end{abstract}
\noindent{\bf Keywords:} High-dimensional significance test; High-dimensional test for deviation of mean;    Two-armed bandit (TAB) process.

\section{Introduction}\label{Sec of Introduction}
\subsection{The equivalence testing problem}
The one-sample or two-sample significance test for a population mean   is a fundamental problem in high-dimensional statistics. Specifically, given a random sample $\{\bbx_t \in \mathbb{R}^n : t=1,\dots,T\}$ with the population mean vector $\bbmu$ and the population covariance matrix $\bbSig$, the one-sample significance test considers the hypotheses
\begin{equation}
	H_0^\texttt{e} : \bbmu = \bbmu_0 \quad \text{versus} \quad H_1^\texttt{e} : \bbmu \neq \bbmu_0,\label{Eq of equality test}
\end{equation}
where $\bbmu_0 \in \mathbb{R}^n$ is a given reference mean vector. The hypothesis test \eqref{Eq of equality test} is classical and has been comprehensively studied in the literature from perspectives of distance-type (Hotelling’s $T^2$ test), maximum-type ($\|\bbmu-\bbmu_0\|_{\infty}$) and projection-type statistics, readers can refer to \cite{hotelling1931generalization,bai1996effect,chen2010two,chen2011regularized} (distance-type), \cite{tony2014two,xu2016adaptive,chang2017simulation} (maximum-type) and \cite{lopes2011more,li2021linear,JMLR:v23:20-1103,liu2024projection} (projection-type) for comprehensive details.

From a real-world application perspective though, the significance test \eqref{Eq of equality test} has a main drawback. Because \eqref{Eq of equality test} operates under the premise that there is no true difference between $\bbmu$ and $\bbmu_0$. Thus, a failure to reject $H_0^{\mathtt{e}}$ is usually interpreted as affirmative evidence of $\bbmu=\bbmu_0$. This methodological blind spot becomes a critical vulnerability in applied sciences where asserting the safety, efficacy, or interchangeability of a product is paramount.

To address this limitation, the equivalence test is developed to invert the null explicitly represents non-equivalence, while the alternative hypothesis indicates equivalence.  Precisely, given a threshold  $d_0 > 0$, the equivalence test corresponding to \eqref{Eq of equality test} is
\begin{equation}
	H_0 : \|\bbmu - \bbmu_0\|_2 > d_0 \quad \text{versus} \quad H_1 : \|\bbmu - \bbmu_0\|_2 \leq d_0. 
	\label{Eq of equivalence test}
\end{equation}
Equivalence is only concluded when $H_0$ is successfully rejected based on statistical evidence, establishing that the observed differences fall within a pre-specified, practically meaningful threshold $d_0>0$.

Practically, the necessity of \eqref{Eq of equivalence test} is most acutely demonstrated in pharmacological bioequivalence. For example, regulatory bodies must ascertain whether a generic drug can safely replace an expensive brand-name drug. Average bioequivalence requires evaluating whether the average responses among patient populations taking the brand-name drug ($\bbmu_0$) and the generic copy ($\bbmu$) are sufficiently close with a pre-specified threshold $d_0$. If regulators relied on traditional test \eqref{Eq of equality test}, an underpowered clinical trial could fail to detect a difference, allowing a potentially unsafe or ineffective generic drug to reach the market under the false assumption of equivalence.

It turns out that constructing an appropriate test statistic for univaraite the equivalence test \eqref{Eq of equivalence test} is technically nontrivial. Note that the equivalence test \eqref{Eq of equivalence test} consists of two disjoint regions, it can be adapted into two one-sided tests, where the Student's-$t$ statistic for the univariate significance test \eqref{Eq of equality test} still works, as its variance is independent of the population mean. This straightforward extension of the $T^2$ statistic for the univariate equivalence test \eqref{Eq of equivalence test} is called the ``\emph{Two One-Sided Tests}'' (TOST) procedure, thoroughly evaluated by \cite{schuirmann1987comparison}. However, this straightforward method suffers from the power loss due its complicated boundary conditions of the rejection region.

Another notable difficulty is that the parameter spaces of \eqref{Eq of equivalence test} under $H_0$ and $H_1$ share the same dimension, so the classic likelihood ratio test (LRT) cannot be directly applied here. Precisely, under this situation, \cite{self1987asymptotic} showed that the statistic of LRT does not converge to a standard continuous distribution, like the $\chi^2$-distribution, and boundaries of corresponding rejection regions are also rather complicated. Thus, it usually demands complex numerical optimization algorithms for calculating the critical values.

\subsection{Asymptotic breakdown of classical equivalence tests}
Recently, the transition to multivariate or even high-dimensional equivalence testing is fundamentally driven by the complexity of modern biological and chemical data. In pharmacological trials, researchers frequently need to prove equivalence across multiple pharmacokinetic parameters simultaneously. In agricultural biotechnology, establishing the safety of a genetically modified (GM) crop requires comparing dozens of biochemical components—such as macronutrients, micronutrients, and potential toxicants—to a reference crop.

In this article, we focus on the equivalence test \eqref{Eq of equivalence test} in high-dimensional settings, that is, the data dimension $n$ and the sample size $T$ tend to infinity proportionally.

Due to the ``curse of dimensionality,'' directly extending univariate equivalence testing procedures to high-dimensional regimes fundamentally breaks down. In the univariate case, the traditional TOST is applicable because the variance of $T^2$ statistics is independent of the population mean $\mu$. However, in high dimensions, this independence completely shatters: the asymptotic variance of any consistent estimator for $\bbmu$ depends heavily on the unknown covariance matrix $\bbSig$, whose value is not uniquely specified under the $L^2$-norm constraint of $H_0$. This dependence renders the derivation of valid rejection regions exceedingly more difficult.

For the LRT, it relies heavily on the exact specification of the statistical model and not robust enough, resulting in severe Type \RomanNumeralCaps{1} error inflation, see \cite{li2016likelihood}. Moreover, due to the strictly constrained for the boundaries of rejection region, calculating the critical is general numerically complex, the significantly increasing computational burden leads LRT inefficient when dealing with multivariate or high-dimensional data, see \cite{claeskens2004restricted}.

Furthermore, one might address the test  \eqref{Eq of equivalence test} using a multiple testing procedure like \cite{benjamini1995controlling} to control the False Discovery Rate (FDR) for the test \eqref{Eq of equality test}, adapting these procedures for the test \eqref{Eq of equivalence test} is not straightforward, as the non-null distance $d_0$ represents a global threshold that cannot be arbitrarily distributed across $n$ dimensions.

\subsection{Two-armed bandit U-statistic for high-dimensional \texorpdfstring{\eqref{Eq of equivalence test}}{(2)}}
To bypass the limitations of TOST and LRT, a recent remarkable methodology \cite{chen2023strategic} successfully leveraged the theory of reinforcement learning \cite{sutton1998reinforcement}, specifically developing test statistics via a strategic-driven two-armed bandit (TAB) process for the univariate random samples $\{x_t \in \mbR : t=1,\dots,T\}$. Precisely, the test statistic $\{\hat{b}_{t} : t=1,\dots,T\}$ is sequentially constructed as below: 
\begin{align}
	\hat{b}_{t} =\left\{\begin{array}{ll}
		h(x_1,\theta_1), & t=1, \\
		\hat{b}_{t-1} + h(x_t,\theta_t), & t=2,\dots,T,
	\end{array}\right.\label{Eq of TAB construction 1}
\end{align}
where $h(\cdot,\cdot)$ is a well-chosen innovation function with specially designed strategies $(\theta_t)$. Basically, $h(\cdot,\cdot)$ can be viewed as a slot machine with two arms, $h(\cdot,\text{L})$ (left arm) and $h(\cdot,\text{R})$ (right arm). Based on the value of $\hat{b}_{t-1}$ from the previous step $t-1$, $\theta_t\in\{\text{L,R}\}$ decides the choice of arm to lead the test statistic $\hat{b}_{t}$ exhibits different asymptotic properties under $H_0$ and $H_1$.

Different with the classic TAB strategies $(\theta_t)$ (e.g., \cite{gittins1979bandit,whittle1988restless,perchet2013multi}), which select arms for maximizing $\mbE[\hat{b}_t]$ as $t\to\infty$, \cite{chen2023strategic} leverages the idea of positive/negative feedback processes from control theory. Precisely, the specially designed strategies $\theta_t$ drives the statistic $\hat{b}_{t}$ follows a negative feedback process under $H_0$ (asymptotically concentrating around 0) and a positive feedback process under $H_1$ (asymptotically diverging to infinity), which is the key for solving \eqref{Eq of equivalence test} in the univariate case. Particularly, the asymptotic distribution of $\hat{b}_{t}$ satisfies a ``\emph{strategic central limit theorem}'' from the nonlinear probability theory (see \cite{peng2019nonlinear}), which is more concentrated around 0 under $H_0$ and   diverging under $H_1$ faster  than the standard $t$ or $T^2$-statistic, thereby enhancing the testing power.


Crucially, the mathematical viability of the univariate TAB process relies on the target variable $h(x_t,\theta_t)$ being strictly independent of the sequential statistic $\hat{b}_{t-1}$. In high dimensions, pursuing the framework from \cite{chen2023strategic} and constructing a proper TAB-based statistic is technically challenging. The core difficulty is the construction of innovation function as in \eqref{Eq of TAB construction 1}. Inspired by the U-statistic introduced in \cite{chen2010two} for the significance test \eqref{Eq of equality test}, we propose a novel TAB-based U-process $\{\fM_t\}$ for testing high-dimensional \eqref{Eq of equivalence test}, where the two-armed innovation function $\tilde{h}$ relies on both the current observation $\bbx_t$ and historical samples $\bbx_s$ (where $s < t$):
\begin{align*}
	\fM_t=\fM_{t-1}+ \tilde{h}(\bbx_t,\bbx_1,\cdots,\bbx_{T_1},\theta_t),\quad t>T_1,
\end{align*}
where $T_1$ is proportional to $T$. This structural necessity introduces a profound theoretical challenge: the sequential statistic $\fM_{t-1}$ becomes highly correlated with the innovation at step $t$. This wide-spread correlation destroys the classical one-dimensional limit, necessitating a totally different asymptotic distribution.

Consequently, due to the widely spread correlation, we establish that the proposed TAB-based U-statistic weakly converges to a novel two-dimensional stochastic differential equation (SDE): 
\begin{align}
	\left\{\begin{array}{l}
		dY_s=\kappa\sign(Y_s)ds+(1-2\beta\sign(Y_s)V_s)^{1/2}dB_s, \\
		dV_s=-\sign(Y_s)ds,
	\end{array}\right.\quad Y_0=V_0=0;s\in[0,1],\label{Eq of U bandit SDE}
\end{align}
which we define as the \emph{U-bandit SDE}. The resulting test statistic $\fM_T$ is the \emph{dependent bandit} distribution. To rigorously validate this high-dimensional equivalence test, the main theoretical contributions of this article are:
\begin{enumerate}[(i)]
	\item \emph{Analysis of the U-Bandit SDE}: We conduct a rigorous stochastic analysis of the two-dimensional U-bandit SDE, utilizing Girsanov's transformations and martingale frameworks to prove the existence of its weak solution and establish exponential tail probabilities. Additionally, smoothing this SDE confirms that its Kolmogorov backward equation has bounded spatial derivatives.
	\item \emph{Weak Convergence of the TAB U-Process}: We establish the asymptotic weak convergence of the highly correlated TAB-based U-process to the continuous U-bandit SDE, demonstrating that the discrete sequence is governed by the dependent bandit distribution in the limit.
	\item \emph{Size and Power Characteristics}: Leveraging the bounded tail probabilities of the SDE, we derive the exact asymptotic size and power of the TAB-based U-statistic, providing a mathematically robust and computationally simple rejection region that bypasses the complexities of TOST and LRT.
\end{enumerate}
The contributions presented above are novel. Compared with the univariate case in \cite{chen2023strategic}, their statistics rely on a simple one dimensional SDE studied in \cite{mel1979strong}. Our study not only comprehensively extend the univariate TAB-process in \cite{chen2023strategic} to more complicated and correlated TAB-U process, but also provide a rigorous framework for the high-dimensional equivalence test \eqref{Eq of equivalence test}.

The rest of this paper is organized as follows: \S\ref{Sec of one sample} formalizes the equivalence test \eqref{Eq of equivalence test} for the one-sample case and derives its asymptotic distribution. \S\ref{Sec of two sample} extends the previous framework to the two-sample setting. \S\ref{Sec of numerical experiment} evaluates finite-sample performance via numerical experiments, and \S\ref{Sec of empirical} demonstrates the methodology through an empirical analysis of microbiome compositions. All rigorous proofs of our main results are postponed to the Appendix.

We end this section by listing useful notations:

\begin{enumerate}
    \item For two sequences $\{a_n\}$ and $\{b_n\}$, $a_n \le O(b_n)$ implies $a_n \le C b_n$ for some constant $C > 0$ as $n$ becomes sufficiently large. The equality $a_n = O(b_n)$ means $a_n \le O(b_n)$ and $b_n \le O(a_n)$. And $|a_n| = o(|b_n|)$ means $\lim_{n\to\infty}|a_n/b_n|=0$.
    \item Given a matrix $\mathbf{A} = [A_{i,j}]_{n \times n}$, $\text{Tr}(\mathbf{A}) = \sum_{i=1}^n A_{i,i}$ and $\mathbf{A}'$ denotes the transpose. $\|\mathbf{A}\|$ denotes the spectral norm.
    \item $\Vert\cdot\Vert_2$ represent the $L^2$ norms of vectors.
    \item The symbol $\overset{\mbP}{\longrightarrow}$ denotes the convergence in probability.
    \item For sequences of random variables $\{X_n\}$ and $\{Y_n\}$, $X_n \overset{d}{\sim} Y_n$ as $n \to \infty$ means that for any $x \ge 0$, $\lim_{n \to \infty} |\mathbb{P}(|X_n| \le x) - \mathbb{P}(|Y_n| \le x)| = 0$.
\end{enumerate}

\section{One-sample case}\label{Sec of one sample}
In this section, we consider the high-dimensinoal equivalence test \eqref{Eq of equivalence test} under the one-sample case. Since the reference mean $\bbmu_0$ is given in practice, without loss of generality, we assume $\bbmu_0=\boldsymbol{0}$ in \eqref{Eq of equivalence test}. First, we construct the statistic for \eqref{Eq of equivalence test} in \S\ref{sec of norm test} and discuss the corresponding rationale in \S\ref{sec of Rationale}, the asymptotic distribution of the proposed statistic for \eqref{Eq of equivalence test} is given in \S\ref{sec of asymptotic distribution}. Finally, we discuss the asymptotic size and power of the proposed statistic in \S\ref{Sec of size and power}.

\subsection{Test statistic for \texorpdfstring{\eqref{Eq of equivalence test}}{(2)}}\label{sec of norm test}
With $\bbmu_0=\boldsymbol{0}$, the test \eqref{Eq of equivalence test} becomes
\begin{align}
    H_0:\Vert\bbmu\Vert_2> d_0\quad\text{versus}\quad H_1:\Vert\bbmu\Vert_2\leq d_0.\label{Eq of high dimensional mean test}
\end{align}
Moreover, we adopt the following asymptotic regime.
\begin{Ap}\label{Ap of high dimensionality}
    Both the sample size $T$ and the dimension $n=n(T)$ of $\bbx_t$ grow to infinity such that
    \begin{align*}
        \lim_{T\to\infty}n/T=c_0\in[0,\infty).
    \end{align*}
\end{Ap}
Note that our main results are valid under $\lim_{T\to\infty}n/T=0$, which includes the classical low-dimensional case with fixed $n$ and $T\to\infty$.

To construct the statistic for \eqref{Eq of high dimensional mean test}, let $T=T_1+T_2$ such that $\lim_{T\to\infty}T_1/T=c_1\in(0,1)$. Define
\begin{align}
    \left\{\begin{array}{l}
        \displaystyle X_t=\frac{1}{T_1}\sum_{s=1}^{T_1}(\bbx_s'\bbx_{t+T_1}-d_0^2),\quad t=1,\cdots,T_2,\\[15pt]
        \displaystyle\hat{\tau}_{T_2}:=\frac{1}{T_2}\sum_{t=1}^{T_2}X_t\quad\text{and}\quad\hat{\sigma}_{T_2}^2:=\frac{1}{T_2}\sum_{t=1}^{T_2}X_t^2-\hat{\tau}_{T_2}^2.
    \end{array}\right.\label{Eq of estimators}
\end{align}
Let
\begin{align}
    \mcM_t:=\frac{1}{\sqrt{T_2}}\sum_{s=1}^t\frac{\theta_sX_s}{\sqrt{\hat{\tau}_{T_2}^2+\hat{\sigma}_{T_2}^2}},\quad t=1,\cdots,T_2,\label{Eq of statistic of similarity test}
\end{align}
where $\vec{\theta}_t=(\theta_1,\cdots,\theta_t)'$ such that $\theta_1=1$ and for $t=2,\cdots,T_2$
\begin{align}
    \theta_t=\left\{\begin{array}{cc}
        1, & \text{if }\ \mcM_{t-1}\leq0, \\[6pt]
        -1, & \text{if }\  \mcM_{t-1}>0.
    \end{array}\right.\label{Eq of control}
\end{align}
Note that the construction process \eqref{Eq of statistic of similarity test} satisfies that for $t\geq2$
\begin{align}
    \mcM_t=\mcM_{t-1}+\frac{\theta_tX_t}{\sqrt{T_2\big(\hat{\tau}_{T_2}^2+\hat{\sigma}_{T_2}^2\big)}}:=\mcM_{t-1}+\tilde{h}(X_t),\label{Eq of TAB procedure 2}
\end{align}
where $\mbE[X_t]=\Vert\bbmu\Vert_2^2-d_0^2$ by \eqref{Eq of estimators}. Intuitively, under $H_0$, $\mbE[X_t]\geq0$ and the strategy $\theta_t$ makes the innovation $\tilde{h}(X_t)$ to have in average an opposite sign to the current value $\mcM_{t-1}\leq0$ and vice versa, which defines a negative feedback process. By contrast, $\{\mcM_t\}$ acts as a positive feedback under $H_1$, where $\mbE[X_t]<0$. Hence, as $t\to\infty$, the negative feedback process will restrict our statistic $\mcM_t$ asymptotically concentrating around 0 under $H_0$, while the positive feedback process will push $\mcM_t$ asymptotically diverging to infinity under $H_1$.

\subsection{Rationale of the test statistic for \texorpdfstring{\eqref{Eq of equivalence test}}{(2)}}\label{sec of Rationale}
There is important difference between our test statistic $\mcM_t$ in \eqref{Eq of statistic of similarity test} and the one for the univariate equivalence test in \cite{chen2023strategic}, the former having little in common with the latter.  Precisely, given the independent univariate samples $\{x_t\in\mbR:t=1,\cdots,T\}$ with population mean of $\mu>0$, \cite{chen2023strategic} consider the following univariate equivalence test:
\begin{align}
    H_0^{\texttt{u}}:\mu-d_0>0\quad\text{versus}\quad H_1^{\texttt{u}}:\mu-d_0\leq0.\label{Eq of equivalence test one dim}
\end{align}
\cite{chen2023strategic} constructed their statistic $\hat{b}_{T,t}$ as in \eqref{Eq of TAB construction 1}, where $\hat{b}_{T,1}=x_1-d_0$ and for $t=2,\cdots,T$
\begin{align}
    \hat{b}_{T,t}=\hat{b}_{T,t-1}+h(x_t),\quad h(x_t)=\eta_t(x_t-d_0),\quad\eta_t=\left\{\begin{array}{cc}
        1, & \text{if }\ \hat{b}_{T,t-1}\leq0, \\[6pt]
        -1, & \text{if }\  \hat{b}_{T,t-1}>0.
    \end{array}\right.\notag
\end{align}
By our previous discussion, $\hat{b}_{T,t}$ acts as a negative/positive feedback process under $H_0/H_1$. Notably, to derive such properties, the innovation function $h(\cdot)$ is a product of the control $\eta_t$ and the target variable $\mcT_t:=x_t-d_0$, where $\eta_t$ controls the sign of $\mcT_t$ based on current value $\hat{b}_{T,t-1}$ and   $\mbE[\mcT_t]=\mu-d_0$ is the term to test  in \eqref{Eq of equivalence test one dim}. Particularly, the independence between $\mcT_t$ and $\hat{b}_{T,t-1}$ is important in technical proofs in \cite{chen2023strategic}.

Back to the high-dimensional equivalence test \eqref{Eq of high dimensional mean test}, to extend \cite{chen2023strategic}'s idea above, the key is to construct the innovation function $\tilde{h}(\cdot)$ in \eqref{Eq of TAB procedure 2}. First, the control $\theta_t$ in \eqref{Eq of control} parallels its univariate counterpart $\eta_t$ above. As for the target variable $\tilde{\mcT}_t$ at each step $t$, we should have   $\mbE[\tilde{\mcT}_t ]=\Vert\bbmu\Vert_2^2-d_0^2$. A natural and direct idea to consider
\begin{align*}
    \tilde{\mcT}_t=\bbx_t'\bbx_t-\tr(\bbSig)-d_0^2.
\end{align*}
Obviously, $\mbE[\tilde{\mcT}_t]=\Vert\bbmu\Vert_2^2-d_0^2$ and $\tilde{\mcT}_t$ is independent of $\hat{b}_{T,t-1}^{\texttt{h}}$, so $\hat{b}_{T,t}^{\texttt{h}}$  follows a negative/positive feedback process as expected. However, even though $\tr(\bbSig)$ is known, $\tilde{\mcT}_t$ faces a main drawback that $\Var(\tilde{\mcT}_t)\asymp\mrO(n)$ compared with $\mbE[\tilde{\mcT}_t]\leq\mrO(1)$. Therefore, the huge variance of $\tilde{\mcT}_t$ will complete mask any significance of the testing term $\mbE[\tilde{\mcT}_t]=\Vert\bbmu\Vert_2^2-d_0^2$ and makes the resulting test powerless.

Therefore, a suitable target variable $\tilde{\mcT}_t$ should simultaneously satisfy $\mbE[\tilde{\mcT}_t]=\Vert\bbmu\Vert_2^2-d_0^2$ and $\Var(\tilde{\mcT}_t)\asymp\mrO(1)$. In high-dimensional statistical literature, a widely-used estimator of $\Vert\bbmu\Vert_2^2$ is $\frac{1}{T(T-1)}\sum_{s\neq t}^T\bbx_s'\bbx_t$, see \cite{chen2010two}. Inspired by this estimator, we consider $X_t=\frac{1}{T_1}\sum_{s=1}^{T_1}(\bbx_s'\bbx_{t+T_1}-d_0^2)$ in \eqref{Eq of estimators} as the target variable, we can show that $\mbE[X_t]=\Vert\bbmu\Vert_2^2-d_0^2$ and $\Var(X_t)\asymp\mrO(1)$.

On the other hand, note that $X_t$ is $\sigma\{\bbx_1,\cdots,\bbx_{T_1},\bbx_{t+T_1}\}$-measurable\footnote{$\sigma\{\bbx_1,\cdots,\bbx_{T_1},\bbx_{t+T_1}\}$ is the sigma field generated by $\{\bbx_1,\cdots,\bbx_{T_1},\bbx_{t+T_1}\}$.}, so $\{X_1,\cdots,X_t\}$ are all correlated. By \eqref{Eq of statistic of similarity test}, $\mcM_{t-1}$ is $\sigma\{X_1,\cdots,X_{t-1}\}$-measurable and thus correlated with the $X_t$. This is a key new difficulty in the high-dimensional case, very different of the univariate case in \cite{chen2023strategic}. To overcome this difficulty, we specially analyze the concentration property of $\mbE[X_t|X_1,\cdots,X_{t-1}]$, and provide the explicit formula about correlations between $\mcM_{t-1}$ and $\tilde{h}(X_t)$. Readers can refer to proofs in the Appendix \S\ref{sec of proof norm test} for more details about this key issue.

\subsection{Asymptotic distribution of the statistic \texorpdfstring{\eqref{Eq of statistic of similarity test}}{(7)}}\label{sec of asymptotic distribution}
Asymptotically, $\mcM_{T_2}$ weakly converges to the \emph{dependent bandit distribution} $\mcD\mcB(\kappa,\beta)$, which is the distribution of $Y_1$ from the U-bandit SDE defined in \eqref{Eq of U bandit SDE}. Readers can refer to Figure \ref{Fig of spiked distribution} for a density plot of bandit distributions under different values of $\kappa$.
\begin{figure}[ht]
    \centering
    \includegraphics[width=0.9\linewidth]{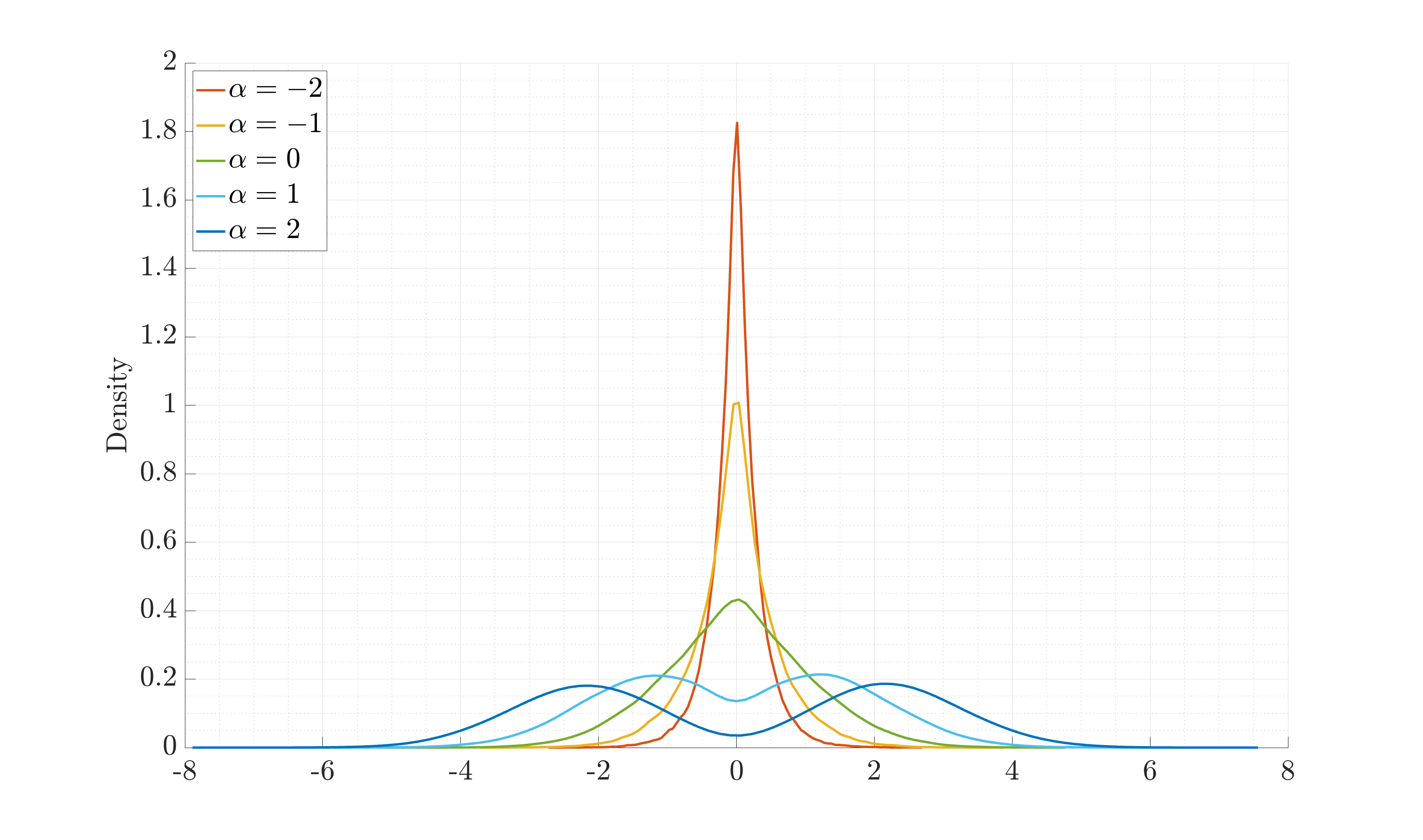}
    \caption{Density plot of $\mcD\mcB(\kappa,\beta)$ via Euler-Maruyama Monte Carlo simulations under different values of $\kappa$ and $\beta=1/4$.}
    \label{Fig of spiked distribution}
\end{figure}
Notably, when $\kappa<0$, $\mcD\mcB(\kappa,\beta)$ is unimodal and more concentrated around 0 than the standard normal distribution, while when $\kappa>0$, $\mcD\mcB(\kappa,\beta)$ has two modes dispersing away from 0. Therefore, the key distinction between the two hypotheses in \eqref{Eq of equivalence test} is the different forms of $\mcD\mcB(\kappa,\beta)$'s density under $\kappa<0$ ($H_0$) and $\kappa\geq0$ ($H_1$).

We now specify the data generation process.
\begin{Ap}\label{Ap of noise}
    We assume that $\bbx_t=\bbmu+\bbGa\bby_t\in\mbR^n$ such that $\bbGa\bbGa'=\bbSig\in\mbR^{n\times n}$ is invertible and $\Vert\bbmu\Vert_2,\Vert\bbSig\Vert,\Vert\bbSig^{-1}\Vert\leq\mrO(1)$. Moreover, $\bby_t=(y_{1,t},\cdots,y_{m,t})'\in\mbR^m$, where all $y_{j,t}$ are independent such that $\mbE[y_{j,t}]=0,\mbE[y_{j,t}^2]=1$ for $j=1,\cdots,m$ and $t=1,\cdots,T$ and $\kappa_4:=\sup_{j,t}\big|\mbE[y_{j,t}^4]-3\big|<\infty$.
\end{Ap}
\begin{remark}
    We are only interested in the case where $\Vert\bbmu\Vert_2\leq\mrO(1)$, since if $\lim_{n\to\infty}\Vert\bbmu\Vert_2=\infty$, the hypothesis test \eqref{Eq of equivalence test} is trivial. Moreover, for moment conditions on $y_{j,t}$ in Assumption \ref{Ap of noise}, they coincide with those in some well-known reference like \cite{bai1996effect,chen2010two}.
\end{remark}
\begin{thm}\label{Thm of strategic CLT}
    Under Assumptions \ref{Ap of high dimensionality} and \ref{Ap of noise}, let 
    \begin{align}
    	\left\{\begin{array}{l}
    		\tau:=\Vert\bbmu\Vert_2^2-d_0^2,\\
    		\sigma^2:=\bbmu'\bbSig\bbmu+\frac{1}{T_1}\tr(\bbSig^2),
    	\end{array}\right.\quad \kappa_{T_2}:=-\tau\sqrt{\frac{T_2}{\tau^2+\sigma^2}},\quad\beta_{T_2}:=\frac{T_2\bbmu'\bbSig\bbmu}{T_1(\tau^2+\sigma^2)},\label{Eq of kappa 1}
    \end{align}
    if $\beta_{T_2}<1/2$, we have as $T_2\to\infty$
    \begin{align*}
        \mcM_{T_2}\overset{d}{\sim}\mcD\mcB(\kappa_{T_2},\beta_{T_2}).
    \end{align*}
\end{thm}
Readers can refer to the proof of Theorem \ref{Thm of strategic CLT} in Appendix \S\ref{sec of proof norm test}.
\begin{remark}
	\begin{enumerate}
		\item The requirement of $\beta_{T_2}<1/2$ in Theorem \ref{Thm of strategic CLT} ensures that $\mcD\mcB(\kappa_{T_2},\beta_{T_2})$ in \eqref{Eq of kappa 1} is well defined. A trivial assumption that $\lim_{T\to\infty}T_2/T_1\in(0,1/2)$ can always guarantee it.
		\item $\hat{\tau}_{T_2}$ and $\hat{\sigma}_{T_2}$ in \eqref{Eq of estimators} are consist estimators of $\tau$ and $\sigma$ in \eqref{Eq of kappa 1}, respectively. Moreover, for estimation of $\beta_{T_2}$ in \eqref{Eq of kappa 1}, it suffices to estimate $\bbmu'\bbSig\bbmu$. Let's divide $\{\bbx_1,\cdots,\bbx_T\}$ into three disjoint equal-size subsamples $\mcI_0,\mcI_1$ and $\mcI_2$, then
		\begin{align}
			\widehat{\bbmu'\bbSig\bbmu}:=\hat{\bbmu}_1'\widehat{\bbS}_0\hat{\bbmu}_2,\quad\left\{\begin{array}{l}
				\hat{\bbmu}_1:=\frac{1}{|\mcI_1|}\sum_{t\in\mcI_1}\bbx_t, \\
				\hat{\bbmu}_2:=\frac{1}{|\mcI_2|}\sum_{t\in\mcI_2}\bbx_t,
			\end{array}\right.\label{Eq of estimator}
		\end{align}
		where $\widehat{\bbS}_0$ is the sample covariance matrix of $\{\bbx_t:t\in\mcI_0\}$.
		\item According to \cite{chen2023strategic}, the asymptotic distribution of $\hat{b}_{T,t}$ for the univariate equivalence test \eqref{Eq of equivalence test one dim} is the bandit distribution $\mcB(\kappa)$, which is the distribution of $Z_1$ following the SDE as below:
		\begin{align*}
			dZ_s=\kappa\sign(Z_s)ds+dB_s,\quad Z_0=0,\ s\in[0,1].
		\end{align*}
		The above SDE has been comprehensively studied in \cite{mel1979strong}. 
		
		As we mentioned in Introduction, some properties of \eqref{Eq of U bandit SDE}, including the existence of unique weak solution and tail probabilities estimations, are essential to establish the asymptotic behavior of $\fM_{T_2}$. 
		
		Compared with \cite{chen2023strategic}'s result, our statistic $\fM_{T_2}$ depends on a more complicate SDE \eqref{Eq of U bandit SDE} involving a more complex diffusion $(1-2\beta\sign(Y_s)V_s)^{1/2}$, and $V_s$ also depends on $Y_s$. Moreover, both the drift and diffusion of the U-bandit SDE \eqref{Eq of U bandit SDE} are non-continuous, these difficulties make theoretical analysis for these properties challenging. Readers can refer to Appendix \S\ref{App of SDE} for more details.

	\end{enumerate}
\end{remark}

\subsection{Size and power}\label{Sec of size and power}
Given a significance level $\alpha\in(0,1)$, let $z_{\alpha/2}$ be the upper $\alpha/2$ quantile of $\mcD\mcB(0,\beta_{T_2})$, the rejection region of our hypothesis test \eqref{Eq of high dimensional mean test} is
\begin{align*}
    \big|\mcM_{T_2}\big|>z_{\alpha/2}.
\end{align*}
\begin{remark}
	Recall the definition of $\beta_{T_2}$ in \eqref{Eq of kappa 1}, we can estimate $\beta_{T_2}$ via $\hat{\tau}_{T_2},\hat{\sigma}_{T_2}$ and $\widehat{\bbmu'\bbSig\bbmu}$, see \eqref{Eq of estimators} and \eqref{Eq of estimator}. Once we obtain $\hat{\beta}_{1,T_2}$, we can further estimate $z_{\alpha/2}$ via Monte Carlo simulations.
\end{remark}
For the size and power of our proposed statistic $\mcM_{T_2}$, Theorem \ref{Thm of strategic CLT} implies that 
\begin{align*}
    \lim_{T\to\infty}\big|\mbP\big(|\mcM_{T_2}|>z_{\alpha/2}\big)-\mbP(|\mcD\mcB(\kappa_{T_2},\beta_{T_2})|> z_{\alpha/2})\big|=0,
\end{align*}
so it suffices to analyze the tail probabilities of $\mcD\mcB(\kappa_{T_2},\beta_{T_2})$. Intuitively, according to Figure \ref{Fig of spiked distribution}, $\mbP(|\mcD\mcB(\kappa,\beta)|> z_{\alpha/2})$ is increasing in $\kappa\in\mbR$. Here, we provide that
\begin{pro}\label{Pro of tail probability}
	For any $\kappa\in\mbR$ and $|\beta|<1/2$, we have
	\begin{align*}
		\left\{\begin{array}{ll}
			\mbP(|\mcD\mcB(\kappa,\beta)|>z_{\alpha/2})\leq\mbP(|\mcD\mcB(0,\beta)|>z_{\alpha/2}), & \kappa\leq0, \\
			\mbP(|\mcD\mcB(\kappa,\beta)|>z_{\alpha/2})\geq\mbP(|\mcD\mcB(0,\beta)|>z_{\alpha/2}), & \kappa>0.
		\end{array}\right.
	\end{align*}
	Moreover, when $\kappa>0$ and $|\kappa|$ is large enough, we have
	\begin{align*}
		1-\mbP(|\mcD\mcB(\kappa,\beta)|>z_{\alpha/2})\leq\mrO(e^{-C\kappa^2})
	\end{align*}
	for some constant $C>0$.
\end{pro}
The prove of Proposition \ref{Pro of tail probability} is presented in \S\ref{sec of tail probability}.

%
\medskip
\noindent
{\bf Size of the test: }Under $H_0$, since $\tau=\Vert\bbmu\Vert_2^2-d_0^2>0$, then $\kappa_{T_2}\leq0$ by \eqref{Eq of kappa 1} and
\begin{align*}
    \mbP(|\mcD\mcB(\kappa_{T_2},\beta_{T_2})|>z_{\alpha/2})\leq \mbP(|\mcD\mcB(0,\beta_{T_2})|>z_{\alpha/2})=\alpha,
\end{align*}
so the asymptotic size of $\mcM_{T_2}$ satisfies 
\begin{align*}
    &\lim_{T\to\infty}\mbP\big(\big|\mcM_{T_2}\big|>z_{\alpha/2}\big|H_0\big)\leq\alpha.
\end{align*} 
{\bf Power of the test: }Under $H_1$, we have $\tau=\Vert\bbmu\Vert_2^2-d_0^2\leq0$. Note that $\tau=\tau(T)$ is a function of the sample size $T$. Let's analyze the asymptotic power of the test in the following three cases:
\begin{itemize}
    \item \emph{Case 1}. $\liminf_{T\to\infty}\tau^2T=\infty$. Since $\sigma^2=\bbmu'\bbSig\bbmu+\frac{1}{T_1}\tr(\bbSig^2)=\mrO(1)$ by Assumption \ref{Ap of noise}, then
    $\liminf_{T\to\infty}\kappa_{T_2}=\infty$ by \eqref{Eq of kappa 1}. Using Proposition \ref{Pro of tail probability}, the asymptotic power of $\mcM_{T_2}$ is full, that is
    \begin{align*}
        \lim_{T\to\infty}\mbP\big(\big|\mcM_{T_2}\big|>z_{\alpha/2}\big|H_1\big)=1.
    \end{align*}
    \item \emph{Case 2}. $\liminf_{T\to\infty}\tau^2T\in(0,\infty)$. By \eqref{Eq of kappa 1}, there exists a constant $C_0>0$ such that
    $$\liminf_{T\to\infty}\kappa_{T_2}>C_0,$$
    so the asymptotic power of $\mcM_{T_2}$ satisfies
    \begin{align*}
        \limsup_{T\to\infty}\mbP\big(\big|\mcM_{T_2}\big|>z_{\alpha/2}\big|H_1\big)<1.
    \end{align*}
    \item \emph{Case 3}. If $\liminf_{T\to\infty}\tau^2T=0$, we have $\liminf_{T\to\infty}\kappa_{T_2}=0$ and the asymptotic power of $\mcM_{T_2}$ satisfies
    \begin{align*}
        \lim_{T\to\infty}\mbP\big(\big|\mcM_{T_2}\big|>z_{\alpha/2}\big|H_1\big)=\alpha.
    \end{align*}
\end{itemize}
In summary, the power of $\mcM_{T_2}$ is asymptotically 1 when $\liminf_{T\to\infty}\tau^2T=\infty$, and otherwise it lies in the interval $(\alpha,1)$. Readers can refer to Figure \ref{Fig of Power} for an illustration of the asymptotic the power of $\mcM_{T_2}$ in these three cases.
\begin{figure}[ht]
    \centering
    \includegraphics[width=0.9\linewidth]{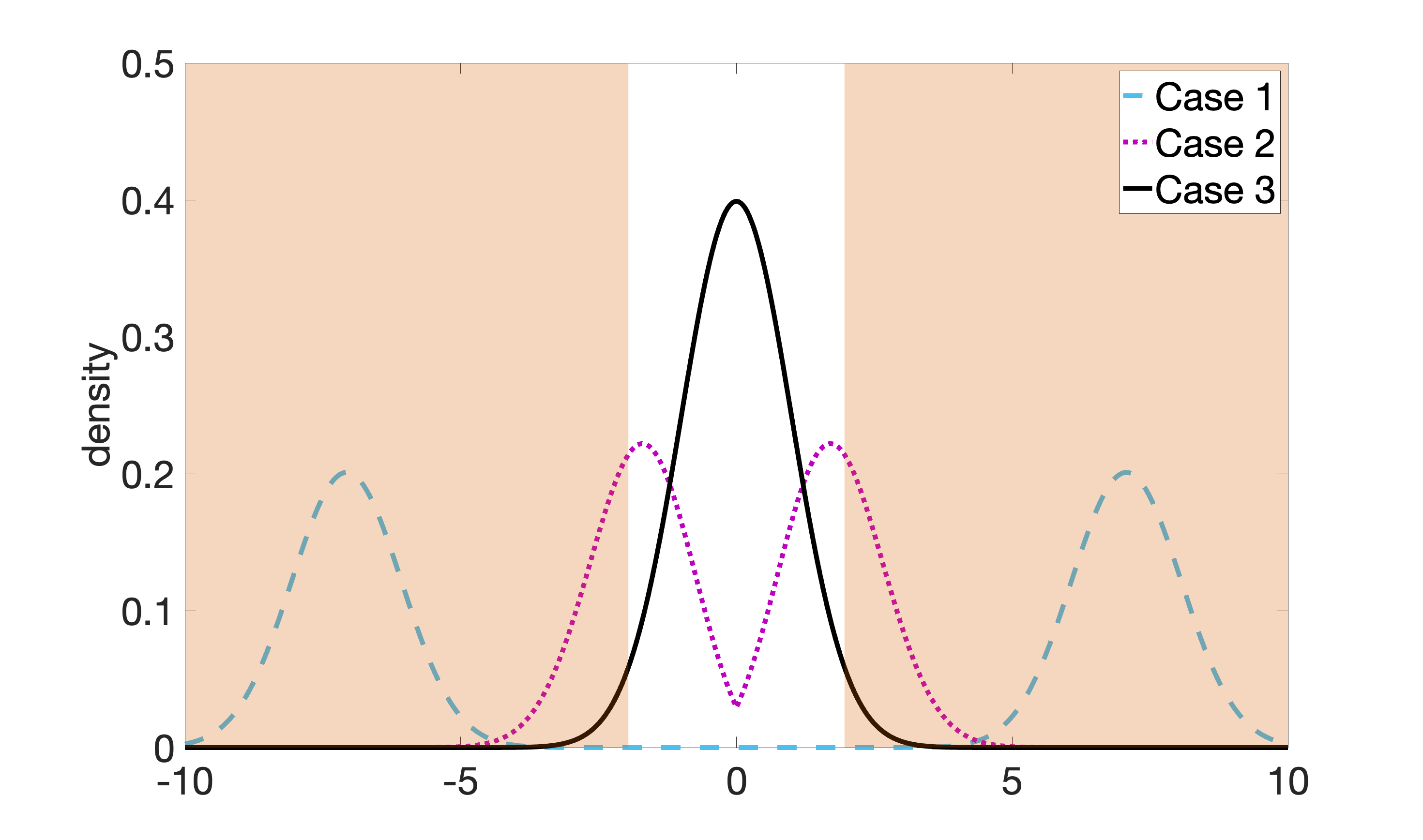}
    \caption{Density plot of $\mcM_{T_2}\overset{d}{\sim}\mcD\mcB(\kappa_{T_2},\beta_{T_2})$ under $H_1$, where the red area is the rejection region $|\mcM_{T_2}|>z_{\alpha/2}$. In case 1 as $|\kappa_{T_2}|\to\infty$, the blue dashed line suggests the densities of $\mcM_{T_2}$ almost lie in the rejection region, so the asymptotic power is full. In case 2 as $|\kappa_{T_2}|\to\mrO(1)$, the purple dotted line shows that only part of $\mcM_{T_2}$'s distributions lies in the rejection region, so the asymptotic power is in $(\alpha,1)$. In case 3 as $|\kappa_{T_2}|\to0$, $\mcM_{T_2}$ is asymptotically normal (continuous curve) and the asymptotic power equals $\alpha$.}
    \label{Fig of Power}
\end{figure}

\section{Two-sample case}\label{Sec of two sample}
In this section, we extend the proposed one-sample test in \S\ref{sec of norm test} to the two-sample case, since the basic frameworks are nearly the same, we only highlight the key differences and omit repetitive details for simplicity. Consider two independent $n$-dimensional datasets $\{\bbx_{j_1}:j_1=1,\cdots,M_1\}$ and $\{\bbz_{j_2}:j_2=1,\cdots,M_2\}$ such that all $\{\bbx_{j_1},\bbz_{j_2}\}$ are independent with $\mbE[\bbx_{j_1}]=\bbmu_1,\mbE[\bbz_{j_2}]=\bbmu_2$ and
$$\mbE[(\bbx_{j_1}-\bbmu_1)(\bbx_{j_1}-\bbmu_1)']=\bbSig_1,\quad\mbE[(\bbz_{j_2}-\bbmu_2)(\bbz_{j_2}-\bbmu_2)']=\bbSig_2$$
for $j_1=1,\cdots,M_1$ and $j_2=1,\cdots,M_2$. Here, the sample sizes $M_1,M_2$ satisfy that
\begin{align}
    \lim_{M_1\to\infty}M_1/M_2=c_2\in(0,\infty).\label{Eq of two sample size}
\end{align}
Similar to the previous one-sample case, we split $\{\bbx_{j_1}:j_1=1,\cdots,M_1\}$ and $\{\bbz_{j_2}:j_2=1,\cdots,M_2\}$ into two subsamples, respectively. Precisely, let $N_0:=N_0(M_1)\in\mbN^+$ such that $N_0<\min\{M_1,M_2\},\lim_{M_1\to\infty}N_0=\infty$ and
\begin{align}
    \lim_{M_1\to\infty}(M_1-N_0)/M_1=r_1\in(0,1)\quad\text{and}\quad\lim_{M_2\to\infty}(M_2-N_0)/M_2=r_2\in(0,1).\label{Eq of size split data two sample}
\end{align}
For any constant $d_0>0$, consider testing the hypotheses 
$$H_{0,\delta}:\Vert\bbmu_1-\bbmu_2\Vert_2> d_0\quad\text{versus}\quad H_{1,\delta}:\Vert\bbmu_1-\bbmu_2\Vert_2\leq d_0.$$
And we specify the following high-dimensional regime:
\begin{align}
    \lim_{M_1\to\infty}n/M_1=c_1\in[0,\infty).\label{Eq of high dimensionality two sample}
\end{align}
First, let $m_i:=M_i-N_0$ for $i=1,2$ and
$$\Delta_0:=\frac{1}{m_1}\sum_{j=1}^{m_1}\bbx_j-\frac{1}{m_2}\sum_{j=1}^{m_2}\bbz_j,$$
and
\begin{align*}
    \left\{\begin{array}{l}
        \displaystyle\Delta_i:=\bbx_{m_1+i}-\bbz_{m_2+i},\quad\text{and}\quad Y_i:=\Delta_0'\Delta_i-d_0^2\quad\text{for}\quad i=1,\cdots,N_0,\\[6pt]
        \displaystyle\hat{\tau}_{\delta}:=\frac{1}{N_0}\sum_{i=1}^{N_0}Y_i\quad\text{and}\quad\hat{\sigma}_{\delta}^2:=\frac{1}{N_0}\sum_{i=1}^{N_0}Y_i^2-\hat{\tau}_{\delta}^2.
    \end{array}\right.
\end{align*}
Next, define
\begin{align}
    \mcM_j^{\delta}=\frac{1}{\sqrt{N_0}}\sum_{i=1}^j\frac{\theta_iY_i}{\sqrt{\hat{\tau}_{\delta}^2+\hat{\sigma}_{\delta}^2}},\quad j=1,\cdots,N_0,\label{Eq of statistic two sample norm test}
\end{align}
where $\theta_1=1$ and for $j=2,\cdots,N_0$
\begin{align*}
    \theta_j=\left\{\begin{array}{cl}
        1, & \text{if }\ \mcM_{j-1}^{\delta}\leq0, \\[6pt]
        -1, & \text{if }\ \mcM_{j-1}^{\delta}>0,
    \end{array}\right.
\end{align*}
The following data generation process parallels that of one-sample case in Assumption \ref{Ap of noise}.
\begin{Ap}\label{Ap of noise two sample}
    Let $\bbx_{j_1}=\bbmu_1+\bbGa_1\bby_{j_1}$ and $\bbz_{j_2}=\bbmu_2+\bbGa_2\bbw_{j_2}$ such that for $i=1,2$, $\bbGa_i\bbGa_i'=\bbSig_i$ and $\Vert\bbmu_i\Vert_2,\Vert\bbSig_i\Vert\leq\mrO(1)$ and $\liminf_{M_i\to\infty}\frac{1}{M_i}\tr(\bbSig_i)>0$. Moreover, $\bby_{j_1}=(y_{1,j_1},\cdots,y_{m_1,j_1})'\in\mbR^{m_1}$, where all $y_{i,j_1}$ are independent random variables such that $\mbE[y_{i,j_1}]=0,\mbE[y_{i,j_1}^2]=1$ and $\sup_{i,j_1}\mbE[y_{i,j_1}^4]<\infty$ for $i=1,\cdots,m_1$ and $j_1=1,\cdots,T$. Similar conditions are also assumed on the vectors $\bbw_{j_2}$. 
\end{Ap}
\begin{thm}\label{Thm of two sample norm test}
    Under \eqref{Eq of two sample size}, \eqref{Eq of size split data two sample}, \eqref{Eq of high dimensionality two sample} and Assumption \ref{Ap of noise two sample}, let $\bbmu_{\delta}:=\bbmu_1-\bbmu_2$ and
    $$\left\{\begin{array}{l}
    	\tau_{\delta}=\Vert\bbmu_{\delta}\Vert^2-d_0^2,\\
    	\sigma_{\delta}^2=\bbmu_{\delta}'(\bbSig_1+\bbSig_2)\bbmu_{\delta}+\tr(\{\bbSig_1+\bbSig_2\}\{\bbSig_1/m_1+\bbSig_2/m_2\}),
    \end{array}\right.$$
    and
    \begin{align*}
    	\kappa_{N_0}^{\delta}:=-\tau_{\delta}\sqrt{\frac{N_0}{\tau_{\delta}^2+\sigma_{\delta}^2}}\quad\beta_{N_0}^{\delta}:=\frac{N_0(m_1^{-1}\bbmu_{\delta}'\bbSig_1\bbmu_{\delta}+m_2^{-1}\bbmu_{\delta}'\bbSig_2\bbmu_{\delta})}{\tau_{\delta}^2+\sigma_{\delta}^2}.
    \end{align*}
    If $\beta_{N_0}<1/2$, as $N_0\to\infty$,
    $$\mcM_{N_0}^{\delta}\overset{d}{\sim}\mcD\mcB(\kappa_{N_0}^{\delta},\beta_{N_0}^{\delta}).$$
\end{thm}
The proof of Theorem \ref{Thm of two sample norm test} is the same as the proof of Theorem \ref{Thm of strategic CLT}, thus omitted. Finally, given a significance level $\alpha\in(0,1)$, let $z_{\alpha/2}$ be the upper $\alpha/2$-quantile of $\mcD\mcB(0,\beta_{N_0}^{\delta})$, we reject $H_{0,\delta}$ if  
\begin{align*}
    \big|\mcM_{N_0,N_0}^{\delta}\big|>z_{\alpha/2}.
\end{align*}
Asymptotic results of this test parallel previous ones in the one-sample case. In particular, and the asymptotic size of $\mcM_{N_0,N_0}^{\delta}$ is upper-bounded by $\alpha$, and under $H_{1,\delta}$, if we suppose $\liminf_{N_0\to\infty}\tau_{\delta}^2N_0=\infty$, the asymptotic power of $\mcM_{N_0,N_0}^{\delta}$ is 1.
\begin{remark}
	In practice, we first divide the whole data $\{\bbx_1,\cdots,\bbx_{M_1}\}$ into three equal-size subsamples $\mcI_0,\mcI_1,\mcI_2$, so do $\mcJ_0,\mcJ_1,\mcJ_2$ for $\{\bbz_1,\cdots,\bbz_{M_2}\}$, then consider the following estimators:
	\begin{align*}
		&\left\{\begin{array}{l}
			\widehat{\bbmu_{\delta}'\bbSig_1\bbmu_{\delta}}:=\hat{\bbmu}_{1,\delta}'\widehat{\bbS}_{x,0}\hat{\bbmu}_{2,\delta},\\[6pt]
			\widehat{\bbmu_{\delta}'\bbSig_2\bbmu_{\delta}}:=\hat{\bbmu}_{1,\delta}'\widehat{\bbS}_{z,0}\hat{\bbmu}_{2,\delta},
		\end{array}\right.\quad\left\{\begin{array}{l}
			\hat{\bbmu}_{1,\delta}:=\frac{1}{|\mcI_1|}\sum_{i\in\mcI_1}\bbx_i-\frac{1}{|\mcJ_1|}\sum_{j\in\mcJ_1}\bbz_j, \\[6pt]
			\hat{\bbmu}_{2,\delta}:=\frac{1}{|\mcI_2|}\sum_{i\in\mcI_2}\bbx_i-\frac{1}{|\mcJ_2|}\sum_{j\in\mcJ_2}\bbz_j,
		\end{array}\right.
	\end{align*}
	where $\widehat{\bbS}_{x,0}$ is the sample covariance matrix of $\{\bbx_i:i\in\mcI_0\}$, so does $\widehat{\bbS}_{z,0}$. After deriving estimators of $\bbmu_{\delta}'\bbSig_1\bbmu_{\delta}$ and $\bbmu_{\delta}'\bbSig_2\bbmu_{\delta}$, we have
	\begin{align*}
		\hat{\beta}_{N_0}^{\delta}:=\frac{N_0(m_1^{-1}\widehat{\bbmu_{\delta}'\bbSig_1\bbmu_{\delta}}+m_2^{-1}\widehat{\bbmu_{\delta}'\bbSig_2\bbmu_{\delta}})}{\hat{\tau}_{\delta}^2+\hat{\sigma}_{\delta}^2},
	\end{align*}
	and the upper $\alpha/2$-quantile of $\mcD\mcB(0,\beta_{N_0}^{\delta})$ can be estimated via the Monte Carlo method.
\end{remark}

\section{Numerical experiments}\label{Sec of numerical experiment}
We conduct two numerical experiments to verify the finite-sample performance of our norm test \eqref{Eq of high dimensional mean test} in the one sample case, respectively. Results for the two-sample case are highly expected to be similar thus not considered.

Here, we set the significance level $\alpha=0.05$ and $\{\bbx_t\overset{i.i.d.}{\sim}\mcN(\bbmu,\bbSig):t=1,\cdots,T\}$, where $\bbmu=(n^{-1/2},\cdots,n^{-1/2})'\in\mbR^n$ is the population mean vector with $\Vert\bbmu\Vert_2=1$ and $\bbSig=[0.5^{-|i-j|}]_{i,j}\in\mbR^{n\times n}$ be the population covariance matrix. Moreover, we split the data such that $\frac{T_1}{T_2}=\frac{3}{2}$, which ensures that $\beta_{T_2}<1/2$ in \eqref{Eq of kappa 1}. Under different $(n,T)$ and $d_0=0.5+0.1l$ for $l=0,1,\cdots,10$, we compute the empirical values of $\mbP\big(|\mcM_{T_2}|> z_{\alpha/2}\big)$ based on 500 independent samples, all results are summarized in Table
\begin{figure}[ht]
	\subfigure[$\frac{T}{n}=0.8$.]{\includegraphics[width=0.49\linewidth]{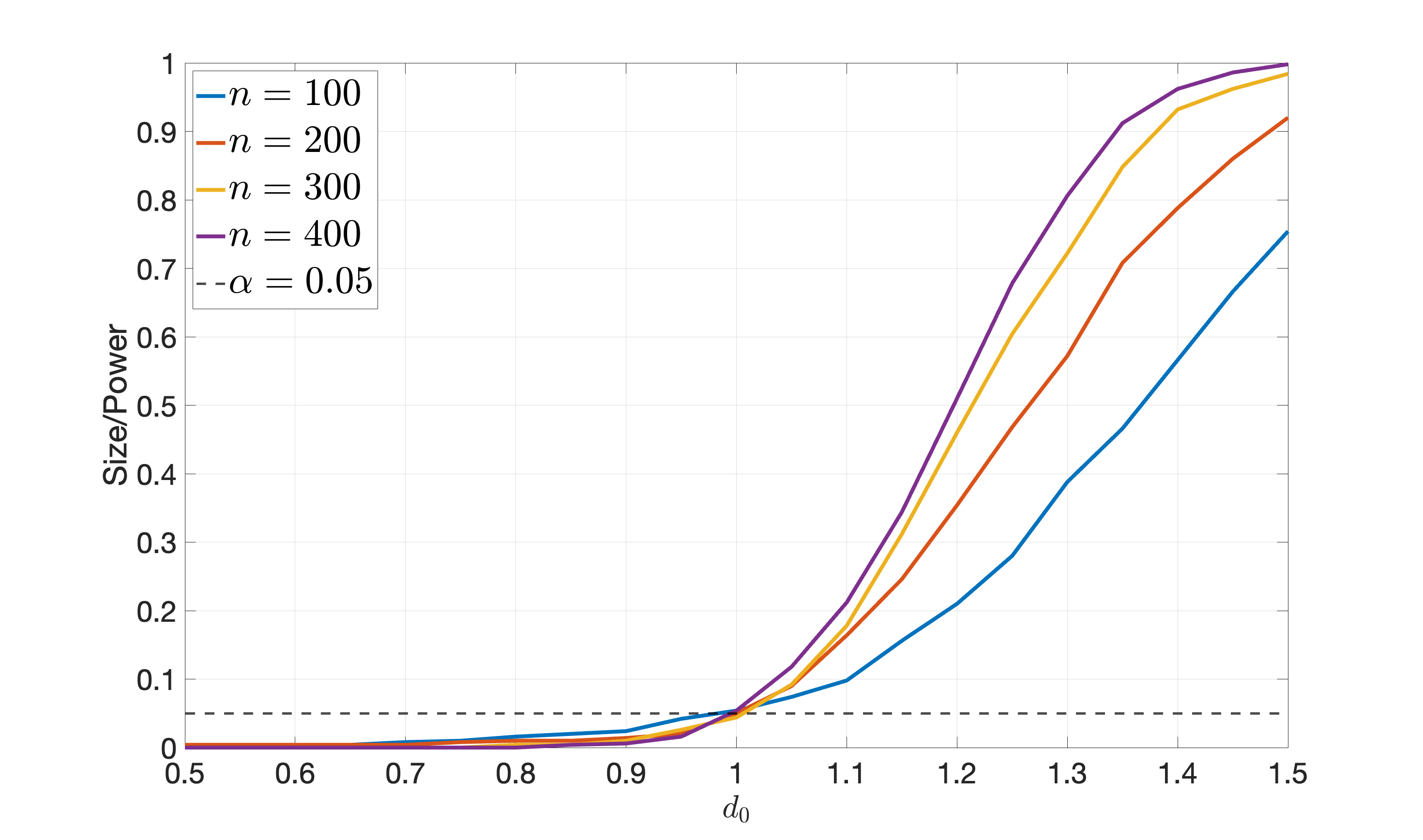}}
	\subfigure[$\frac{T}{n}=1.0$.]{\includegraphics[width=0.49\linewidth]{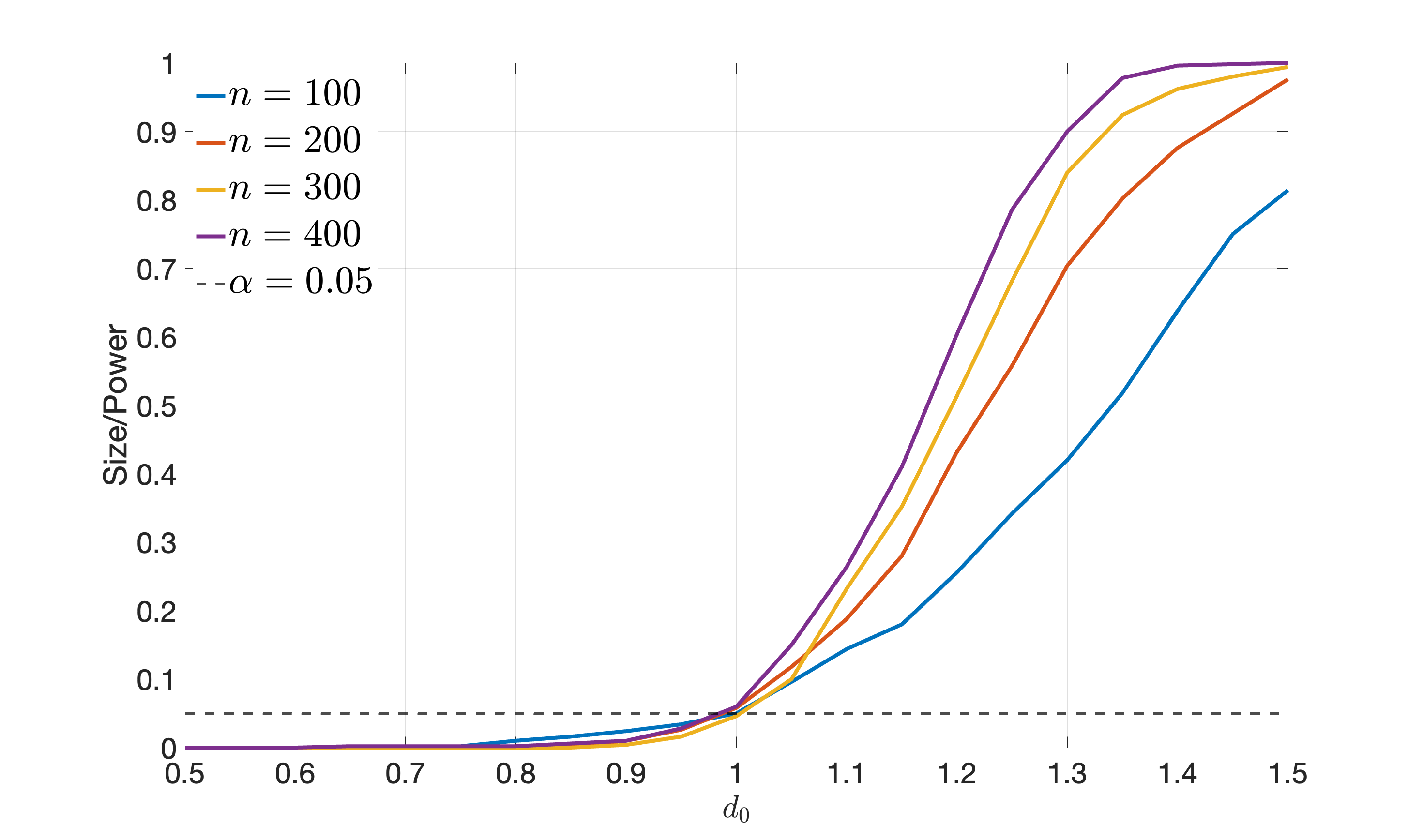}}
	\subfigure[$\frac{T}{n}=1.2$.]{\includegraphics[width=0.49\linewidth]{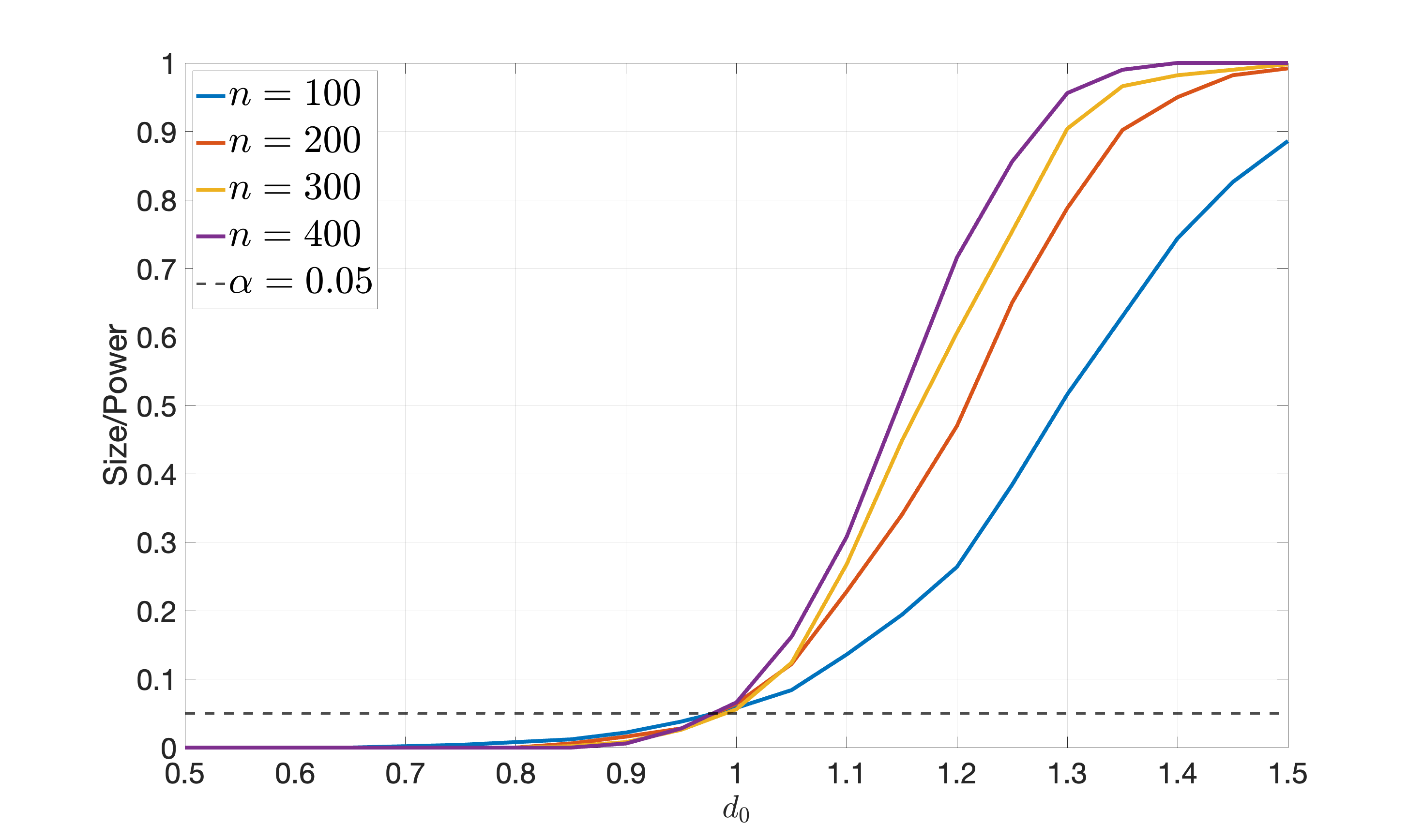}}
	\subfigure[$\frac{T}{n}=1.5$.]{\includegraphics[width=0.49\linewidth]{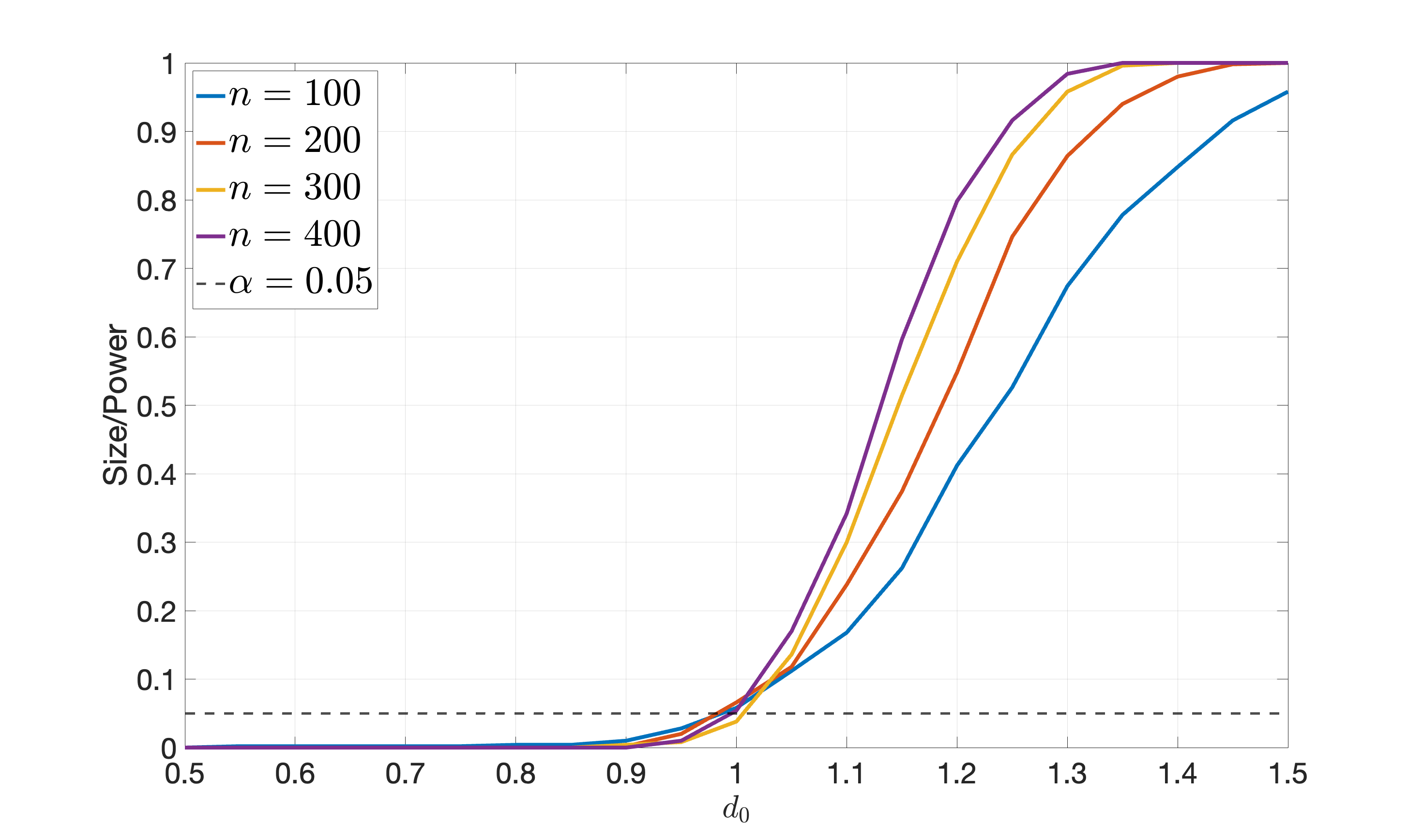}}
	\caption{Empirical values of $\mbP(|\fM_{T_2}|>z_{\alpha/2})$ under significance level of $\alpha=0.05$. The null hypothesis is ``$d_0<1$'' and the alternative hypothesis is ``$d_0\geq1$''}
\end{figure}

Recall our discussions about the asymptotic size and power of $\mcM_{T_2}$ in the end of \S\ref{Sec of one sample}, since 
\begin{align*}
    \lim_{T\to\infty}\big|\mbP(|\mcM_{T_2}|>z_{\alpha/2})-\mbP(|\mcD\mcB(\kappa_{T_2},\beta_{T_2})|>z_{\alpha/2})\big|=0,
\end{align*}
when $d_0>\Vert\bbmu\Vert_2=1$, we have $\kappa_{T_2}=\mrO(\sqrt{T})$, so the empirical values of $\mbP\big(|\mcM_{T_2}|> z_{\alpha/2}\big)$ increasingly tend to 1 as $(n,T)$ tend to infinity proportionally. On the other hand, when $d_0<\Vert\bbmu\Vert_2=1$, then $\kappa_{T_2}=\mrO(-\sqrt{T})$ and the empirical values of $\mbP\big(|\mcM_{T_2}|> z_{\alpha/2}\big)$ is decreasing and asymptotically smaller than the significance level $\alpha=0.05$ as $(n,T)$ tend to infinity proportionally. These results well confirm the theoretical findings for the equivalence test in \S\ref{sec of asymptotic distribution}.

\section{Empirical study}\label{Sec of empirical}
We exam an intestinal microbiota dataset in \cite{lahti2014tipping}. This dataset includes 1006 samples and each sample contains 130 genus-like phylogenetic groups that cover the majority of the known bacterial
diversity of the human intestine. Based on this dataset, \cite{zhang2020distance} and \cite{jiang2025testing} have investigated whether the population mean of microbiome compositions differs between the younger age group and the elder age group. Both of these two papers conclude that this difference is significant but provide no information on the magnitude of this difference. Clearly, it is important to quantify this difference of microbiome compositions between the younger group and the elder group. Thus, we reanalyze this dataset based on our two-sample equivalence tests developed in \S\ref{Sec of two sample}.

Here, we select individuals aged 35 to 50 as the elder group (313 samples in total) and those under 35 as
the younger group (303 samples in total), and consider the logarithm of the raw data. Let $\bbmu_1$ and $\bbmu_2$ be the population mean vectors of the younger and the elder group, respectively. We then consider the following test: 
\begin{align*}
    H_0:\Vert\bbmu_1-\bbmu_2\Vert_2> d_0\quad\text{versus}\quad H_1:\Vert\bbmu_1-\bbmu_2\Vert_2\leq d_0.
\end{align*}
Let $d_0=1.4+0.02i$ for $i=0,1,\cdots,10$, we apply our two-sample equivalence test in \S\ref{Sec of two sample}, choosing $N_0=100$, then compute the statistic $|\mcM_{N_0}^{\delta}|$ in \eqref{Eq of statistic two sample norm test} under different values of $d_0$. The obtained values of the test statistic and corresponding p-values are given in Table \ref{Tab of empirical study norm}.
\begin{table}[ht]
    \centering
    \caption{Values of $|\mcM_{N_0}^{\delta}|$ under different $d_0$.}
    \label{Tab of empirical study norm}
    \renewcommand{\arraystretch}{1.2}
    \begin{tabular}{|c|ccccccccccc|}
        \hline
        $d_0$ & 1.90 & 1.92 & 1.94 & 1.96 & 1.98 & 2.00 & 2.02 & 2.04 & 2.06 & 2.08 & 2.10 \\\hline
        $\big|\mcM_{N_0}^{\delta}\big|$ & 2.00 &  2.02 &   2.08 &  2.15 &   2.19  &  2.21 &  2.25  &  2.29  &  2.36 & 2.44  &  2.46 \\\hline
    \end{tabular}
\end{table}
If we set the significance level $\alpha=0.05$, then the critical value is $z_{\alpha/2}=2.21$ via the Monte Carlo simulations. According to Table \ref{Tab of empirical study norm}, when $d_0\geq2.02$, we will reject $\Vert\bbmu_1-\bbmu_2\Vert_2>d_0$. Following the theory developed in \S\ref{Sec of two sample}, we conclude that $\Vert\bbmu_1-\bbmu_2\Vert_2$ is nearly 2.02 at risk level of 0.05.

In summary, according to results in Table \ref{Tab of empirical study norm}, under the significance level $\alpha=0.05$, our test suggests the deviation of microbiome compositions between these two groups is smaller than 2.02. From a practical perspective, compared with previous empirical studies \cite{zhang2020distance} and \cite{jiang2025testing}, our methods indeed provide concretely quantitative characteristic about the bacterial  difference between these two groups, and such quantitative information can provide more insights to further assist the biological research.

\appendix

\section{Appendix \texorpdfstring{\RomanNumeralCaps{1}}{I}}\label{App}
The purpose of \S\ref{App} is to provide necessary auxiliary lemmas for proving Theorem \ref{Thm of strategic CLT}. Precisely,
\begin{itemize}
    \item In \S\ref{sec of Auxiliary lemmas}, for variables $\hat{\tau}_{T_2},\hat{\sigma}_{T_2}$ and $\widehat{\bbmu'\bbSig\bbmu}$ defined in \eqref{Eq of estimators} and \eqref{Eq of estimator}, we show that they are consistent estimators of $\tau,\sigma$ and $\bbmu'\bbSig\bbmu$, respectively.
    \item In \S\ref{subsec of existence and uniqueness of the weak solution}, we prove that the SDE \eqref{Eq of U bandit SDE} admits a unique weak solution. This is achieved by using the Girsanov's theorem to transform the equation into a driftless system, followed by an application of the Dambis-Dubins-Schwarz theorem to establish pathwise uniqueness.
    \item In \S\ref{sec of tail probability}, we derive explicit exponential tail decay bounds for the weak solution. It establishes that the distribution is more concentrated around zero when $\kappa \le 0$ and stochastically more dispersed when $\kappa > 0$. 
    \item In \S\ref{sec of bounded derivatives}, to address the non-differentiability caused by the $\text{sign}(\cdot)$ function in \eqref{Eq of U bandit SDE}, we introduce a smoothed SDE approximation. We prove that the weak solution of this smoothed SDE converges to the original \eqref{Eq of U bandit SDE}. This smoothing step enables the rigorous derivation of bounded upper limits for the spatial derivatives of the corresponding Kolmogorov backward equation.
\end{itemize}

\subsection{Auxiliary lemmas}\label{sec of Auxiliary lemmas}
\begin{lem}\label{Lem of Consistence}
    Under Assumptions \ref{Ap of high dimensionality} and \ref{Ap of noise}, we have
    \begin{align}
        &\hat{\tau}_{T_2}\overset{\mbP}{\longrightarrow}\tau,\quad\hat{\sigma}_{T_2}^2\overset{\mbP}{\longrightarrow}\sigma^2\quad\text{and}\quad\widehat{\bbmu'\bbSig\bbmu}\overset{\mbP}{\longrightarrow}\bbmu'\bbSig\bbmu.\notag
    \end{align}
\end{lem}
\begin{proof}
    Let's first show that $\mbE[(\hat{\tau}_{T_2}-\tau)^4]\leq\mrO(T^{-2})$. Note that
    \begin{align*}
        X_t-\tau=\underbrace{\bbmu'\bbGa\bby_{t+T_1}+\frac{1}{T_1}\sum_{s=1}^{T_1}\bby_s'\bbGa'\bbGa\bby_{t+T_1}}_{:=A_t}+\underbrace{\frac{1}{T_1}\sum_{s=1}^{T_1}\bbmu'\bbGa\bby_s}_{:=\mfv},
    \end{align*}
    so $\hat{\tau}_{T_2}-\tau=\frac{1}{T_2}\sum_{t=1}^{T_2}A_t+\mfv$. By direct calculations and Assumption \ref{Ap of noise}, we have $\mbE[\mfv^4]\leq\mrO(T^{-2}\bbmu'\bbSig\bbmu+\kappa_4T^{-3}\sum_{j=1}^m\{\bbmu\Gamma_{\cdot j}\}^4)\leq\mrO(1)$. By the Minkowski's inequality, we obtain that
    \begin{align*}
        \mbE[(\hat{\tau}_{T_2}-\tau)^4]\leq8\mbE[\mfv^4]+8T_2^{-4}\mbE\left[\left(\sum_{t=1}^{T_2}A_t\right)^4\right]\leq8T_2^{-4}\sum_{t_1,t_2=1}^T\mbE[A_{t_1}^2A_{t_2}^2]+\mrO(T^{-2}).
    \end{align*}
    Here, let's prove that $\mbE[A_t^4]\leq\mrO(1)$. For simplicity, let
    $$\fy_{T_1}:=\frac{1}{T_1}\sum_{s=1}^{T_1}\bby_s.$$
    By the Minkowski's inequality, we have
    $$\mbE[A_t^4]\leq8\mbE[(\bbmu'\bbGa\bby_{t+T_1})^4]+8\mbE[(\fy_{T_1}'\bbGa'\bbGa\bby_{t+T_1})^4].$$
    Since $\mbE[(\bbmu'\bbGa\bby_{t+T_1})^4]\leq\mrO(1)$ by Assumption \ref{Ap of noise}, it suffices to show that $\mbE[(\fy_{T_1}'\bbGa'\bbGa\bby_{t+T_1})^4]\leq\mrO(1)$. Denote $\bbB:=\bbGa'\bbGa\fy_{T_1}\fy_{T_1}'\bbGa'\bbGa$, since $\fy_{T_1}$ and $\bby_{t+T_1}$ are independent, then $\bbB$ and $\bby_{t+T_1}$ are independent, by direct calculations, we have
    \begin{align*}
        &\mbE[(\fy_{T_1}'\bbGa'\bbGa\bby_{t+T_1})^4]=\mbE[(\bby_{t+T_1}'\bbB\bby_{t+T_1})^2]\\
        &=\sum_{j_1\cdots j_4=1}^m\mbE[B_{j_1,j_2}B_{j_3,j_4}]\mbE[y_{j_1,t+T_1}y_{j_2,t+T_1}y_{j_3,t+T_1}y_{j_4,t+T_1}]\\
        &=\mbE[\tr(\bbB)^2]+2\mbE[\tr(\bbB^2)]+\sum_{j=1}^m\mbE[B_{j,j}^2](\mbE[y_{j;t+T_1}^4]-3),
    \end{align*}
    where
    \begin{align}
        \begin{array}{l}
            \displaystyle\mbE[\tr(\bbB)^2]=T_1^{-2}\tr(\bbSig^2)^2+\sum_{j=1}^m(\Gamma_{\cdot j}'\bbSig\Gamma_{\cdot j})^2(\mbE[\fy_{j,T_1}^4]-T_1^{-2}), \\
            \displaystyle\mbE[\tr(\bbB^2)]=T_1^{-2}\tr(\bbSig^2)^2+2T_1^{-2}\tr(\bbSig^4)+\sum_{j=1}^m(\Gamma_{\cdot j}'\bbSig\Gamma_{\cdot j})^2(\mbE[\fy_{j,T_1}^4]-3T_1^{-2}).
        \end{array}\label{Eq of Variance trick}
    \end{align}
    By Assumption \ref{Ap of noise}, $\mbE[y_{j;t+T_1}^4]-3\leq|\kappa_4|<\infty$, then $\big|\mbE[\fy_{j,T_1}^4]-3T_1^{-2}\big|\leq|\kappa_4|T_1^{-2}$. This implies that $\mbE[\tr(\bbB)^2],\mbE[\tr(\bbB^2)]\leq\mrO(1)$. Similarly, since $\mbE[B_{j,j}^2]=\mbE[(\Gamma_{\cdot j}'\bbGa\fy_{T_1})^4]\leq\mrO(T_1^{-2})$, then $\big|\sum_{j=1}^m\mbE[B_{j,j}^2](\mbE[y_{j;t+T_1}^4]-3)\big|\leq\mrO(T^{-1})$ and $\mbE[A_t^4]\leq\mrO(1)$. Consequently, we have
    \begin{align*}
        \mbE[(\hat{\tau}_{T_2}-\tau)^4]\leq8T_2^{-4}\sum_{t_1,t_2=1}^T\mbE[A_{t_1}^2A_{t_2}^2]+\mrO(T^{-2})\leq\mrO(T^{-2}).
    \end{align*}
    Next, let's prove $\mbE[(\hat{\sigma}_{T_2}^2-\sigma^2)^2]\leq\mrO(T^{-1})$. Note that
    \begin{align*}
        \hat{\sigma}_{T_2}^2-\sigma^2=\frac{1}{T_2}\sum_{t=1}^{T_2}(X_t-\tau)^2-\sigma^2+\tau^2-\hat{\tau}_{T_2}^2,
    \end{align*}
    where $\mbE[(\tau^2-\hat{\tau}_{T_2}^2)^2]\leq\mbE[(\hat{\tau}_{T_2}-\tau)^4]^{1/2}\mbE[(\hat{\tau}_{T_2}+\tau)^4]^{1/2}\leq\mrO(T^{-1})$, so it suffices to show that
    $$\frac{1}{T_2^2}\mbE\left[\left(\sum_{t=1}^{T_2}(X_t-\tau)^2-\sigma^2\right)^2\right]=\frac{1}{T_2^2}\sum_{t_1,t_2=1}^{T_2}\mbE[(X_{t_1}-\tau)^2(X_{t_2}-\tau)^2]-\sigma^4\leq\mrO(T^{-1/2}).$$
    Since $(X_t-\tau)^2=A_t^2+2A_t\mfv+\mfv^2$ and $\mbE[A_t^4]\leq\mrO(1)$, for $t_1\neq t_2$,
    \begin{align*}
        \mbE[A_{t_1}\mfv^3]=\mbE[A_{t_2}\mfv^3]=\mbE[A_{t_1}^2A_{t_2}\mfv]=\mbE[A_{t_1}A_{t_2}^2\mfv]=\mbE[A_{t_1}A_{t_2}\mfv^2]=0.
    \end{align*}
    By Cauchy–Schwarz inequality, we have $\mbE[A_{t_1}^2\mfv^2]\leq\mbE[A_{t_1}^2]^{1/2}\mbE[\mfv^2]^{1/2}\leq\mrO(T^{-1})$. Similarly, $\mbE[A_{t_2}^2\mfv^2]\leq\mrO(T^{-1})$. Therefore,
    \begin{align*}
        \frac{1}{T_2^2}\sum_{t_1\neq t_2}^{T_2}\mbE[(X_{t_1}-\tau)^2(X_{t_2}-\tau)^2]=\frac{1}{T_2^2}\sum_{t_1\neq t_2}^{T_2}\mbE[A_{t_1}^2A_{t_2}^2]+\mrO(T^{-1}).
    \end{align*}
    Moreover, by direct calculations, we have
    \begin{align*}
        &\mbE[A_{t_1}^2A_{t_2}^2]=\sigma^4+\mrO(T^{-1}),\quad t_1\neq t_2.
    \end{align*}
    Consequently,
    \begin{align*}
        &\frac{1}{T_2^2}\sum_{t_1\neq t_2}^{T_2}(\mbE[(X_{t_1}-\tau)^2(X_{t_2}-\tau)^2]-\sigma^4)=\frac{1}{T_2^2}\sum_{t_1\neq t_2}^{T_2}(\mbE[A_{t_1}^2A_{t_2}^2]-\sigma^4)+\mrO(T^{-1})\leq\mrO(T^{-1}).
    \end{align*}
    Note that $\mbE[(X_t-\tau)^4]\leq8\mbE[A_t^4]+8\mbE[\mfv^4]\leq\mrO(1)$, so $T_2^{-2}\sum_{t=1}^{T_2}\mbE[(X_t-\tau)^4]-\sigma^4\leq\mrO(1)$, which completes the proof of $\hat{\sigma}_{T_2}^2\overset{\mbP}{\longrightarrow}\sigma^2$.
    
    Finally, recall the definition of $\widehat{\bbmu'\bbSig\bbmu}$ in \eqref{Eq of estimator}, we have
    \begin{align*}
    	&\widehat{\bbmu'\bbSig\bbmu}-\bbmu'\bbSig\bbmu\\
    	&=(\hat{\bbmu}_1-\bbmu)'\widehat{\bbS}_0(\hat{\bbmu}_2-\bbmu)+(\hat{\bbmu}_1-\bbmu)'\widehat{\bbS}_0\bbmu+\bbmu'\widehat{\bbS}_0(\hat{\bbmu}_2-\bbmu)+\bbmu'(\widehat{\bbS}_0-\bbSig)\bbmu,
    \end{align*}
    since $\mcI_0,\mcI_1,\mcI_2$ are disjoint, then $\hat{\bbmu}_1,\hat{\bbmu}_2$ and $\widehat{\bbS}_0$ are independent, which implies that $\mbE[\widehat{\bbmu'\bbSig\bbmu}-\bbmu'\bbSig\bbmu]=0$. Note that
    \begin{align*}
    	&\Var((\hat{\bbmu}_1-\bbmu)'\widehat{\bbS}_0(\hat{\bbmu}_2-\bbmu))=\frac{1}{|\mcI_1|^2\cdot|\mcI_2|^2}\sum_{t_1\in\mcI_1,t_2\in\mcI_2}\mbE[(\bby_{t_1}'\bbGa'\widehat{\bbS}_0\bbGa\bby_{t_2})^2]\\
    	&=\frac{1}{|\mcI_1|\cdot|\mcI_2|}\mbE[\tr((\bbGa'\widehat{\bbS}_0\bbGa)^2)]\leq\frac{n\|\bbSig\|^2}{|\mcI_1|\cdot|\mcI_2|}\mbE[\|\widehat{\bbS}_0\|^2]\leq\mrO(T^{-1}),
    \end{align*}
    so $(\hat{\bbmu}_1-\bbmu)'\widehat{\bbS}_0(\hat{\bbmu}_2-\bbmu)\overset{\mbP}{\longrightarrow}0$. Similarly, we can prove that
    $$(\hat{\bbmu}_1-\bbmu)'\widehat{\bbS}_0\bbmu,\bbmu'\widehat{\bbS}_0(\hat{\bbmu}_2-\bbmu)\overset{\mbP}{\longrightarrow}0$$
    by showing their variances are asymptotically vanishing. For $\bbmu'(\widehat{\bbS}_0-\bbSig)\bbmu$, by the law of large number, we have
    \begin{align*}
    	\bbmu'(\widehat{\bbS}_0-\bbSig)\bbmu\overset{\mbP}{\sim}\frac{1}{|\mcI_0|}\sum_{t\in\mcI_0}\{(\bbmu'\bbGa\bby_t)^2-\bbmu'\bbSig\bbmu\}\overset{\mbP}{\longrightarrow}0,
    \end{align*}
    which completes the proof.
\end{proof}

\subsection{Properties of SDE}\label{App of SDE}
In this part, we investigate several important properties of the following SDE:
\begin{align}
	\left\{\begin{array}{l}
		dY_s^{x,v,t}=\kappa\sign(Y_s^{x,v,t})ds+(1-2\beta\sign(Y_s^{x,v,t})V_s^{v,t})^{1/2}dB_s, \\
		dV_s^{v,t}=-\sign(Y_s^{x,v,t})ds,
	\end{array}\right.\quad s\in[t,1],\label{Eq of SDE covariance}
\end{align}
where $Y_t^{x,v,t}=x,V_t^{v,t}=v$ and $x,\kappa\in\mbR,|v|\leq t$ and $0\leq\beta<1/2$. Here, $\{B_s:s\in[0,1]\}$ is the standard Brownian motion.

\subsubsection{Existence and uniqueness of the weak solution}\label{subsec of existence and uniqueness of the weak solution}
\begin{pro}\label{Pro of existence and uniqueness of the weak solution}
	The SDE \eqref{Eq of SDE covariance} admits a unique weak solution.
\end{pro}
\begin{proof}
	First, note that $|V_s^{v,t}|\leq |v|+(s-t)\leq 1$, so $1-2\beta\sign(Y_s^{x,v,t})V_s^{v,t}>0$ is always valid, then let's define the following drift-to-diffusion ratio process:
	\begin{align*}
		\gamma_s:=\frac{\kappa\sign(Y_s^{x,v,t})}{\sqrt{1-2\beta\sign(Y_s^{x,v,t})V_s^{v,t}}},
	\end{align*}
	where $|\gamma_s|\leq\frac{|\kappa|}{\sqrt{1-2|\beta|}}\leq\mrO(1)$. Thus,
	\begin{align*}
		\mbE\left[\exp\left(\frac{1}{2}\int_t^1\gamma_s^2ds\right)\right]\leq\mrO(1).
	\end{align*}
	For simplicity, let $\mbP_{\kappa}$ be the probability measure of $(Y_s^{x,v,t}(\kappa,\beta),V_s^{v,t}(\kappa,\beta))$. By the Girsanov's theorem, we have
	\begin{align*}
		\frac{d\mbP_0}{d\mbP_{\kappa}}=\exp\left(-\int_t^1\gamma_sdB_s-\frac{1}{2}\int_t^1\gamma_s^2ds\right)
	\end{align*}
	such that $dW_s=dB_s+\gamma_sds$ is a standard Brownian motion under $\mbP_0$. Besides, under $\mbP_0$, the original SDE \eqref{Eq of SDE covariance} transforms into the following driftless system:
	\begin{align*}
		\left\{\begin{array}{l}
			dY_s^{x,v,t}=(1-2\beta\sign(Y_s^{x,v,t})V_s^{v,t})^{1/2}dW_s,\\
			dV_s^{v,t}=-\sign(Y_s^{x,v,t})ds,
		\end{array}\right.
	\end{align*}
	Hence, the original system admits \eqref{Eq of SDE covariance} a unique weak solution if and only if the above driftless system admits a unique weak solution.
	
	Next, assume that $\{W_s:s\in[0,1]\}$ is a standard Brownian motion defined on the probability space $\{\Omega,\mcF_s^{(0)},\mbP_0:s\in[0,1]\}$, let's show that for any path $\{W_s(\omega):s\in[t,1],\omega\in\Omega\}$, $\{Y_s^{x,v,t}(\omega):s\in[t,1],\omega\in\Omega\}$ is uniquely determined. Since the quadratic variation $\langle Y^{x,v,t}\rangle_s$ satisfies that
	\begin{align*}
		d\langle Y^{x,v,t}\rangle_s=ds-\beta d(V_s^{v,t})^2,\quad\implies\quad \langle Y^{x,v,t}\rangle_s=(s-t)+\beta\{(V_s^{v,t})^2-v^2\},
	\end{align*}
	by the Dambis-Dubins-Schwarz theorem, any continuous martingale can be represented as a time-changed Brownian motion, then $Y_s^{x,v,t}$ satisfies that
	\begin{align*}
		Y_s^{x,v,t}\overset{d}{\sim}x+W_{\tau_s},
	\end{align*}
	where $\tau_s:=\langle Y^{x,v,t}\rangle_s$. Hence, the pathwise uniqueness of $Y_s^{x,v,t}(\omega)$ is equivalent to those of $V_s^{v,t}(\omega)$. Note that $\frac{d\tau_s(\omega)}{ds}>0$ by $1>2\beta|\sign(Y_s^{x,v,t})V_s^{v,t}|$ for all $s\in[t,1]$, so the inverse function $s=\rho(\tau_s(\omega))$ exists and $V_s^{v,t}(\omega)=V_{\rho(\tau_s(\omega))}^{v,t}(\omega)$. Consequently, we have
	\begin{align*}
		\frac{dV_{\rho(\tau_s(\omega))}^{v,t}(\omega)}{d\tau_s(\omega)}&=\frac{dV_s^{v,t}(\omega)}{ds}\frac{ds}{d\tau_s(\omega)}=\frac{\sign(Y_s^{x,v,t}(\omega))}{1-2\beta V_s^{v,t}(\omega)\sign(Y_s^{x,v,t}(\omega))}\\
		&=\frac{\sign(x+W_{\tau_s}(\omega))}{1-2\beta V_s^{v,t}(\omega)\sign(x+W_{\tau_s}(\omega))}:=h\big(V_{\rho(\tau_s(\omega))}^{v,t}(\omega),\tau_s(\omega)\big).
	\end{align*}
	Since the partial derivative of $h(v,\tau)$ with respect to $v$ satisfies that
	\begin{align*}
		\left|\frac{\partial h(v,\tau)}{\partial v}\right|=\frac{2|\beta|}{(1-2\beta v\sign(x+W_{\tau}(\omega)))^2}<\infty,
	\end{align*}
	the function $h(v,\tau)$ is globally Lipschitz continuous with respect to $v$. By the Picard–Lindelöf theorem, $V_{\rho(\tau_s(\omega))}^{v,t}(\omega)$ is uniquely determined once $W_{\tau_s}(\omega)$ is given, so does
	\begin{align*}
		s=\rho(\tau_s(\omega))=t+\tau_s(\omega)+\beta\{V_{\tau_s(\omega)}^{v,t}(\omega)^2-v^2\},
	\end{align*}
	which implies that $x+W_{\tau_s}$ is pathwise uniquely determined under $\mbP_0$, then the driftless system admits a unique weak solution by the Yamada-Watanabe theorem.
\end{proof}

\subsubsection{Tail probabilities of the dependent bandit distribution}\label{sec of tail probability}
\begin{pro}
	Let $Y_1^{0,0,0}$ be the solution of \eqref{Eq of SDE covariance}, for any fixed $a>0$, we have
	\begin{align}
		\left\{\begin{array}{ll}
			\mbP(|Y_1^{0,0,0}(\kappa,\beta)|<a)\leq\mbP(|Y_1^{0,0,0}(0,\beta)|<a), & \kappa\geq0, \\
			\mbP(|Y_1^{0,0,0}(\kappa,\beta)|<a)\geq\mbP(|Y_1^{0,0,0}(0,\beta)|<a), & \kappa<0.
		\end{array}\right.\label{Eq of tail estimation 1}
	\end{align}
	Moreover, when $|\kappa|>a$ is sufficiently large, we have
	\begin{align}
		\left\{\begin{array}{cc}
			\mbP(|Y_1^{0,0,0}|<a)\leq\mrO(e^{-C_{\beta}(a-\kappa)^2}), & \kappa>0, \\
			\mbP(|Y_1^{0,0,0}|>a)\leq\mrO(e^{-C_{\beta}(a-\kappa)^2}), & \kappa<0,
		\end{array}\right.\label{Eq of tail estimation 2}
	\end{align}
	where $C_{\beta}>0$ is a constant depending on $\beta$.
\end{pro}
\begin{proof}
	We first prove \eqref{Eq of tail estimation 1}. Without loss of generality, we assume that $\kappa\geq0$. For simplicity, we denote $Y_s=Y_s^{0,0,0}$ and $V_s=V_s^{0,0}$. Here, let  
	\begin{align}
		\varphi_{a,\epsilon}(y)=\Phi((a-y)/\epsilon)-\Phi(-(a+y)/\epsilon),\label{Eq of approximator}
	\end{align}
	where $\epsilon>0$ and $\Phi(\cdot)$ is the CDF of the standard normal distribution. We have
	\begin{align*}
		\lim_{\epsilon\downarrow0}\sup_{y\in\mbR}|\varphi_{a,\epsilon}(y)-1_{|y|<a}|=0,
	\end{align*}
	and $\varphi_{a,\epsilon}(y)$ is a smooth function. To conclude \eqref{Eq of tail estimation 1}, it suffices that prove that
	\begin{align*}
		\mbE[\varphi_{a,\epsilon}(Y_1^{0,0,0}(\kappa,\beta))]\leq\mbE[\varphi_{a,\epsilon}(Y_1^{0,0,0}(0,\beta))],\quad\kappa\geq0.
	\end{align*}
	Recall the Girsanov's transformation in Proposition \ref{Pro of existence and uniqueness of the weak solution}, we have
	\begin{align*}
		\mbE_{\kappa}[\varphi_{a,\epsilon}(Y_1)]=\mbE_0[\varphi_{a,\epsilon}(Y_1)Z_{\kappa}],
	\end{align*}
	where $\mbE_{\kappa}$ is the expectation under the probability measure $\mbP_{\kappa}$ and
	\begin{align*}
		Z_{\kappa}=\exp\left(\kappa\int_0^1\gamma_s dW_s-\frac{\kappa^2}{2}\int_0^1\gamma_s^2ds\right),\quad\gamma_s=\frac{\sign(Y_s)}{(1-2\beta\sign(Y_s)V_s)^{1/2}}.
	\end{align*}
	By the Itô-Tanaka formula \cite{saporito2018functional}, we have
	\begin{align}
		d|Y_s|=\sign(Y_s)dY_s+d L_s^0(Y),\label{Eq of Tanaka formula}
	\end{align}
	where $L_s=L_s^0(Y)$ is the local time of $Y_s$ at zero defined as
	\begin{align}
		L_s^0(Y)=\int_0^s\delta_{\mathtt{Dirac}}(Y_r)d\langle Y\rangle_r,\label{Eq of local time}
	\end{align}
	and $\delta_{\mathtt{Dirac}}(\cdot)$ is the Dirac delta function. Thus, we derive that
	\begin{align*}
		&\int_0^1\gamma_s dW_s=\int_0^1\frac{1+2\beta\sign(Y_s)V_s}{1-4\beta^2V_s^2}dY_s=\int_0^1\frac{d|Y_s|-dL_s}{1-4\beta^2V_s^2}+\int_0^1\frac{2\beta V_s}{1-4\beta^2V_s^2}dY_s.
	\end{align*}
	Based on integration by parts, we have
	\begin{align*}
		&\int_0^1\frac{2\beta V_s}{1-4\beta^2V_s^2}dY_s=\frac{2\beta Y_1V_1}{1-4\beta^2V_1^2}-\int_0^1Y_sd\frac{2\beta V_s}{1-4\beta^2V_s^2}\\
		&=\frac{2\beta Y_1V_1}{1-4\beta^2V_1^2}+\int_0^1\frac{2\beta|Y_s|(1+4\beta^2V_s^2)}{(1-4\beta^2V_s^2)^2}ds,
	\end{align*}
	and
	\begin{align*}
		&\int_0^1\frac{d|Y_s|-dL_s}{1-4\beta^2V_s^2}=\frac{|Y_1|}{1-4\beta^2V_1}+\int_0^1\frac{8\beta^2Y_sV_s}{(1-4\beta^2V_s^2)^2}ds-\int_0^1\frac{1}{1-4\beta^2V_s^2}dL_s,
	\end{align*}
	so
	\begin{align*}
		&\int_0^1\gamma_s dW_s=\frac{|Y_1|+2\beta Y_1V_1}{1-4\beta^2V_1}+\int_0^1\frac{2\beta|Y_s|(1+4\beta^2V_s^2)+8\beta^2Y_sV_s}{(1-4\beta^2V_s^2)^2}ds-\int_0^1\frac{1}{1-4\beta^2V_s^2}dL_s\\
		&=\frac{|Y_1|}{1-2\beta \sign(Y_1)V_1}+\int_0^1\frac{2\beta|Y_s|}{(1-2\beta\sign(Y_s)V_s)^2}ds-\int_0^1\frac{1}{1-4\beta^2V_s^2}dL_s.
	\end{align*}
	Next, under $\mbP_0$, we define a skeletal $\sigma$-field $\mcF_{\mathtt{skel}}$ as below:
	\begin{align}
		\mcF_{\mathtt{skel}}:=\sigma(\mcF_{\mcS} \cup \mcF_{\mcL}\cup \mcN_{\mbP_0}),\label{Eq of skeleta sigma field}
	\end{align}
	where $\mcN_{\mbP_0}$ is the collection of all $\mbP_0$-null sets and
	\begin{align*}
		\left\{\begin{array}{l}
			\mcF_{\mcS}:=\sigma(\sign(Y_s):s\in[0,1]), \\
			\mcF_{\mcL}:=\sigma(L_s(Y_s):s\in[0,1]). 
		\end{array}\right.
	\end{align*}
	Conditional on $\mcF_{\mathtt{skel}}$, the procceses $\sign(Y_s),V_s=-\int_0^s\sign(Y_r)dr$ and $L_s(Y_s)$ are completely deterministic, so do $\int_0^1\gamma_s^2ds$ and $\int_0^1\frac{1}{1-4\beta^2V_s^2}dL_s$, which yields that
	\begin{align*}
		Z_{\kappa}|\mcF_{\mathtt{skel}}\propto\exp\left(\frac{\kappa|Y_1|}{1-2\beta \sign(Y_1)V_1}+\int_0^1\frac{2\kappa\beta|Y_s|}{(1-2\beta\sign(Y_s)V_s)^2}ds\right).
	\end{align*}
	Moreover, we define the excursion intervals as below:
	\begin{align}
		[0,1]\backslash\{s\in[0,1]:Y_s=0\}=\bigcup_{i=1}^{\infty}(g_i,d_i),\label{Eq of excursion interval}
	\end{align}
	and $(g_i,d_i)$ are also deterministic conditional on $\mcF_{\mathtt{skel}}$. Note that
	\begin{align*}
		\mbE_{\kappa}[\varphi_{a,\epsilon}(|Y_1|)]-\mbE_0[\varphi_{a,\epsilon}(|Y_1|)]=\mbE_0[\varphi_{a,\epsilon}(|Y_1|)Z_{\kappa}]-\mbE_0[\varphi_{a,\epsilon}(|Y_1|)]\mbE_0[Z_{\kappa}]=\Cov_0(Z_{\kappa},\varphi_{a,\epsilon}(|Y_1|)),
	\end{align*}
	and
	\begin{align*}
		\Cov_0(Z_{\kappa},\varphi_{a,\epsilon}(|Y_1|))=\mbE_0[\Cov_0(Z_{\kappa},\varphi_{a,\epsilon}(|Y_1|)|\mcF_{\mathtt{skel}})]+\Cov_0(\mbE_0[Z_{\kappa}|\mcF_{\mathtt{skel}}],\mbE_0[\varphi_{a,\epsilon}(|Y_1|)|\mcF_{\mathtt{skel}}]).
	\end{align*}
	We first show that
	\begin{align}
		\mbE_0[\Cov_0(Z_{\kappa},\varphi_{a,\epsilon}(|Y_1|)|\mcF_{\mathtt{skel}})]\leq0.\label{Eq of negative covariance 1}
	\end{align}
	By Lemma \ref{Lem of BES3}, for $s\in[g_i,d_i]$, $|Y_s||\mcF_{\mathtt{skel}}$ has the same distribution as a 3-dimensional Bessel bridge $(\text{BES}^3)$ on the time interval $[0,\int_{d_i}^{g_i}(1-2\beta\sign(Y_s)V_s)ds]$, and its transition kernel is stochastically monotone. Note that $-\varphi_{a,\epsilon}(|Y_1|)$ and $Z_{\kappa}|\mcF_{\mathtt{skel}}$ as non-decreasing continuous functionals of $\{|Y_s|:s\in[0,1]\}$, by Lemma \ref{Lem of FKG}, we have
	\begin{align*}
		-\mbE_0[\varphi_{a,\epsilon}(|Y_1|)Z_{\kappa}|\mcF_{\mathtt{skel}}]\geq-\mbE_0[\varphi_{a,\epsilon}(|Y_1|)|\mcF_{\mathtt{skel}}]\mbE_0[Z_{\kappa}|\mcF_{\mathtt{skel}}],
	\end{align*}
	so $\Cov_0(Z_{\kappa},\varphi_{a,\epsilon}(|Y_1|)|\mcF_{\mathtt{skel}})\leq0$ and it concludes \eqref{Eq of negative covariance 1}. Furthermore, we claim that
	\begin{align}
		\Cov_0(\mbE_0[Z_{\kappa}|\mcF_{\mathtt{skel}}],\mbE_0[\varphi_{a,\epsilon}(|Y_1|)|\mcF_{\mathtt{skel}}])\leq0.\label{Eq of negative covariance 2}
	\end{align}
	Let $g_*:=\sup\{s\in[0,1]:|Y_s|=0\}$. Conditional on $\mcF_{\mathtt{skel}}$, by \cite[Theorem 1]{williams1970decomposing}, $|Y_1|$ has the same distribution as a $\text{BES}^3$ process $B_s^{(3)}$ at time $\Delta_*:=\int_{g_*}^1(1-2\beta\sign(Y_s)V_s)ds$. Since
	\begin{align*}
		B_{\Delta_*}^{(3)}\overset{d}{\sim}\sqrt{\Delta_*}B_1^{(3)},
	\end{align*}
	then
	\begin{align*}
		\mbE_0[\varphi_{a,\epsilon}(|Y_1|)|\mcF_{\mathtt{skel}}]=\int_0^{\infty}\varphi_{a,\epsilon}(\sqrt{\Delta_*}y)p_{B_1^{(3)}}(y)dy,
	\end{align*}
	where $p_{B_1^{(3)}}(\cdot)$ is the density of $B_1^{(3)}$ and the above function is non-increasing in $\Delta_*$ as $\varphi_{a,\epsilon}(\cdot)$ is non-increasing. On the other hand, note that
	\begin{align*}
		\mbE_0[Z_{\kappa}|\mcF_{\mathtt{skel}}]|\Delta_*\propto\mbE_0\left[\exp\left(\frac{\kappa|Y_1|}{1-2\beta \sign(Y_1)V_1}+\int_{g_*}^1\frac{2\kappa\beta|Y_s|}{(1-2\beta\sign(Y_s)V_s)^2}ds\right)\Big|\mcF_{\mathtt{skel}}\right],
	\end{align*}
	and 
	\begin{align*}
		&\frac{\kappa|Y_1|}{1-2\beta \sign(Y_1)V_1}+\int_{g_*}^1\frac{2\kappa\beta|Y_s|}{(1-2\beta\sign(Y_s)V_s)^2}ds\Big|\mcF_{\mathtt{skel}}\\
		\overset{d}{\sim}&\frac{\kappa\sqrt{\Delta_*}B_1^{(3)}}{1-2\beta \sign(Y_1)V_1}+(\Delta_*)^{3/2}\int_0^1\frac{2\kappa\beta B_s^{(3)}}{(1-2\beta\sign(Y_s)V_s)^2}ds,
	\end{align*}
	which is non-decreasing in $\Delta_*$, so does $\mbE_0[Z_{\kappa}|\mcF_{\mathtt{skel}}]|\Delta_*$. Finally, by \eqref{Eq of 1-dim FKG}, we conclude \eqref{Eq of negative covariance 2}. Combining with \eqref{Eq of negative covariance 1} and \eqref{Eq of negative covariance 2}, we derive $\mbE_{\kappa}[\varphi_{a,\epsilon}(|Y_1|)]\leq\mbE_0[\varphi_{a,\epsilon}(|Y_1|)]$ for $\kappa\geq0$, which completes proofs of the first conclusion.

	For the second conclusion,  by \eqref{Eq of Tanaka formula}, we have
	\begin{align*}
		Z_1=\int_0^1\sign(Y_s)dY_s+L_1^0(Y),
	\end{align*}
	where $L_1^0(Y)$ is defined in \eqref{Eq of local time}. By \eqref{Eq of SDE covariance}, we derive that
	\begin{align*}
		|Y_1|=\kappa+\int_0^1\sign(Y_s)(1-2\beta\sign(Y_s)V_s)^{1/2}dB_s+L_1^0(Y).
	\end{align*}
	Note that 
	\begin{align*}
		|Y_1|<a\implies M_1:=\int_0^1\sign(Y_s)(1-2\beta\sign(Y_s)V_s)^{1/2}dB_s<a-\kappa,
	\end{align*}
	where $M_1$ is a continuous martingale and its total quadratic variation is
	$$\langle M\rangle_1=\int_0^1(1-2\beta\sign(Y_s)V_s)ds\leq1+2\beta,$$
	by the Bernstein's inequality for continuous martingales \cite{dzhaparidze2001bernstein}, when $\kappa>x$, we have
	\begin{align*}
		\mbP(Z_1<a)\leq\mbP(M_1<a-\kappa)\leq\frac{1}{2}\mbP(|M_1|>\kappa-a)\leq\frac{1}{2}\exp\left(-\frac{(\kappa-a)^2}{2(1+2\beta)}\right).
	\end{align*}
	On the other hand, when $\kappa<0$, we have
	\begin{align*}
		Z_1>a\implies M_1>a-\kappa+L_1^0(Y)\geq a-\kappa,
	\end{align*}
	where we use the fact that $L_1^0(Y)\geq0$. Thus, the Bernstein's inequality implies that
	\begin{align*}
		\mbP(Z_1>a)\leq\frac{1}{2}\mbP(|M_1|>(|\kappa|+a)/2)\leq\frac{1}{2}\exp\left(-\frac{(|\kappa|+a)^2}{8(1+2\beta)}\right),
	\end{align*} 
	which completes the proof of \eqref{Eq of tail estimation 2}.
\end{proof}

\begin{lem}\label{Lem of BES3}
	For $\mcF_{\mathtt{skel}}$ in \eqref{Eq of skeleta sigma field}, $Y_s=Y_s^{0,0,0}(0,\beta),V_s=V_s^{0,0}(0,\beta)$ in \eqref{Eq of SDE covariance} and the excursion intervals $(g_i,d_i)$ in \eqref{Eq of excursion interval}, conditional on $\mcF_{\mathtt{skel}}$, we have
	\begin{enumerate}
		\item The distribution of $|Y_s|$ for $s\in[g_i,d_i]$ is the same as a Bessel ($\text{BES}^3$) bridge on the time interval $[0,\int_{g_i}^{d_i}(1-2\beta\sign(Y_s)V_s)ds]$;
		\item The transition kernel of $|Y_s|$ is stochastically monotone, i.e. for any $0\leq s<t\leq1,0\leq x\leq x'$ and measurable set $B\subset\Omega$
		\begin{align*}
			\mbP_{\mcF_{\mathtt{skel}}}(|Y_t|\in B|\ |Y_s|=x)\leq \mbP_{\mcF_{\mathtt{skel}}}(|Y_t|\in B|\ |Y_s|=x').
		\end{align*}
	\end{enumerate}
\end{lem}
\begin{proof}
	By our previous arguments in Proposition \ref{Pro of existence and uniqueness of the weak solution}, we have
	\begin{align*}
		Y_s=W_{\tau_s},\quad \tau_s=\int_0^s(1-2\beta\sign(Y_s)V_s)ds,
	\end{align*}
	where $\tau_s$ is a strictly increasing, absolutely continuous function of $s\in[0,1]$. Conditional on $\mcF_{\mathtt{skel}}$, $1-2\beta\sign(Y_s)V_s$ is deterministic for $s\in[0,1]$, so do excursion intervals $(g_i,d_i)$ defined in \eqref{Eq of excursion interval}. Next, note that $W_{\tau_{g_i}}=W_{\tau_{d_i}}=0$ and $W_{\tau_s}\neq0$ for $\tau_s\in(\tau_{g_i},\tau_{d_i})$, thus, conditional on $\mcF_{\mathtt{skel}}$, $|Y_s|=|W_{\tau_s}|$ is a standard Brownian excursion on $s\in[g_i,d_i]$. By \cite[Theorem 1 \& 2]{williams1970decomposing}, the standard Brownian excursion on any deterministic time interval has the same law as the $\text{BES}^3$ bridge on the same time interval.

	For the second conclusion, let $B_s^{(3)}$ be a $\text{BES}^3$ process. Based on the first conclusion, it suffices to prove that the transition kernel of $\text{BES}^3$ bridge is stochastically monotone. For $0\leq s_1<s_2\leq T$ and $x_1,x_2\in\mbR^+$, the transition density of $B_{s_1}^{(3)}=x_1$ and $B_{s_2}^{(3)}=x_2$ conditional on $B_T^{(3)}=0$ is given as
	\begin{align*}
		p_{s_1,s_2}^{\mathtt{bridge}}(x_1,x_2)=p_{s_2-s_1}(x_1,x_2)\frac{p_{T-s_2}(x_2,0)}{p_{T-s_1}(x_1,0)},
	\end{align*}
	where we use the Bayes’ rule and the unconditional transition density of $B_{s_1}^{(3)}=x_1$ and $B_{s_2}^{(3)}=x_2$ as below:
	\begin{align*}
		p_{s_2-s_1}(x_1,x_2)&=\frac{1}{\sqrt{2\pi(s_2-s_1)}}\frac{x_2}{x_1}\left\{\exp\left(-\frac{(x_2-x_1)^2}{2(s_2-s_1)}\right)-\exp\left(-\frac{(x_2+x_1)^2}{2(s_2-s_1)}\right)\right\}\\
		&=\frac{2}{\sqrt{2\pi(s_2-s_1)}}\frac{x_2}{x_1}\exp\left(-\frac{x_1^2+x_2^2}{2(s_2-s_1)}\right)\sinh\left(\frac{x_1x_2}{2(s_2-s_1)}\right)
	\end{align*}
	Here, let's prove that
	\begin{align*}
		\frac{\partial^2}{\partial x_1\partial x_2}\log\big(p_{s_1,s_2}^{\mathtt{bridge}}(x_1,x_2)\big)\geq0.
	\end{align*}
	Note that
	\begin{align*}
		&\frac{\partial^2}{\partial x_1\partial x_2}\log\big(p_{s_1,s_2}^{\mathtt{bridge}}(x_1,x_2)\big)=\frac{\partial^2}{\partial x_1\partial x_2}\log\big(p_{s_2-s_1}(x_1,x_2)\big)\\
		&=\frac{\partial^2}{\partial x_1\partial x_2}\log \sinh\left(\frac{x_1x_2}{2(s_2-s_1)}\right)=\frac{\sinh\left(\frac{x_1x_2}{s_2-s_1}\right)-\frac{x_1x_2}{s_2-s_1}}{2(s_2-s_1)\sinh^2\left(\frac{x_1x_2}{2(s_2-s_1)}\right)}\geq0,
	\end{align*}
	where we use the fact that $\sinh(x)\geq x$ for $x\geq0$. Consequently, for any $0\leq x_1<x_2$ and $y\geq0$, we have 
	\begin{align*}
		p_{s_1,s_2}^{\mathtt{bridge}}(x_1,y)\leq p_{s_1,s_2}^{\mathtt{bridge}}(x_2,y)\implies\int_{y\in B}p_{s_1,s_2}^{\mathtt{bridge}}(x_1,y)dy\leq\int_{y\in B}p_{s_1,s_2}^{\mathtt{bridge}}(x_2,y)dy,
	\end{align*}
	which proves the second conclusion.
\end{proof}

\begin{lem}\label{Lem of FKG}
	Let $\{X_s:s\in[0,1]\}$ be a Markov process such that its transition kernel is stochasticaly monotone and its canonical path space 
	\begin{align*}
		\{X_s(\omega):\omega\in\Omega,s\in[0,1]\}\subset C([0,1],\mbR^+\cup\{0\}).
	\end{align*}
	Given any two non-decreasing functionals $F(\cdot)$ and $G(\cdot)$ defined on $C([0,1],\mbR^+\cup\{0\})$, that is, for any $\vec{x},\vec{y}\in C([0,1],\mbR^+\cup\{0\})$ such that $\vec{x}_s\leq\vec{y}_s$ for all $s\in[0,1]$, 
	$$F(\vec{x})\leq F(\vec{y}),$$
	suppose $\mbE[F(X)^2],\mbE[G(X)^2]\leq\mrO(1)$, then we have
	\begin{align}
		\mbE[F(X)G(X)]\geq\mbE[F(X)]\mbE[G(X)].\label{Eq of FKG inequality}
	\end{align}
\end{lem}
\begin{proof}
	Consider a time partition of $[0,1]$ as follows:
	\begin{align*}
		\pi_n:=\{t_k=k/n:k=0,1,\cdots,n\},\quad n\in\mbN^+,\quad\pi_0=\{0\},
	\end{align*}
	for any path $\vec{x}\in C([0,1],\mbR^+\cup\{0\})$, define a finite-dimensional projection operator $P_n$ and a order-preserving linear interpolation operator $\Lambda_n$ such that
	\begin{align*}
		P_n(\vec{x})=(\vec{x}_{t_0},\cdots,\vec{x}_{t_n})'\in\mbR^{2^n+1},\quad\Lambda_n(P_n(\vec{x}))_s=x_{t_k}+n(x_{t_{k+1}}-x_{t_k})(s-t_k),\quad s\in[t_k,t_{k+1}],
	\end{align*}
	it is easy to see that for any $\vec{x}_s\leq \vec{y}_s$, $\Lambda_n(P_n(\vec{x}))_s\leq\Lambda_n(P_n(\vec{y}))_s$. Here, let's first prove that
	\begin{align}
		\mbE[F(\Lambda_n(P_n(X)))G(\Lambda_n(P_n(X)))]\geq\mbE[F(\Lambda_n(P_n(X)))]\mbE[G(\Lambda_n(P_n(X)))].\label{Eq of FKG inequality discrete}
	\end{align}
	When $n=0$, we define $P_0(X)=X_0$, then for an independent copy $X_0'$ of $X_0$,
	\begin{align*}
		(F(\Lambda_0(X_0))-F(\Lambda_0(X_0')))(G(\Lambda_0(X_0))-G(\Lambda_0(X_0')))\geq0,
	\end{align*}
	which implies that
	\begin{align}
		\mbE[F(\Lambda_0(X_0))G(\Lambda_0(X_0))]\geq\mbE[F(\Lambda_0(X_0))]\mbE[G(\Lambda_0(X_0))].\label{Eq of 1-dim FKG}
	\end{align}
	By the induction, suppose \eqref{Eq of FKG inequality discrete} holds for $n-1$, define
	\begin{align*}
		&\tilde{f}(x_0,\cdots,x_{n-1}):=\mbE[F(\Lambda_n(X_{t_0},\cdots,X_{t_n}))|X_{t_0}=x_0,\cdots,X_{t_{n-1}}=x_{n-1}]\\
		&=\int_0^RF(\Lambda_n(x_0,\cdots,x_{n-1},y))p_{t_{n-1},t_n}(x_{n-1},y)dy,
	\end{align*}
	so does $\tilde{g}(x_0,\cdots,x_{n-1})$. Since the transition kernel $p_{t_{n-1},t_n}(x_{n-1},y)$ is stochastically monotone, combining with $F(\cdot)$ is non-decreasing in $x_{n-1}$, we have for any $x_{n-1}\leq x_{n-1}'$,
	\begin{align*}
		\tilde{f}(x_0,\cdots,x_{n-1})\leq\tilde{f}(x_0,\cdots,x_{n-1}'),
	\end{align*}
	so does $\tilde{g}(x_0,\cdots,x_{n-1})$. By the induction assumption at step $n-1$, we have
	\begin{align*}
		\mbE[\tilde{f}(X_{t_0},\cdots,X_{t_{n-1}})\tilde{g}(X_{t_0},\cdots,X_{t_{n-1}})]\geq\mbE[\tilde{f}(X_{t_0},\cdots,X_{t_{n-1}})]\mbE[\tilde{g}(X_{t_0},\cdots,X_{t_{n-1}})],
	\end{align*}
	i.e. \eqref{Eq of FKG inequality discrete} holds for $n$. Finally, since the sample paths of $X_t$ are almost surely continuous, $\Lambda_n(P_n(X))$ converges uniformly to $X$ as $n\to\infty$. Both $F$ and $G$ are also continuous, so $F(\Lambda_n(P_n(X)))\to F(X)$ and $G(\Lambda_n(P_n(X)))\to G(X)$ almost surely. By the Lebesgue dominated convergence theorem, we conclude \eqref{Eq of FKG inequality}.
\end{proof}

\subsubsection{Boundedness of derivatives}\label{sec of bounded derivatives}
Note that the drift and diffusion terms in the SDE \eqref{Eq of SDE covariance} are non-continuous due to the existence of $\sign(\cdot)$, to facilitate the analysis of partial derivatives of \eqref{Eq of SDE covariance} with respect to the initial points $(x,v)$, let $S_{\delta}(x)$ be a smooth approximation of $\sign(x)$ such that $\sup_{x\in\mbR}|S_{\delta}(x)|\leq1$ and $S_{\delta}(x)=\sign(x)$ for $|x|>\delta$, where $\delta>0$ is a sufficiently small constant, then we consider the following smoothed SDE of \eqref{Eq of SDE covariance}:
\begin{align}
	\left\{\begin{array}{l}
		dY_{s;\delta}^{x,v,t}=\kappa S_{\delta}(Y_{s;\delta}^{x,v,t})ds+(1-2\beta S_{\delta}(Y_{s;\delta}^{x,v,t})V_{s;\delta}^{v,t})^{1/2}dB_s, \\
		dV_{s;\delta}^{v,t}=-S_{\delta}(Y_{s;\delta}^{x,v,t})ds,
	\end{array}\right.\quad s\in[t,1],\label{Eq of smoothed SDE}
\end{align}
where $Y_{t;\delta}^{x,v,t}=x$ and $V_{t;\delta}^{v,t}=v$. The existence and uniqueness of the weak solution of \eqref{Eq of smoothed SDE} can be derived by the same arguments as those in Proposition \ref{Pro of existence and uniqueness of the weak solution}, thus we omit details. Here, we provide
\begin{lem}\label{Lem of SDE approximation}
	For the weak solutions $(Y_s^{x,v,t},V_s^{v,t})$ and $(Y_{s;\delta}^{x,v,t},V_{s;\delta}^{v,t})$ of the SDEs \eqref{Eq of SDE covariance} and \eqref{Eq of smoothed SDE}, respectively, we have as $\delta\downarrow0$,
	\begin{align*}
		(Y_{s;\delta}^{x,v,t},V_{s;\delta}^{v,t})\overset{d}{\longrightarrow}(Y_s^{x,v,t},V_s^{v,t}).
	\end{align*}
\end{lem}
\begin{proof}
	For simplicity, we define a positive sequence $\{a_i\in\mbR^+:i\in\mbN^+\}$ such that $a_i>a_{i+1}$ and $\lim_{i\to\infty}a_i=0$, then denote $(Y_{s;i}^{x,v,t},V_{s;i}^{v,t})=(Y_{s;a_i}^{x,v,t},V_{s;a_i}^{v,t})$, let's show that $(Y_{s;i}^{x,v,t},V_{s;i}^{v,t})\overset{d}{\longrightarrow}(Y_s^{x,v,t},V_s^{v,t})$. By direct calculations, for any $s\in(t,1)$ and $s+\tau\in(t,1)$ with sufficiently small $\tau$, we have
	\begin{align}
		&\mbE[|Y_{s+\tau;a_i}^{x,v,t}-Y_{s;i}^{x,v,t}|^4]\leq\mrO\left(\tau^4+\mbE\left[\left(\int_s^{s+\tau}dB_r\right)^4\right]\right)\leq\mrO(\tau^2),\label{Eq of tightness}
	\end{align}
	and $\mbE[|V_{s+\tau;a_i}^{v,t}-V_{s;i}^{v,t}|^4]\leq\mrO(\tau^t)$, so $(Y_{s;i}^{x,v,t},V_{s;i}^{v,t})$ is a tight stochastic process. Moreover, let $s_i$ be a stopping time on $(Y_{s;i}^{x,v,t},V_{s;i}^{v,t})$ and $\tau_i\in[0,1]$ be a constant such that $\lim_{i\to\infty}\tau_i=0$, we can use the same argument as \eqref{Eq of tightness} to conclude that
	\begin{align*}
		(Y_{s_i+\tau_i;i}^{x,v,t},V_{s_i+\tau_i;i}^{v,t})-(Y_{s_i;i}^{x,v,t},V_{s_i;i}^{v,t})\overset{\mbP}{\longrightarrow}\boldsymbol{0}.
	\end{align*}
	Hence, by \cite[Theorem 1]{aldous1978stopping}, $\{(Y_{s;i}^{x,v,t},V_{s;i}^{v,t}):i\in\mbN^+\}$ is a tight sequence of stochastic processes, meaning every subsequence has a further weakly convergent subsequence. Without loss of generality, let $(Y_{s;*}^{x,v,t},V_{s;*}^{v,t})$ be the weak limit stochastic process of $\{(Y_{s;i}^{x,v,t},V_{s;i}^{v,t}):i\in\mbN^+\}$. Moreover, let $\mbP_i$ be the probability measure induced by $\{(Y_{s;i}^{x,v,t},V_{s;i}^{v,t}):i\in\mbN^+\}$, so does $\mbP_*$ of $(Y_{s;*}^{x,v,t},V_{s;*}^{v,t})$.
	
	Next, we define a canonical path space (readers can refer to \cite[Chapter 1, \S3]{stroock2007multidimensional} for a precise definition of the canonical path space):
	\begin{align*}
		\mcP:=\{X_s(\omega)=(Y_s^{x,v,t}(\omega),V_s^{v,t}(\omega))\in\mbR^2:\omega\in\Omega,s\in[t,1]\},
	\end{align*}
	with corresponding collection of $\sigma$-fields defined as $\{\mcF_s^{\mcP}=\sigma(X_r:t\leq r\leq s):t\leq s\leq1\}$, and $X_s$ under $\mbP_i$ is equivalent to $(Y_{s;i}^{x,v,t},V_{s;i}^{v,t})$. By \cite[Chapter 6, \S6.0]{stroock2007multidimensional}, $(Y_{s;i}^{x,v,t},V_{s;i}^{v,t})$ is a weak solution of \eqref{Eq of smoothed SDE} is equivalent to for any bounded and smooth test function $\phi(\cdot,\cdot)$,
	\begin{align*}
		M_{s;i}^X:=\phi(X_s)-\phi(x,v)-\int_t^s\mcL_i\phi(X_r)dr
	\end{align*}
	is a martingale under $\mbP_i$, where $\mcL_i:=\kappa S_{a_i}(x)\frac{\partial}{\partial x}-S_{a_i}(x)\frac{\partial}{\partial v}+\frac{1-2\beta S_{a_i}(x)v}{2}\frac{\partial^2}{\partial x^2}$. Hence, combining with the existence and uniqueness of the weak solution of \eqref{Eq of SDE covariance} in Proposition \ref{Pro of existence and uniqueness of the weak solution}, to conclude that $(Y_{s;i}^{x,v,t},V_{s;i}^{v,t})\overset{d}{\longrightarrow}(Y_s^{x,v,t},V_s^{v,t})$, it suffices to prove that
	\begin{align*}
		M_s^X:=\phi(X_s)-\phi(x,v)-\int_t^s\mcL\phi(X_r)dr
	\end{align*}
	is a martingale under $\mbP_*$, where $\mcL:=\kappa\sign(x)\frac{\partial}{\partial x}-\sign(x)\frac{\partial}{\partial v}+\frac{1-2\beta\sign(x)v}{2}\frac{\partial^2}{\partial x^2}$. For any $t\leq s_1<s_2\leq1$, let $\Phi(\cdot)$ be any bounded $\mcF_{s_1}^{\mcP}$-measurable function, it suffices that show
	\begin{align*}
		\mbE_{\mbP_*}\left[\left(\phi(X_{s_2})-\phi(X_{s_1})-\int_{s_1}^{s_2}\mcL\phi(X_r)dr\right)\Phi\right]=0.
	\end{align*}
	Note that $\mbE_{\mbP_*}=\lim_{i\to\infty}\mbE_{\mbP_i}$ and
	\begin{align*}
		\mbE_{\mbP_i}\left[\left(\phi(X_{s_2})-\phi(X_{s_1})-\int_{s_1}^{s_2}\mcL_i\phi(X_r)dr\right)\Phi\right]=0,
	\end{align*}
	so we will show that
	\begin{align}
		\lim_{i\to\infty}\mbE_{\mbP_i}\left[\Phi\int_{s_1}^{s_2}(\mcL-\mcL_i)\phi(X_r)dr\right]=0.\label{Eq of measure limit}
	\end{align}
	By direct calculations, $|(\mcL-\mcL_i)\phi(X_r)|\leq\mrO(1\{|Y_{r,i}^{x,v,t}|\leq a_i\})$, then
	\begin{align*}
		\left|\mbE_{\mbP_i}\left[\Phi\int_{s_1}^{s_2}(\mcL-\mcL_i)\phi(X_r)dr\right]\right|\leq\mrO\left(\mbE_{\mbP_i}\left[\int_{s_1}^{s_2}1\{|Y_{r,i}^{x,v,t}|\leq a_i\}dr\right]\right).
	\end{align*}
	Since the quadratic variation of $Y_{r,i}^{x,v,t}$ satisfies that $d\langle Y_i^{x,v,t}\rangle_r>(1-2\beta)dr$, i.e. $dr<\frac{1}{1-2|\beta|}d\langle Y_i^{x,v,t}\rangle_r$ and it gives that 
	$$\int_{s_1}^{s_2}1\{|Y_{r,i}^{x,v,t}|\leq a_i\}dr<\frac{1}{1-2|\beta|}\int_{s_1}^{s_2}1\{|Y_{r,i}^{x,v,t}|\leq a_i\}d\langle Y\rangle_r.$$
	Combining with the occupation time formula as below:
	\begin{align*}
		\int_{s_1}^{s_2}1\{|Y_{r,i}^{x,v,t}|\leq a_i\}d\langle Y\rangle_r=\int_{-a_i}^{a_i}L_{s_2,s_1}^z(Y_i^{x,v,t})dz,
	\end{align*}
	where $L_{s_2,s_1}^z(Y_i^{x,v,t})=\int_{s_1}^{s_2}\delta_{\mathtt{Dirac}}(z-Y_{r,i}^{x,v,t})d\langle Y_i^{x,v,t}\rangle_r$ is the local time of the process $Y_{r,i}^{x,v,t}$ at the spatial level $z$ from time $s_1$ up to time $s_2$. By the Ito-Tanaka's formula for $|Y_{r,i}^{x,v,t}-z|$, we have
	\begin{align*}
		&L_{s_2,s_1}^z(Y_i^{x,v,t})=|Y_{s_2,i}^{x,v,t}-z|-|Y_{s_1,i}^{x,v,t}-z|-\int_{s_1}^{s_2}\sign(Y_{r,i}^{x,v,t}-z)dY_{r,i}^{x,v,t},
	\end{align*}
	so
	\begin{align*}
		&\mbE_{\mbP_i}[L_{s_2,s_1}^z(Y_i^{x,v,t})]\leq\mbE_{\mbP_i}[|Y_{s_2,i}^{x,v,t}-Y_{s_1,i}^{x,v,t}|]+\mbE_{\mbP_i}\left[\left|\int_{s_1}^{s_2}\sign(Y_{r,i}^{x,v,t}-z)dY_{r,i}^{x,v,t}\right|\right]\leq\mrO(\sqrt{s_2-s_1}),
	\end{align*}
	which implies that
	\begin{align*}
		\mbE_{\mbP_i}\left[\int_{s_1}^{s_2}1\{|Y_{r,i}^{x,v,t}|\leq a_i\}dr\right]=\int_{-a_i}^{a_i}\mbE_{\mbP_i}[L_{s_2,s_1}^z(Y_i^{x,v,t})]dz\leq\mrO(a_i),
	\end{align*}
	then we conclude \eqref{Eq of measure limit}.
\end{proof}

\begin{pro}\label{Pro of bounded derivatives}
	For any $\delta>0$, let $h(\cdot,\cdot)$ be a bounded and smooth function, let $f_{\delta}(x,v,t):=\mbE[h(Y_{1;\delta}^{x,v,t},V_{1;\delta}^{v,t})]$, where $(Y_{1;\delta}^{x,v,t},V_{1;\delta}^{v,t})$ is the solution of \eqref{Eq of smoothed SDE}, for $x\in\mbR,|v|\leq t$ and $t\in[0,1]$, then we have
	\begin{enumerate}
		\item $\big|\frac{\partial^k f_{\delta}(x,v,t)}{\partial x^k}\big|\leq\mrO(\tilde{C}_k\delta^{-(2k-1)}e^{C_{k,\kappa,\beta}\delta^{-2}})$ for $k=1,2,3,4$,
		\item $\big|\frac{\partial^l f_{\delta}(x,v,t)}{\partial v^l}\big|\leq\mrO(\tilde{C}_l\delta^{-(2l-1)}e^{C_{l,\kappa,\beta}\delta^{-2}})$ for $l=1,2$,
	\end{enumerate}
	where $C_{k,\kappa,\beta}>0$ is a constant depending only on $k,\kappa,\beta$, so does $C_{l,\kappa,\beta}>0$, and 
	\begin{align*}
		\tilde{C}_k:=\max_{1\leq m\leq k}\left\{\sup_{(x,v)\in\mbR^2}\left|\frac{\partial^m }{\partial x^m}h(x,v)\right|,\sup_{(x,v)\in\mbR^2}\left|\frac{\partial^m }{\partial v^m}h(x,v)\right|\right\}.
	\end{align*}
\end{pro}
\begin{proof}
	In the following proof, we simplify $f_{\delta}(x,v,t)$ via $f(x,v,t)$ and $(Y_{s;\delta}^{x,v,t},V_{s;\delta}^{v,t})$ via $(Y_s,V_s)$. For the SDE in \eqref{Eq of smoothed SDE}, since entries of its drift vector $(\kappa S_{\delta}(y),-S_{\delta}(y))$ and diffusion matrix 
	$$\left(\begin{array}{cc}
		1-2\beta S_{\delta}(y)w & 0 \\
		0 & 0
	\end{array}\right)$$
	are bounded and smooth function of $(y,w)$ for $|w|\leq 1$, by \cite[Theorem 4.6.5]{kunita1990stochastic}, \eqref{Eq of smoothed SDE} admits a solution that is a forward stochastic flow of $C^4$-diffeomorphisms, which guarantees the existence of $\frac{\partial^k Y_s}{\partial x^k}$, $\frac{\partial^k V_s}{\partial x^k}$, $\frac{\partial^l Y_s}{\partial v^l}$ and $\frac{\partial^l V_s}{\partial v^l}$ for $k=1,2,3,4;l=1,2$ and $s\in[t,1]$. Next, let's compute the upper bound of $\mbE\big[\big|\frac{\partial^k Y_s}{\partial x^k}\big|^{2p}\big]$ for fixed $p\in\mbN^+$. When $k=1$, we have
	\begin{align*}
		d(\partial_xY_s)&=\kappa S_{\delta}'(Y_s)(\partial_xY_s)ds-\frac{\beta(S_{\delta}'(Y_s)(\partial_xY_s)V_s+S_{\delta}(Y_s)(\partial_xV_s))}{(1-2\beta S_{\delta}(Y_s)V_s)^{1/2}}dB_s\\
		:&=\mcA_s^{(1)}(\partial_xY_s)ds+\{\mcB_s^{(1)}(\partial_xY_s)+\mcD_s^{(1)}(\partial_xV_s)\}dB_s,
	\end{align*}
	where $|\mcA_s^{(1)}|\leq|\kappa|\cdot|S_{\delta}'(Y_s)|\leq|\kappa|\delta^{-1}$, so do $|\mcB_s^{(1)}|\leq\frac{|\beta|\delta^{-1}}{\sqrt{1-2|\beta|}}$ and $|\mcD_s^{(1)}|\leq\frac{|\beta|}{\sqrt{1-2|\beta|}}$. By the Itô's formula for $(\partial_xY_s)^{2p}$, we have
	\begin{align*}
		&\mbE[(\partial_xY_s)^{2p}]=1+\int_t^s\mbE[2p\mcA_s^{(1)}(\partial_xY_s)^{2p}+p(2p-1)\{\mcB_s^{(1)}(\partial_xY_s)+\mcD_s^{(1)}(\partial_xV_s)\}^2(\partial_xY_s)^{2p-2}]dr\\
		&\leq1+2p\{|\kappa|\delta^{-1}+(2p-1)\beta^2\delta^{-2}(1-2|\beta|)^{-1}\}\int_t^s\mbE[(\partial_xY_r)^{2p}]dr\\
		&+2(2p-1)\beta^2(1-2|\beta|)^{-1}\int_t^s\mbE[(\partial_xV_r)^{2p}+(p-1)(\partial_xY_r)^{2p}]dr
	\end{align*}
	By the Hölder's inequality, we have
	\begin{align*}
		&(\partial_xV_s)^{2p}=\left(\int_t^sS_{\delta}'(Y_r)(\partial_xY_r)dr\right)^{2p}\leq\delta^{-2p}\int_t^s(\partial_xY_r)^{2p}dr,
	\end{align*}
	and  
	$$\int_t^s(\partial_xV_r)^{2p}dr\leq\delta^{-2p}\int_t^s\int_t^r(\partial_xY_u)^{2p}dudr\leq\delta^{-2p}\int_t^s(\partial_xY_r)^{2p}dr$$
	we conclude that
	\begin{align*}
		&\mbE[(\partial_xY_s)^{2p}]\leq1+\left\{2p|\kappa|\delta^{-1}+\frac{2(2p-1)\beta^2}{(1-2|\beta|)}(p\delta^{-2}+\delta^{-2p}+p-1)\right\}\int_t^s\mbE[(\partial_xY_r)^{2p}]dr,
	\end{align*}
	and the Gronwall's Inequality implies that 
	$$\mbE[(\partial_xY_s)^{2p}]\leq\exp\left(2p|\kappa|\delta^{-1}+\frac{2(2p-1)\beta^2}{(1-2|\beta|)}(p\delta^{-2}+\delta^{-2p}+p-1)\right)\leq\mrO(\exp(C_{p,\kappa,\beta}^{(1)}\delta^{-2p}))$$
	for a constant $C_{p,\kappa,\beta}^{(1)}>0$ depending only on $p,\kappa,\beta$, so $\mbE[(\partial_xV_s)^{2p}]\leq\mrO(\delta^{-2p}\exp(C_{p,\kappa,\beta}^{(1)}\delta^{-2p}))$. Consequently,
	\begin{align*}
		&\left|\frac{\partial}{\partial x}f(x,v,t)\right|=\left|\mbE\left[\partial_y h(Y_1,V_1)(\partial_xY_1)+\partial_v h(Y_1,V_1)(\partial_xV_1)\right]\right|\\
		&\leq\mrO(C_{1,x,v}\{\mbE[(\partial_xY_s)^2]^{1/2}+\mbE[(\partial_xV_s)^2]^{1/2}\})\leq\mrO(\delta^{-1}\exp(C_{\kappa,\beta}\delta^{-2}))
	\end{align*}
	for a constant $C_{\kappa,\beta}>0$. For the higher order derivatives, we consider an induction procedure, note that
	\begin{align*}
		\frac{\partial^k}{\partial x^k}f(x,v,t)=\mbE[\partial_y h(Y_1,V_1)(\partial_x^kY_1)+\partial_v h(Y_1,V_1)(\partial_x^kV_1)+Q_k],
	\end{align*}
	where $Q_k$ contains derivatives $\frac{\partial^lY_s}{\partial x^l}$ and $\frac{\partial^lV_s}{\partial x^l}$ with orders $l<k$. By the induction assumption, their $2p$-th moments are bounded depending on $\delta^{-1}$. Since
	\begin{align*}
		d(\partial_x^kY_s)=\{\mcA_s^{(k)}(\partial_x^kY_s)+\mcU_s^{(k)}\}ds+\{\mcB_s^{(k)}(\partial_x^kY_s)+\mcD_s^{(k)}(\partial_x^kV_s)+\mcV_s^{(k)}\}dB_s,
	\end{align*}
	where the processes $\mcU_s^{(k)}$ and $\mcV_s^{(k)}$ contains derivatives $\frac{\partial^lY_s}{\partial x^l}$ and $\frac{\partial^lV_s}{\partial x^l}$ with orders $l<k$, we can apply the above arguments to conclude that
	\begin{align*}
		\mbE[(\partial_x^kY_s)^{2p}]\leq\mrO(\delta^{-4p(k-1)}\exp(C_{p,\kappa,\beta}^{(k)}\delta^{-2p})),\quad\mbE[(\partial_x^kV_s)^{2p}]\leq\mrO(\delta^{-2p(2k-1)}\exp(C_{p,\kappa,\beta}^{(k)}\delta^{-2p}))
	\end{align*}
	for a constant $C_{p,\kappa,\beta}^{(k)}>0$ depending only on $p,\kappa,\beta$, here we omit details for simplicity. Thus, we can obtain that
	\begin{align*}
		\left|\frac{\partial^k}{\partial x^k}f(x,v,t)\right|\leq\mrO(C_{k,x,v}\delta^{-(2k-1)}\exp(C_{k,\kappa,\beta}\delta^{-2})),
	\end{align*}
	which concludes the first claim. For the second claim relating to $\big|\frac{\partial^l}{\partial v^l}f(x,v,t)\big|,l=1,2$, we have 
	\begin{align*}
		\left|\frac{\partial^l}{\partial v^l}f(x,v,t)\right|\leq\mrO(C_{l,x,v}\delta^{-(2l-1)}\exp(C_{l,\kappa,\beta}\delta^{-2})),
	\end{align*}
	the proof arguments are the same as those for $\big|\frac{\partial^k}{\partial x^k}f(x,v,t)\big|$, thus omitted.
\end{proof}

\section{Proof of Theorem \texorpdfstring{\ref{Thm of strategic CLT}}{1}}\label{sec of proof norm test}
We first define a sequence auxiliary variables as below:
\begin{align}
	\left\{\mcS_t:=\frac{1}{T_2}\sum_{s=1}^t\theta_s:t=1,\cdots,T_2\right\},\quad\mcS_0=0,
\end{align}
where $\theta_t$ is defined in \eqref{Eq of control}. To conclude Theorem \ref{Thm of strategic CLT}, it suffices to show that $(\fM_{T_2},\mcS_{T_2})\overset{d}{\sim}(Y_1^{0,0,0},V_1^{0,0})$, i.e. for any deterministic $a_1,a_2>0$
$$\lim_{n\to\infty}\big|\mbP\big(\big|\mcM_{T_2}\big|\leq a_1,|\mcS_{T_2}|\leq a_2\big)-\mbP\big(\big|Y_1^{0,0,0}\big|\leq a_1,|V_1^{0,0}|\leq a_2\big)\big|=0.$$
Actually, note that 
$$\mbP\big(\big|\mcM_{T_2}\big|\leq a_1,|\mcS_{T_2}|\leq a_2\big)=\mbE\big[1\big\{\big|\mcM_{T_2}\big|\leq a_1\big\}1\big\{\big|\mcS_{T_2}\big|\leq a_2\big\}\big],$$
where the function $1_{|y|\leq a}$ is not smooth, so let's consider the approximator $\varphi_{a,\epsilon}(y)=\Phi((a-y)/\epsilon)-\Phi(-(a+y)/\epsilon)$ defined in \eqref{Eq of approximator}, and we have
\begin{align}
    \lim_{\epsilon\downarrow0}\sup_{y\in\mbR}|\varphi_{a,\epsilon}(y)-1_{|y|\leq a}|=0.\label{Eq of CDF approximation}
\end{align}
Therefore, to conclude Theorem \ref{Thm of strategic CLT}, it suffices to prove that
\begin{align*}
	\lim_{T\to\infty}\big|\mbE[g_{a_1,a_2,\epsilon}(\fM_{T_2},\mcS_{T_2})]-\mbE[g_{a_1,a_2,\epsilon}(Y_1^{0,0,0},V_1^{0,0})]\big|=0,
\end{align*}
where $g_{a_1,a_2,\epsilon}(y,v):=\varphi_{a_1,\epsilon}(y)\varphi_{a_2,\epsilon}(v)$. Moreover, as we mentioned in \S\ref{sec of bounded derivatives}, since both the drift and diffusion of the SDE in \eqref{Eq of SDE covariance} are non-continuous, we approximate $(Y_1^{0,0,0},V_1^{0,0})$ via the $(Y_{1,\delta}^{0,0,0},V_{1,\delta}^{0,0})$ defined in \eqref{Eq of smoothed SDE}, and Lemma \ref{Lem of SDE approximation} ensures that
\begin{align*}
	\lim_{\delta\downarrow0}\big|\mbE[g_{a_1,a_2,\epsilon}(Y_1^{0,0,0},V_1^{0,0})]-\mbE[g_{a_1,a_2,\epsilon}(Y_{1,\delta}^{0,0,0},V_{1,\delta}^{0,0})]\big|=0.
\end{align*}
Moreover, for the tuning parameter $\epsilon>0$ in the approximator $\varphi_{a,\epsilon}$ and $\delta>0$ in $(Y_{1,\delta}^{0,0,0},V_{1,\delta}^{0,0})$, we review them as functions of $T$ such that
\begin{align}
	\left\{\begin{array}{l}
		\lim_{T\to\infty}\epsilon_T=0, \\
		\lim_{T\to\infty}\delta_T=0,
	\end{array}\right.\quad\left\{\begin{array}{l}
	\lim_{T\to\infty}\log(\log(T))\epsilon_T=\infty, \\
	\lim_{T\to\infty}\log(\log(T))\delta_T=\infty,
	\end{array}\right.\quad\lim_{T\to\infty}\delta_T/\epsilon_T^2=0,\label{Eq of convergence rate of tuning parameters}
\end{align}
then we will prove that
\begin{align*}
	\lim_{T\to\infty}\big|\mbE[g_{a_1,a_2,\epsilon_T}(\fM_{T_2},\mcS_{T_2})]-\mbE[g_{a_1,a_2,\epsilon_T}(Y_{1,\delta_T}^{0,0,0},V_{1,\delta_T}^{0,0})]\big|=0.
\end{align*}
Precisely, we split the proof of Theorem \ref{Thm of strategic CLT} into two situations: $\lim_{T\to\infty}|\kappa_{T_2}|<\infty$ and $\lim_{T\to\infty}|\kappa_{T_2}|=\infty$.

\subsection{Situation 1: \texorpdfstring{$|\kappa_{T_2}|\leq\mrO(1)$}{ }}
Based on the Lindeberg’s principle, we have
\begin{align*}
	&\mbE[g_{a_1,a_2,\epsilon_T}(\fM_{T_2},\mcS_{T_2})]-\mbE[g_{a_1,a_2,\epsilon_T}(Y_{1,\delta_T}^{0,0,0},V_{1,\delta_T}^{0,0})]\\
	&=\sum_{t=1}^{T_2}\mbE\big[g_{a_1,a_2,\epsilon_T}\big(Y_{1,\delta_T}^{\fM_t,\mcS_t,t/T_2},V_{1,\delta_T}^{\mcS_t,t/T_2}\big)\big]-\mbE\big[g_{a_1,a_2,\epsilon_T}\big(Y_{1,\delta_T}^{\fM_{t-1},\mcS_{t-1},(t-1)/T_2},V_{1,\delta_T}^{\mcS_{t-1},(t-1)/T_2}\big)\big]\\
	&=\sum_{t=1}^{T_2}\underbrace{\mbE\big[g_{a_1,a_2,\epsilon_T}\big(Y_{1,\delta_T}^{\fM_t,\mcS_t,t/T_2},V_{1,\delta_T}^{\mcS_t,t/T_2}\big)-g_{a_1,a_2,\epsilon_T}\big(Y_{1,\delta_T}^{\fM_{t-1},\mcS_{t-1},t/T_2},V_{1,\delta_T}^{\mcS_{t-1},t/T_2}\big)\big]}_{\mbI_{1,t}:=}\\
	&+\sum_{t=1}^{T_2}\underbrace{\mbE\big[g_{a_1,a_2,\epsilon_T}\big(Y_{1,\delta_T}^{\fM_{t-1},\mcS_{t-1},t/T_2},V_{1,\delta_T}^{\mcS_{t-1},t/T_2}\big)-g_{a_1,a_2,\epsilon_T}\big(Y_{1,\delta_T}^{\fM_{t-1},\mcS_{t-1},(t-1)/T_2},V_{1,\delta_T}^{\mcS_{t-1},(t-1)/T_2}\big)\big]}_{\mbI_{2,t}:=}.
\end{align*}
In the following context, we will compute $\mbI_{1,t}$ and $\mbI_{2,t}$, respectively, and derive that
\begin{align*}
	\lim_{T\to\infty}\left|\sum_{t=1}^{T_2}\mbI_{1,t}+\mbI_{2,t}\right|=0.
\end{align*}

\subsubsection{Computation of \texorpdfstring{$\mbI_{1,t}$}{I1}}
For simplicity, we define a collection of $\sigma$-fields as below:
\begin{align}
	\mcF_t:=\sigma(\bbx_1,\cdots,\bbx_{t+T_1}),\quad t=0,1,\cdots,T_2.\label{Eq of sigma field}
\end{align}
Before presenting the calculation procedures of $\mbI_{1,t}$, we list the following auxiliary lemma:
\begin{lem}[Cumulant expansion]\label{Lem of cumulant expansion}
	For any real-valued random variable $\xi$ with $\mbE[|\xi|^{K+2}]<\infty$ $(K\in\mbN^+)$ and real-valued function $g(\cdot)$ with continuous and bounded $K+1$ derivatives, then
	\begin{align*}
		\mbE[\xi g(\xi)]=\sum_{l=0}^K\frac{\kappa_{l+1}}{l!}\mbE[g^{(l)}(\xi)]+\varepsilon_{(K+1)},
	\end{align*}
	where the remainder term $\varepsilon_{(K+1)}$ satisfies that
	\begin{align*}
		|\varepsilon_{(K+1)}|\leq C_K\sup_{x\in\mbR}\big|g^{(K+1)}(x)\big|\mbE[|\xi|^{K+2}],
	\end{align*}
	and $\kappa_l$ is the $l$-th cumulant of $\xi$ defined as follows:
	\begin{align*}
		\log\big(\mbE[e^{{\rm i}x\xi}]\big)=\sum_{l=1}^{\infty}\kappa_l\frac{({\rm i}x)^l}{l!},\quad x\in\mbR.
	\end{align*}
\end{lem}
For $\mbI_{1,t}$, since  
\begin{align}
	h(x,v,t)=h_{a_1,a_2,\epsilon_T}(x,v,t):=\mbE[g_{a_1,a_2,\epsilon_T}(Y_{1,\delta_T}^{x,v,t},V_{1,\delta_T}^{x,v})]\label{Eq of function h}
\end{align}
is a bounded smooth function of $(x,v)$. By the second-order Taylor's expansion with Lagrange remainder, we have
\begin{align}
	&\mbI_{1,t}=\mbE\big[h\big(\fM_t,\mcS_t,t/T_2\big)-h\big(\fM_{t-1},\mcS_{t-1},t/T_2\big)\big]\notag\\
	&=\frac{1}{\sqrt{T_2(\tau^2+\sigma^2)}}\mbE\big[\theta_t X_t\partial_x h\big(\fM_{t-1},\mcS_{t-1},t/T_2\big)\big]+\frac{1}{T_2}\mbE\big[\theta_t\partial_v h\big(\fM_{t-1},\mcS_{t-1},t/T_2\big)\big]\label{Eq of major 1}\\
	&+\frac{1}{2T_2(\tau^2+\sigma^2)}\mbE\big[X_t^2\partial_{xx}h\big(\fM_{t-1},\mcS_{t-1},t/T_2\big)\big]\label{Eq of major 3}\\
	&+\text{remainders}.\label{Eq of Major 2}
\end{align}
\noindent
{\bf Computation of \eqref{Eq of major 1}}: Recall the definition of $X_t$ in \eqref{Eq of estimators}, 
\begin{small}
\begin{align*}
	&\frac{\mbE[\partial_x h(\fM_{t-1},\mcS_{t-1},t/T_2)\theta_tX_t]}{\sqrt{T_2(\tau^2+\sigma^2)}}=\frac{\mbE[\partial_x h(\fM_{t-1},\mcS_{t-1},t/T_2)\theta_t\mbE[X_t|\mcF_{t-1}]]}{\sqrt{T_2(\tau^2+\sigma^2)}}\\
	&=\frac{1}{T_2}\Bigg(-\kappa_{T_2}\mbE[\partial_x h(\fM_{t-1},\mcS_{t-1},t/T_2)\theta_t]+\sqrt{\frac{T_2}{T_1^2(\tau^2+\sigma^2)}}\sum_{s=1}^{T_1}\mbE[\partial_x h(\fM_{t-1},\mcS_{t-1},t/T_2)\theta_t\bbmu'\bbGa\bby_s]\Bigg),
\end{align*}
\end{small}\noindent
where $\mcF_{t-1}$ is defined in \eqref{Eq of sigma field}. By the cumulant expansion, we have
\begin{align*}
	&\sqrt{\frac{T_2}{T_1^2(\tau^2+\sigma^2)}}\sum_{s=1}^{T_1}\mbE[\partial_x h(\fM_{t-1},\mcS_{t-1},t/T_2)\theta_t\bbmu'\bbGa\bby_s]\\
	&=\sqrt{\frac{T_2}{T_1^2(\tau^2+\sigma^2)}}\sum_{s=1}^{T_1}\sum_{j=1}^m\mbE[\theta_t\bbmu'\Gamma_{\cdot j}\partial_{y_{j,s}}\{\partial_x h(\fM_{t-1},\mcS_{t-1},t/T_2)\}]+\varepsilon_{(2)},
\end{align*}
where $\varepsilon_{(2)}$ is defined as in Lemma \ref{Lem of cumulant expansion} and
\begin{align*}
	&\partial_{y_{j,s}}\{\partial_x h(\fM_{t-1},\mcS_{t-1},t/T_2)\}\\
	&=\partial_{xx} h(\fM_{t-1},\mcS_{t-1},t/T_2)\times\frac{1}{T_1\sqrt{T_2(\tau^2+\sigma^2)}}\sum_{r=1}^{t-1}\theta_r(\bbmu+\bbGa\bby_{r+T_1})'\Gamma_{\cdot j},
\end{align*}
so we obtain that
\begin{align*}
	&\sqrt{\frac{T_2}{T_1^2(\tau^2+\sigma^2)}}\sum_{s=1}^{T_1}\mbE[\partial_x h(\fM_{t-1},\mcS_{t-1},t/T_2)\theta_t\bbmu'\bbGa\bby_s]\\
	&=\underbrace{\frac{T_2\bbmu'\bbSig\bbmu}{T_1(\tau^2+\sigma^2)}}_{=\beta_{T_2}}\mbE[\mcS_{t-1}\theta_t\partial_{xx} h(\fM_{t-1},\mcS_{t-1},t/T_2)]\\
	&+\frac{1}{T_1(\tau^2+\sigma^2)}\sum_{r=1}^{t-1}\mbE[\theta_t\theta_r\bbmu'\bbSig\bbGa\bby_{r+T_1}\partial_{xx} h(\fM_{t-1},\mcS_{t-1},t/T_2)]+\varepsilon_{(2)}.
\end{align*}
Since 
\begin{align*}
	&\mbE\left[\left(\frac{1}{T_1}\sum_{r=1}^{t-1}\theta_r\bbmu'\bbSig\bbGa\bby_{r+T_1}\right)^2\right]=\frac{1}{T_1^2}\sum_{r_1,r_2=1}^{t-1}\mbE[\theta_{r_1}\theta_{r_2}(\bbmu'\bbSig\bbGa\bby_{r_1+T_1})(\bbmu'\bbSig\bbGa\bby_{r_2+T_1})]\\
	&=\frac{1}{T_1^2}\sum_{r=1}^{t-1}\mbE[(\bbmu'\bbSig\bbGa\bby_{r+T_1})^2]+\frac{2}{T_1^2}\sum_{1\leq r_1<r_2\leq t-1}\mbE[\theta_{r_1}\theta_{r_2}(\bbmu'\bbSig\bbGa\bby_{r_1+T_1})\mbE[\bbmu'\bbSig\bbGa\bby_{r_2+T_1}|\theta_{r_2}]]\\
	&=\frac{1}{T_1^2}\sum_{r=1}^{t-1}\mbE[(\bbmu'\bbSig\bbGa\bby_{r+T_1})^2]=\frac{(t-1)\bbmu'\bbSig^3\bbmu}{T_1^2}\leq\mrO(T_2^{-1}),
\end{align*}
combining with the fact that 
\begin{align*}
	\sup_{x\in\mbR,|v|\leq t,t\in[0,1]}|\partial_{xx} h(x,v,t)|\leq\mrO\big(\tilde{C}_2\delta_T^{-3}\exp\big(C_{\kappa_{T_2},\beta_{T_2}}\delta_T^{-2}\big)\big)
\end{align*}
by Proposition \ref{Pro of bounded derivatives}, where 
\begin{align}
	\tilde{C}_2=\max_{1\leq m\leq 2}\left\{\sup_{(x,v)\in\mbR^2}\left|\frac{\partial^m }{\partial x^m}g_{a_1,a_2,\epsilon_T}(x,v)\right|,\sup_{(x,v)\in\mbR^2}\left|\frac{\partial^m }{\partial v^m}g_{a_1,a_2,\epsilon_T}(x,v)\right|\right\}\leq\mrO(\epsilon_T^{-2}),\label{Eq of bound of derivatives of g}
\end{align}
the Cauchy's inequality implies that
\begin{align*}
	\left|\frac{1}{T_1^2(\tau^2+\sigma^2)}\sum_{r=1}^{t-1}\sum_{s=1}^{T_1}\mbE[\theta_r\theta_t\bbmu'\bbSig\bbGa\bby_{r+T_1}\partial_{xx} h(\fM_{t-1},\mcS_{t-1},t/T_2)]\right|\leq\mro(1),
\end{align*}
where we use the fact that $\epsilon_T^{-1},\delta_T^{-1}\leq\mro(\log(\log(T)))$ in \eqref{Eq of convergence rate of tuning parameters}. Next, we show that $|\varepsilon_{(2)}|=\mro(1)$. By Lemma \ref{Lem of cumulant expansion}, we have
\begin{align*}
	\varepsilon_{(2)}=\sqrt{\frac{T_2}{T_1^2(\tau^2+\sigma^2)}}\sum_{s=1}^{T_1}\sum_{j=1}^m\mbE[\theta_t\bbmu'\Gamma_{\cdot j}\partial_{y_{j,s}}^2\{\partial_x h(\fM_{t-1},\mcS_{t-1},t/T_2)\}],
\end{align*}
where
\begin{align*}
	&\partial_{y_{j,s}}^2\{\partial_x h(\fM_{t-1},\mcS_{t-1},t/T_2)\}\\
	&=\partial_{xxx} h(\fM_{t-1},\mcS_{t-1},t/T_2)\times\left(\frac{1}{T_1\sqrt{T_2(\tau^2+\sigma^2)}}\sum_{r=1}^{t-1}\theta_r(\bbmu+\bbGa\bby_{r+T_1})'\Gamma_{\cdot j}\right)^2.
\end{align*}
so
\begin{align*}
	&\varepsilon_{(2)}=\frac{T_2^{3/2}\bbmu'\bbGa\diag(\bbGa'\bbmu)\bbGa'\bbmu}{T_1^2(\tau^2+\sigma^2)^{3/2}}\mbE[\theta_t\mcS_{t-1}^2\partial_{xxx} h(\fM_{t-1},\mcS_{t-1},t/T_2)]\\
	&+\frac{2T_2^{1/2}}{T_1^2(\tau^2+\sigma^2)^{3/2}}\sum_{r=1}^{t-1}\mbE[\theta_t\mcS_{t-1}\theta_r\bbmu'\bbGa\diag(\bbGa'\bbmu)\bbGa\bby_{r+T_1}\partial_{xxx} h(\fM_{t-1},\mcS_{t-1},t/T_2)]\\
	&+\frac{1}{T_1^2T_2^{1/2}(\tau^2+\sigma^2)^{3/2}}\sum_{r_1,r_2=1}^{t-1}\mbE[\theta_t\theta_{r_1}\theta_{r_2}\bby_{r_2+T_1}'\bbGa'\diag(\bbGa'\bbmu)\bbGa\bby_{r_1+T_1}\partial_{xxx} h(\fM_{t-1},\mcS_{t-1},t/T_2)].
\end{align*}
Based on Assumption \ref{Ap of noise}, Proposition \ref{Pro of bounded derivatives} and the fact that $|S_{t-1}|\leq1$, we can conclude that
\begin{align*}
	&\left|\frac{T_2^{3/2}\bbmu'\bbGa\diag(\bbGa'\bbmu)\bbGa'\bbmu}{T_1^2(\tau^2+\sigma^2)^{3/2}}\mbE[\theta_t\mcS_{t-1}^2\partial_{xxx} h(\fM_{t-1},\mcS_{t-1},t/T_2)]\right|\leq\mro(1).
\end{align*}
Similarly, we have
\begin{align*}
	\left|\frac{2T_2^{1/2}}{T_1^2(\tau^2+\sigma^2)^{3/2}}\sum_{r=1}^{t-1}\mbE[\theta_t\mcS_{t-1}\theta_r\bbmu'\bbGa\diag(\bbGa'\bbmu)\bbGa\bby_{r+T_1}\partial_{xxx} h(\fM_{t-1},\mcS_{t-1},t/T_2)]\right|\leq\mro(1).
\end{align*}
Finally, note that for $1\leq r_1<r_2\leq t-1$ and $1\leq r_3<r_4\leq t-1$, if $\{r_1,r_2\}\neq\{r_3,r_4\}$, we have
\begin{align*}
	\mbE\big[\theta_{r_1}\theta_{r_2}\theta_{r_3}\theta_{r_4}(\bby_{r_2+T_1}'\bbGa'\diag(\bbGa'\bbmu)\bbGa\bby_{r_1+T_1})(\bby_{r_4+T_1}'\bbGa'\diag(\bbGa'\bbmu)\bbGa\bby_{r_3+T_1})\big]=0,
\end{align*}
then
\begin{align*}
	&\mbE\left[\left(\frac{1}{T_1^2}\sum_{1\leq r_1<r_2\leq t-1}\theta_{r_1}\theta_{r_2}\bby_{r_2+T_1}'\bbGa'\diag(\bbGa'\bbmu)\bbGa\bby_{r_1+T_1}\right)^2\right]\\
	&=\frac{1}{T_1^4}\sum_{1\leq r_1<r_2\leq t-1}\mbE\big[\big(\bby_{r_2+T_1}'\bbGa'\diag(\bbGa'\bbmu)\bbGa\bby_{r_1+T_1}\big)^2\big]\\
	&=\frac{(t-1)(t-2)}{2T_1^4}\bbmu'\bbGa(\bbSig\odot\bbSig)\bbGa\bbmu'\leq\mrO(T^{-2}).
\end{align*}
Moreover, since
\begin{align*}
	&T_1^{-2}\Var\big(\bby_{r+T_1}'\bbGa'\diag(\bbGa'\bbmu)\bbGa\bby_{r+T_1}\big)\leq\mrO\big(T_1^{-2}(\bbmu'\diag(\bbSig)\boldsymbol{1}_n)^2\big)\leq\mrO(T^{-1}),
\end{align*}
it gives that
\begin{align*}
	\left|\frac{1}{T_1^2T_2^{1/2}(\tau^2+\sigma^2)^{3/2}}\sum_{r_1,r_2=1}^{t-1}\mbE[\theta_t\theta_{r_1}\theta_{r_2}\bby_{r_2+T_1}'\bbGa'\diag(\bbGa'\bbmu)\bbGa\bby_{r_1+T_1}\partial_{xxx} h(\fM_{t-1},\mcS_{t-1},t/T_2)]\right|\leq\mro(1),
\end{align*}
which implies that $|\varepsilon_{(2)}|\leq\mro(1)$. In summary, we obtain that
\begin{align*}
	&\eqref{Eq of major 1}=\frac{1}{T_2}\mbE\big[\theta_t\partial_v h\big(\fM_{t-1},\mcS_{t-1},t/T_2\big)\big]\\
	&+\frac{1}{T_2}\big(-\kappa_{T_2}\mbE[\partial_x h(\fM_{t-1},\mcS_{t-1},t/T_2)\theta_t]+\beta_{T_2}\mbE[\mcS_{t-1}\theta_t\partial_{xx} h(\fM_{t-1},\mcS_{t-1},t/T_2)]\big)+\mro(T_2^{-1}).
\end{align*}
{\bf Computation of \eqref{Eq of major 3}}: Note that
\begin{align*}
	\mbE[X_t^2|\mcF_{t-1}]=\frac{1}{T_1^2}\sum_{s_1,s_2=1}^{T_1}\mbE[\bbx_{s_1}'\bbx_{t+T_1}\bbx_{t+T_1}'\bbx_{s_2}|\mcF_{t-1}]-\frac{2d_0^2}{T_1}\sum_{s=1}^{T_1}\bbx_s'\bbmu+d_0^4,
\end{align*}
where
\begin{align*}
	&\mbE[\bbx_{s_1}'\bbx_{t+T_1}\bbx_{t+T_1}'\bbx_{s_2}|\mcF_{t-1}]=\bbx_{s_1}'\bbmu\bbmu'\bbx_{s_2}+\mbE[\bbx_{s_1}'\bbGa\bby_{t+T_1}\bby_{t+T_1}'\bbGa'\bbx_{s_2}|\mcF_{t-1}]\\
	&=\bbx_{s_1}'\bbmu\bbmu'\bbx_{s_2}+\bbx_{s_1}'\bbSig\bbx_{s_2}=\Vert\bbmu\Vert_2^4+\Vert\bbmu\Vert_2^2\bbmu'\bbGa(\bby_{s_1}+\bby_{s_2})+\bby_{s_1}'\bbGa'\bbmu\bbmu'\bbGa\bby_{s_2}\\
	&+\bbmu'\bbSig\bbmu+\bbmu'\bbSig\bbGa(\bby_{s_1}+\bby_{s_2})+\bby_{s_1}'\bbGa'\bbSig\bbGa\bby_{s_2}.
\end{align*}
Thus, we obtain that
\begin{align*}
	&\mbE[X_t^2|\mcF_{t-1}]=\tau^2+\bbmu'\bbSig\bbmu+2(\tau\bbmu'+\bbmu'\bbSig)\bbGa\fy_{T_1}+\fy_{T_1}'(\bbGa'\bbmu\bbmu'\bbGa+\bbGa'\bbSig\bbGa)\fy_{T_1},
\end{align*}
where $\fy_{T_1}=T_1^{-1}\sum_{s=1}^{T_1}\bby_s$, and
\begin{align*}
	&\eqref{Eq of major 3}=\frac{1}{2T_2(\tau^2+\sigma^2)}\mbE\big[\mbE[X_t^2|\mcF_{t-1}]\partial_{xx}h(\fM_{t-1},\mcS_{t-1},t/T_2)\big]\\
	&=\frac{\tau^2+\bbmu'\bbSig\bbmu}{2T_2(\tau^2+\sigma^2)}\mbE[\partial_{xx}h(\fM_{t-1},\mcS_{t-1},t/T_2)]\\
	&+\frac{1}{2T_2(\tau^2+\sigma^2)}\mbE[\fy_{T_1}'(\bbGa'\bbmu\bbmu'\bbGa+\bbGa'\bbSig\bbGa)\fy_{T_1}\partial_{xx}h(\fM_{t-1},\mcS_{t-1},t/T_2)]\\
	&+\frac{1}{T_2(\tau^2+\sigma^2)}\mbE[(\tau\bbmu'+\bbmu'\bbSig)\bbGa\fy_{T_1}\partial_{xx}h(\fM_{t-1},\mcS_{t-1},t/T_2)].
\end{align*}
Since
\begin{align*}
	\mbE[\fy_{T_1}'(\bbGa'\bbmu\bbmu'\bbGa+\bbGa'\bbSig\bbGa)\fy_{T_1}]=\frac{\|\bbmu\|_2^2}{T_1}+\frac{1}{T_1}\tr(\bbSig^2),
\end{align*}
by \eqref{Eq of Variance trick} and Assumption \ref{Ap of noise}, we know that
\begin{align*}
	&\Var\big(\fy_{T_1}'(\bbGa'\bbmu\bbmu'\bbGa+\bbGa'\bbSig\bbGa)\fy_{T_1}\big),\Var((\tau\bbmu'+\bbmu'\bbSig)\bbGa\fy_{T_1})\leq\mrO(T^{-1}).
\end{align*}
Combining with the fact that $|\partial_{xx}h(x,v,t)|$ is uniformly bounded in $t\in[0,1]$ by Proposition \ref{Pro of bounded derivatives}, we derive that
\begin{align*}
	&\eqref{Eq of major 3}=\frac{1}{2T_2}\mbE[\partial_{xx}h(\fM_{t-1},\mcS_{t-1},t/T_2)]+\mro(T_2^{-1}).
\end{align*}
{\bf Computation of \eqref{Eq of Major 2}}: The explicit form of remainders are given as below:
\begin{align*}
	&\text{remainders}=\frac{1}{2T_2^2}\mbE[\partial_{vv}h(\xi_{1,t},\xi_{2,t},t/T_2)]+\frac{1}{T_2^{3/2}(\tau^2+\sigma^2)^{1/2}}\mbE[X_t\partial_{xv}h(\xi_{1,t},\xi_{2,t},t/T_2)]\\
	&+\frac{1}{2T_2(\tau^2+\sigma^2)}\mbE[X_t^2\{\partial_{xx}h(\xi_{1,t},\xi_{2,t},t/T_2)-\partial_{xx}h(\fM_{t-1},\mcS_{t-1},t/T_2)\}],
\end{align*}
where $\xi_{1,t}\in[\min\{\fM_{t-1},\fM_t\},\max\{\fM_{t-1},\fM_t\}]$ and $\xi_{2,t}\in[\min\{\mcS_{t-1},\mcS_t\},\max\{\mcS_{t-1},\mcS_t\}]$. By Proposition \ref{Pro of bounded derivatives}, $\mbE[|\partial_{vv}h(\xi_{1,t},\xi_{2,t},t/T_2)|],\mbE[|\partial_{xv}h(\xi_{1,t},\xi_{2,t},t/T_2)|]\leq\mrO(\epsilon_T^{-2}\delta_T^{-3}e^{C_{\kappa_{T_2},\beta_{T_2}}\delta_T^{-2}})$, combining with $\mbE[X_t],\Var(X_t)\leq\mrO(1)$ and \eqref{Eq of convergence rate of tuning parameters}, we have
\begin{align*}
	&\text{remainders}=\frac{1}{2T_2(\tau^2+\sigma^2)}\mbE[X_t^2\{\partial_{xx}h(\xi_{1,t},\xi_{2,t},t/T_2)-\partial_{xx}h(\fM_{t-1},\mcS_{t-1},t/T_2)\}]+\mro(T_2^{-1}).
\end{align*}
Moreover, by Proposition \ref{Pro of bounded derivatives}, we have $|\partial_{xx}h|,|\partial_{xxx}h|\leq\mrO(\epsilon_T^{-3}\delta_T^{-5}e^{C_{\kappa_{T_2},\beta_{T_2}}\delta_T^{-2}})\leq\mrO(\log(T))$ uniformly in $x,\in\mbR,|v|\leq t,t\in[0,1]$, then
\begin{align*}
	&X_t^2\big|\partial_{xx}h(\xi_{1,t},\xi_{2,t},t/T_2)-\partial_{xx}h(\fM_{t-1},\mcS_{t-1},t/T_2)\big|\leq\frac{L_1\eta1_{|X_t|\leq\eta X_t^2}}{\sqrt{T_2(\tau^2+\sigma^2)}}+L_2X_t^21_{|X_t|>\eta}+\frac{L_3X_t^2}{T_2},
\end{align*}
where $L_1,L_2,L_3>0$ such that $\max\{L_1,L_2,L_3\}\leq\mrO(\log(T))$, $\eta=T_2^{1/6}$ and we use the fact that $|\xi_{2,t}-\mcS_{t-1}|\leq T_2^{-1}$ above. According to Lemma \ref{Lem of Consistence}, we have $\mbE[X_t^4]\leq\mrO(1)$, then by Chebyshev inequality,
\begin{align*}
	&T_2^{-1}\mbE[X_t^21_{|X_t|>\eta}]\leq T_2^{-1}\mbE[X_t^4]^{1/2}\mbP(|X_t|>\eta)\leq\mrO(T_2^{-1}\eta^{-2}\mbE[X_t^2])\leq\mrO(T_2^{-1}\eta^{-2})\leq\mrO(T_2^{-4/3}).
\end{align*}
Combining with
$$\eta T_2^{-3/2}\mbE[X_t^21_{|X_t|\leq\eta}]\leq\mrO(\eta T_2^{-1}\mbE[X_t^2])\leq\mrO(T_2^{-4/3}),$$
we conclude that
\begin{align*}
	&|\text{remainders}|\leq\mrO((L_1+L_2+L_3)T_2^{-4/3})\leq\mro(T_2^{-1}).
\end{align*}
In summary, we obtain that
\begin{align*}
	&\mbI_{1,t}=\frac{1}{T_2}\mbE[\sign(\fM_{t-1})\{\kappa_{T_2}\partial_x h(\fM_{t-1},\mcS_{t-1},t/T_2)-\partial_v h(\fM_{t-1},\mcS_{t-1},t/T_2)\}]\\
	&+\frac{1}{2T_2}\mbE[(1-2\beta_{T_2}\sign(\fM_{t-1})\mcS_{t-1})\partial_{xx}h(\fM_{t-1},\mcS_{t-1},t/T_2)]+\mro(T_2^{-1}).
\end{align*}
and
\begin{align}
	\sum_{t=1}^{T_2}\mbI_{1,t}&=\frac{1}{T_2}\sum_{t=1}^{T_2}\mbE\left[\sign(\fM_{t-1})\{\kappa_{T_2}\partial_x h(\fM_{t-1},\mcS_{t-1},t/T_2)-\partial_v h(\fM_{t-1},\mcS_{t-1},t/T_2)\}\right]\notag\\
	&+\frac{1}{2T_2}\sum_{t=1}^{T_2}\mbE\big[\big(1-2\beta_{T_2}\sign(\fM_{t-1})\mcS_{t-1}\big)\partial_{xx} h(\fM_{t-1},\mcS_{t-1},t/T_2)\big]+\mro(1).\label{Eq of mbI 1}
\end{align}

\subsubsection{Computation of \texorpdfstring{$\mbI_{2,t}$}{I2}}
By the Markov property of $(Y_{s,\delta_T}^{x,v,t},V_{s,\delta_T}^{v,t})$ and Ito's formula, we have
\begin{small}
	\begin{align*}
		&h(x,v,(t-1)/T_2)-h(x,v,t/T_2)\\
		&=\mbE\big[h\big(Y_{t/T_2}^{x,v,(t-1)/T_2},V_{t/T_2}^{v,(t-1)/T_2},t/T_2\big)\big]-h(x,v,t/T_2)\\
		&=\mbE\left[\int_{(t-1)/T_2}^{t/T_2}\kappa_{T_2}S_{\delta_T}\big(Y_{s,\delta_T}^{x,v,(t-1)/T_2}\big)\partial_x h\big(Y_{s,\delta_T}^{x,v,(t-1)/T_2},V_{s,\delta_T}^{v,(t-1)/T_2},t/T_2\big)ds\right]\\
		&-\mbE\left[\int_{(t-1)/T_2}^{t/T_2}S_{\delta_T}\big(Y_{s,\delta_T}^{x,v,(t-1)/T_2}\big)\partial_v h\big(Y_{s,\delta_T}^{x,v,(t-1)/T_2},V_{s,\delta_T}^{v,(t-1)/T_2},t/T_2\big)ds\right]\\
		&+\frac{1}{2}\mbE\left[\int_{(t-1)/T_2}^{t/T_2}\big(1-2\beta S_{\delta_T}\big(Y_{s,\delta_T}^{x,v,(t-1)/T_2}\big)V_{s,\delta_T}^{v,(t-1)/T_2}\big)\partial_{xx} h\big(Y_{s,\delta_T}^{x,v,(t-1)/T_2},V_{s,\delta_T}^{v,(t-1)/T_2},t/T_2\big)ds\right].
	\end{align*}
\end{small}\noindent
By direct calculation, we have
\begin{align}
	\sup_{s\in[(t-1)/T_2,t/T_2]}\big|Y_{s,\delta_T}^{x,v,(t-1)/T_2}-x\big|\leq\frac{\kappa_{T_2}}{T_2}+(1+2\beta_{T_2})\sup_{s\in[(t-1)/T_2,t/T_2]}|B_s-B_{(t-1)/T_2}|,\label{Eq of Ys increase}
\end{align}
and
\begin{align*}
	\sup_{s\in[(t-1)/T_2,t/T_2]}\big|V_{s,\delta_T}^{v,(t-1)/T_2}-v\big|\leq\frac{1}{T_2},
\end{align*}
then by Proposition \ref{Pro of bounded derivatives}, we have
\begin{align*}
	&\mbE\left[\int_{(t-1)/T_2}^{t/T_2}\big|\partial_x h\big(Y_{s,\delta_T}^{x,v,(t-1)/T_2},V_{s,\delta_T}^{v,(t-1)/T_2},t/T_2\big)-\partial_x h(x,v,t/T_2)\big|ds\right]\\
	&\leq\mrO\left(\epsilon_T^{-2}\delta_T^{-3}e^{C_{\kappa_{T_2},\beta_{T_2}}\delta_T^{-2}}\mbE\left[\int_{(t-1)/T_2}^{t/T_2}\sup_{s\in[(t-1)/T_2,t/T_2]}\big(\big|Y_{s,\delta_T}^{x,v,(t-1)/T_2}-x\big|+\big|V_{s,\delta_T}^{v,(t-1)/T_2}-v\big|\big)ds\right]\right)\\
	&\leq\mrO\left(\frac{\log(T)}{T_2^2}+\log(T)\mbE\left[\int_{(t-1)/T_2}^{t/T_2}\sup_{s\in[(t-1)/T_2,t/T_2]}\big|B_s-B_{(t-1)/T_2}\big|ds\right]\right)\leq\mrO(T_2^{-3/2}),
\end{align*}
and we further derive that
\begin{align}
	&h(x,v,(t-1)/T_2)-h(x,v,t/T_2)\label{Eq of I2 estimation}\\
	&=\{\kappa_{T_2}\partial_x h(x,v,t/T_2)-\partial_v h(x,v,t/T_2)\}\mbE\left[\int_{(t-1)/T_2}^{t/T_2}S_{\delta_T}\big(Y_{s,\delta_T}^{x,v,(t-1)/T_2}\big)ds\right]\notag\\
	&+\frac{\partial_{xx} h(x,v,t/T_2)}{2}\mbE\left[\int_{(t-1)/T_2}^{t/T_2}\big(1-2\beta S_{\delta_T}\big(Y_{s,\delta_T}^{x,v,(t-1)/T_2}\big)v\big)ds\right].\notag
\end{align}
Next, by \eqref{Eq of Ys increase}, we have
\begin{align*}
	\mbP\left(\sup_{s\in[(t-1)/T_2,t/T_2]}\big|Y_{s,\delta_T}^{x,v,(t-1)/T_2}-x\big|>\frac{\log(T_2)}{\sqrt{T_2}}\right)\leq\mro(1),
\end{align*}
so
\begin{align*}
	\mbE\left[1_{|x|>\delta_T}\int_{(t-1)/T_2}^{t/T_2}\big|S_{\delta_T}\big(Y_{s,\delta_T}^{x,v,(t-1)/T_2}\big)-\sign(x)\big|ds\right]\leq\mro(T_2^{-1}),
\end{align*}
where we use the fact that $S_{\delta_T}(x)=\sign(x)$ for $|X|\geq\delta_T$ and \eqref{Eq of convergence rate of tuning parameters}. Consequently, we obtain that
\begin{align*}
	&\mbE\big[1_{|\fM_{t-1}|>\delta_T}\big(h(\fM_{t-1},\mcS_{t-1},(t-1)/T_2)-h(\fM_{t-1},\mcS_{t-1},t/T_2)\big)\big]\\
	&=\frac{1}{T_2}\mbE\left[1_{|\fM_{t-1}|>\delta_T}\sign(\fM_{t-1})\{\kappa_{T_2}\partial_x h(\fM_{t-1},\mcS_{t-1},t/T_2)-\partial_v h(\fM_{t-1},\mcS_{t-1},t/T_2)\}\right]\\
	&+\frac{1}{2T_2}\mbE\big[1_{|\fM_{t-1}|>\delta_T}\big(1-2\beta_{T_2}\sign(\fM_{t-1})\mcS_{t-1}\big)\partial_{xx} h(\fM_{t-1},\mcS_{t-1},t/T_2)\big]+\mro(T_2^{-1}).
\end{align*}
Moreover, by Lemma \ref{Lem of partial derivative of time}, we have
\begin{align*}
	&\big|h(x,v,(t-1)/T_2)-h(x,v,t/T_2)\big|1_{|x|\leq\delta_T}\leq\mrO(1_{|x|\leq\delta_T}\epsilon_T^{-2}T_2^{-1}),
\end{align*}
so
\begin{align*}
	&\big|\mbE\big[1_{|\fM_{t-1}|\leq\delta_T}\big(h(\fM_{t-1},\mcS_{t-1},(t-1)/T_2)-h(\fM_{t-1},\mcS_{t-1},t/T_2)\big)\big]\big|\leq\mrO\big(\epsilon_T^{-2}T_2^{-1}\mbP(|\fM_{t-1}|\leq\delta_T)\big).
\end{align*}
Thus, we conclude that
\begin{align*}
	&\sum_{t=1}^{T_2}\mbI_{2,t}=\sum_{t=1}^{T_2}\mbE[h(\fM_{t-1},\mcS_{t-1},t/T_2)-h(\fM_{t-1},\mcS_{t-1},(t-1)/T_2)]\\
	&=-\frac{1}{T_2}\mbE\left[\sign(\fM_{t-1})\{\kappa_{T_2}\partial_x h(\fM_{t-1},\mcS_{t-1},t/T_2)-\partial_v h(\fM_{t-1},\mcS_{t-1},t/T_2)\}\right]\\
	&-\frac{1}{2T_2}\mbE\big[\big(1-2\beta_{T_2}\sign(\fM_{t-1})\mcS_{t-1}\big)\partial_{xx} h(\fM_{t-1},\mcS_{t-1},t/T_2)\big]\\
	&+\mrO\left(\frac{1}{\epsilon_T^2T_2}\sum_{t=1}^{T_2}\mbP(|\fM_{t-1}|\leq\delta_T)\right)+\mro(1).
\end{align*}
By Lemma \ref{Lem of minor probability} and \eqref{Eq of convergence rate of tuning parameters}, since $\lim_{T\to\infty}\sqrt{T}\delta_T^2=\infty$, we have   
$$\frac{1}{T_2}\sum_{t=1}^{T_2}\mbP(|\fM_{t-1}|\leq\delta_T)\leq\mrO(\delta_T),$$
combining with \eqref{Eq of mbI 1} and \eqref{Eq of convergence rate of tuning parameters}, it yields that
\begin{align*}
	&\left|\sum_{t=1}^{T_2}\mbI_{1,t}+\mbI_{2,t}\right|\leq\mrO\left(\frac{\delta_T}{\epsilon_T^2T_2}\right)+\mro(1)\leq\mro(1),
\end{align*}
which completes the proof Theorem \ref{Thm of strategic CLT} when $|\kappa_{T_2}|\leq\mrO(1)$.

\begin{lem}\label{Lem of partial derivative of time}
	For $h(x,v,t)$ defined in \eqref{Eq of function h} and any sufficiently small $r>0$, we have
	\begin{align*}
		\sup_{x\in\mbR,|v|\leq t,t\in[0,1-r]}|h(x,v,t+r)-h(x,v,t)|\leq\mrO(\epsilon_T^{-2}r).
	\end{align*}
\end{lem}
\begin{proof}
	Since $(Y_{1,\delta_T}^{x,v,t},V_{1,\delta_T}^{v,t})\overset{d}{\sim}(Y_{1-t,\delta_T}^{x,v,0},V_{1-t,\delta_T}^{v,0})$, then
	\begin{align*}
		h(x,v,t)=\mbE\big[g_{a_1,a_2,\epsilon_T}(Y_{1,\delta_T}^{x,v,t},V_{1,\delta_T}^{v,t})\big]=\mbE\big[g_{a_1,a_2,\epsilon_T}(Y_{1-t,\delta_T}^{x,v,0},V_{1-t,\delta_T}^{v,0})\big].
	\end{align*}
	For simplicity, denote $g=g_{a_1,a_2,\epsilon_T}$ and
	\begin{align*}
		\tilde{y}_{t,r}&:=\int_{1-t-r}^{1-t}\kappa_{T_2} S_{\delta_T}(Y_{s,\delta_T}^{x,v,0})ds+\int_{1-t-r}^{1-t}\big(1-2\beta_{T_2} S_{\delta_T}(Y_{s,\delta_T}^{x,v,0})V_{s,\delta_T}^{v,0}\big)^{1/2}dB_s,\\
		\tilde{v}_{t,r}&:=-\int_{1-t-r}^{1-t}S_{\delta_T}(Y_{s,\delta_T}^{x,v,0})ds,
	\end{align*}
	then by the Taylor expansion with Lagrange remainders,
	\begin{align*}
		&h(x,v,t+r)-h(x,v,t)\\
		&=\mbE\big[g(Y_{1-t-r,\delta_T}^{x,v,0}+\tilde{y}_{t,r},V_{1-t-r,\delta_T}^{v,0}+\tilde{v}_{t,r})-g(Y_{1-t-r,\delta_T}^{x,v,0},V_{1-t-r,\delta_T}^{v,0})\big]\\
		&=\mbE\big[\partial_y g(Y_{1-t-r,\delta_T}^{x,v,0},V_{1-t-r,\delta_T}^{v,0})\tilde{y}_{t,r}+\partial_v g(Y_{1-t-r,\delta_T}^{x,v,0},V_{1-t-r,\delta_T}^{v,0})\tilde{v}_{t,r}\big]\\
		&+\mbE\big[\partial_{yy} g(\xi_{1,t},\xi_{2,t})\tilde{y}_{t,r}^2+\partial_{vv} g(\xi_{1,t},\xi_{2,t})\tilde{v}_{t,r}^2+2\partial_{yv} g(\xi_{1,t},\xi_{2,t})\tilde{y}_{t,r}\tilde{v}_{t,r}\big],
	\end{align*}
	where
	\begin{align*}
		\left\{\begin{array}{l}
			\xi_{1,t}\in[\min\{Y_{1-t-r,\delta_T}^{x,v,0},Y_{1-t,\delta_T}^{x,v,0}\},\max\{Y_{1-t-r,\delta_T}^{x,v,0},Y_{1-t,\delta_T}^{x,v,0}\}],\\
			\xi_{2,t}\in[\min\{V_{1-t-r,\delta_T}^{v,0},V_{1-t,\delta_T}^{v,0}\},\max\{V_{1-t-r,\delta_T}^{v,0},V_{1-t,\delta_T}^{v,0}\}].
		\end{array}\right.
	\end{align*}
	By the Markov property of $(Y_{1-t,\delta_T}^{x,v,0},V_{1-t,\delta_T}^{v,0})$ and \eqref{Eq of bound of derivatives of g}, we have
	\begin{align*}
		&\big|\mbE\big[\partial_y g(Y_{1-t-r,\delta_T}^{x,v,0},V_{1-t-r,\delta_T}^{v,0})\tilde{y}_{t,r}\big]\big|\\
		&=\left|\mbE\left[\partial_y g(Y_{1-t-r,\delta_T}^{x,v,0},V_{1-t-r,\delta_T}^{v,0})\int_{1-t-r}^{1-t}\kappa_{T_2}S_{\delta_T}(Y_{s,\delta_T}^{x,v,0})ds\right]\right|\leq\mrO(\epsilon_T^{-1}r).
	\end{align*}
	Next, by direct calcualtions, we have
	\begin{align*}
		&\big|\mbE\big[\partial_{yy} g(Y_{1-t-r,\delta_T}^{x,v,0},V_{1-t-r,\delta_T}^{v,0})\tilde{y}_{t,r}^2\big]\big|\leq\mrO(\epsilon_T^{-2}\mbE[\tilde{y}_{t,r}^2])\leq\mrO(\epsilon_T^{-2}r),
	\end{align*}
	combining with the fact that $|\tilde{y}_{t,r}|\leq r$, we complete the proof.
\end{proof}

\begin{lem}\label{Lem of minor probability}
	For $\kappa_{T_2}$ in \eqref{Eq of kappa 1} such that $\lim_{T\to\infty}|\kappa_{T_2}|\leq\mrO(1)$ and $\epsilon=\epsilon_T>0$ such that $\lim_{T\to\infty}\epsilon=0$ and $\lim_{T\to\infty}\sqrt{T}\epsilon^2=\infty$, we have
	\begin{align*}
		\frac{1}{T_2}\sum_{t=1}^{T_2-1}\mbP(|\fM_t|<\epsilon)\leq\mrO(\epsilon).
	\end{align*}
\end{lem}
\begin{proof}
	In the following proof, we simplify $\kappa=\kappa_{T_2}$ and $\beta=\beta_{T_2}$. Without loss of generality, we assume that $\kappa\leq0$ (under $H_0$). Let $g(x)$ be a smooth function such that $1_{|x|\leq\epsilon}\leq g(x)\leq 1_{|x|\leq2\epsilon}$, we define the following function
	\begin{align*}
		f(x)=\left\{\begin{array}{ll}
			2\int_0^x e^{2|\kappa| y}\big(\int_0^y(g(z)-C_g)e^{-2|\kappa| z}dz\big)dy, & x\geq0, \\[6pt]
			-2\int_0^x e^{-2|\kappa| y}\big(\int_0^{-y}(g(z)-C_g)e^{-2|\kappa| z}dz\big)dy, & x<0,
		\end{array}\right.
	\end{align*}
	where $C_g=2|\kappa|\int_0^{\infty}g(z)e^{-2|\kappa| z}dz$. Actually, $f(x)$ is the solution of the following ordinary differential equation (ODE):
	\begin{align}
		\kappa\sign(x)f'(x)+\frac{1}{2}f''(x)=g(x)-C_g.\label{Eq of ODE}
	\end{align}
	By direct calculations, we have
	\begin{align*}
		f'(x)=\left\{\begin{array}{ll}
			2 e^{2|\kappa| x}\int_0^x(g(z)-C_g)e^{-2|\kappa| z}dz, & x\geq0, \\[6pt]
			-2 e^{-2|\kappa| x}\int_0^{-x}(g(z)-C_g)e^{-2|\kappa| z}dz, & x<0,
		\end{array}\right.
	\end{align*}
	so we have $\sup_{x\in\mbR}|f'(x)|\leq C_g$. Moreover, by \eqref{Eq of ODE}, we further can show that $|f''(x)|\leq\mrO(1)$ and $|f'''(x)|\leq\epsilon^{-1}$. By the same arguments as those for \eqref{Eq of mbI 1}, we have
	\begin{align*}
		&\mbE[f(\fM_t)]-\mbE[f(\fM_{t-1})]\\
		&=\frac{\kappa}{T_2}\mbE[\sign(\fM_{t-1})f'(\fM_{t-1})]+\frac{1}{2T_2}\mbE[(1-2\beta\sign(\fM_{t-1})\mcS_{t-1})f''(\fM_{t-1})]+\mrO(\epsilon^{-1}T_2^{-3/2})\\
		&=\frac{1}{T_2}\mbE[(1-2\beta\sign(\fM_{t-1})\mcS_{t-1})(g(\fM_{t-1})-C_g)]+\frac{2\kappa\beta}{T_2}\mbE[\mcS_{t-1}f'(\fM_{t-1})]+\mrO(\epsilon^{-1}T_2^{-3/2}),
	\end{align*}
	which implies that
	\begin{small}
		\begin{align*}
			&\frac{(1-2\beta)}{T_2}\mbP(|\fM_{t-1}|\leq\epsilon)\leq\frac{(1-2\beta)}{T_2}\mbE[g(\fM_{t-1})]\leq\frac{1}{T_2}\mbE[(1-2\beta\sign(\fM_{t-1})\mcS_{t-1})g(\fM_{t-1})]\\
			&=\mbE[f(\fM_t)]-\mbE[f(\fM_{t-1})]+\frac{C_g}{T_2}\mbE[1-2\beta\sign(\fM_{t-1})\mcS_{t-1}]-\frac{2\kappa\beta}{T_2}\mbE[\mcS_{t-1}f'(\fM_{t-1})]+\mrO(\epsilon^{-1}T_2^{-3/2})\\
			&\leq\mbE[f(\fM_t)]-\mbE[f(\fM_{t-1})]+\frac{(1+2\beta)C_g}{T_2}+\frac{2|\kappa|\beta}{T_2}\mbE[|f'(\fM_{t-1})|]+\mrO(\epsilon^{-1}T_2^{-3/2})\\
			&\leq\mbE[f(\fM_t)]-\mbE[f(\fM_{t-1})]+\frac{(1+2\beta+2|\kappa|\beta)C_g}{T_2}+\mrO(\epsilon^{-1}T_2^{-3/2}).
		\end{align*}
	\end{small}\noindent
	Hence, we derive that
	\begin{align*}
		\frac{1}{T_2}\sum_{t=1}^{T_2-1}\mbP(|\fM_t|\leq\epsilon)\leq\frac{\mbE[f(\fM_{T_2})-f(0)]}{1-2\beta}+\frac{(1+2\beta+2|\kappa|\beta)C_g}{1-2\beta}+\mrO(\epsilon^{-1}T_2^{-1/2}).
	\end{align*}
	By arguments in Lemma \ref{Lem of situation 2}, we have $\mbE[|\fM_{T_2}|]\leq\mrO(1)$ when $\lim_{T\to\infty}|\kappa|\leq\mrO(1)$. Since $\sup_{x\in\mbR}|f'(x)|\leq C_g$, we have
	\begin{align*}
		\mbE[f(\fM_{T_2})-f(0)]\leq C_g\mbE[|\fM_{T_2}|]\leq\mrO(C_g).
	\end{align*}
	Finally, when $\epsilon>0$ is sufficiently small, $C_g=1-e^{-4|\kappa|\epsilon}\leq4|\kappa|\epsilon$, so we conclude that
	\begin{align*}
		\frac{1}{T_2}\sum_{t=1}^{T_2-1}\mbP(|\fM_t|\leq\epsilon)\leq\mrO(\epsilon)+\mrO(\epsilon^{-1}T_2^{-1/2})\leq\mrO(\epsilon),
	\end{align*}
	which completes the proof.
\end{proof}

\subsection{Situation 2: \texorpdfstring{$|\kappa_{T_2}|\to\infty$}{ }}
In this part, we assume that $\lim_{T\to\infty}|\kappa_{T_2}|=\infty$, let's prove that
\begin{lem}\label{Lem of situation 2}
	Under Assumptions \ref{Ap of high dimensionality} and \ref{Ap of noise}, suppose $\lim_{T\to\infty}|\kappa_{T_2}|=\infty$, we have
	\begin{align*}
		\left\{\begin{array}{cl}
			\fM_{T_2}\overset{\mbP}{\longrightarrow}0, & \kappa_{T_2}<0, \\
			|\fM_{T_2}|\overset{\mbP}{\longrightarrow}\infty, & \kappa_{T_2}>0.
		\end{array}\right.
	\end{align*}
\end{lem}
Combining the above lemma with \eqref{Eq of tail estimation 2}, we can conclude that Theorem \ref{Thm of strategic CLT} when $\lim_{T\to\infty}|\kappa_{T_2}|=\infty$. Here, we define an auxiliary variable as below:
\begin{align*}
	G_t=\left|\sum_{s=1}^t\theta_s X_s\right|.
\end{align*}
Note that $|\fM_t|=\frac{G_t}{\sqrt{T_2(\tau^2+\sigma^2)}}$, to conclude Lemma \ref{Lem of situation 2}, it suffices to analyze the asymptotic behavior of $G_t$.

\subsubsection{\texorpdfstring{$\kappa_{T_2}<0$}{Case 1}}
Since $\theta_t=-\sign\big(\sum_{s=1}^{t-1}\theta_s X_s)$, we have
\begin{align*}
	\sum_{s=1}^t\theta_s X_s=\left\{\begin{array}{cc}
		A_{t-1}-X_t, & \sum_{s=1}^{t-1}\theta_s X_s>0, \\
		X_t-A_{t-1}, & \sum_{s=1}^{t-1}\theta_s X_s\leq 0,
	\end{array}\right.\implies G_t=|G_{t-1}-X_t|.
\end{align*}
For any $a\geq0$ and $x\in\mbR$, since $|a-x|^3\leq a^3-3a^2x+2ax^2+|x|^3$, then
\begin{align*}
	&\mbE[G_t^3|\mcF_{t-1}]\leq G_{t-1}^2-3G_{t-1}^2\mbE[X_t|\mcF_{t-1}]+3G_{t-1}\mbE[X_t^2|\mcF_{t-1}]+\mbE[|X_t|^3|\mcF_{t-1}],
\end{align*}
where $\mcF_{t-1}$ is defined in \eqref{Eq of sigma field} and $\mbE[X_t|G_{t-1}]=\tau+\bbmu'\bbGa\fy_{T_1}$ for $\fy_{T_1}=\frac{1}{T_1}\sum_{s=1}^{T_1}\bby_s$, and
\begin{align*}
	&\mbE[X_t^2|\mcF_{t-1}]=\frac{1}{T_1^2}\sum_{s_1,s_2=1}^{T_1}\bbx_{s_1}'(\bbmu\bbmu'+\bbSig)\bbx_{s_2}-2\tau d_0-2d_0\bbmu'\bbGa\fy_{T_1}+d_0^2\\
	&=\tau^2+\bbmu'\bbSig\bbmu-2d_0\bbmu'\bbGa\fy_{T_1}+2\bbmu'(\bbmu\bbmu'+\bbSig)\bbGa\fy_{T_1}+\fy_{T_1}'\bbGa'(\bbmu\bbmu'+\bbSig)\bbGa\fy_{T_1}\\
	:&=\tau^2+\bbmu'\bbSig\bbmu+\tilde{\sigma}^2.
\end{align*}
Since $|\kappa_{T_2}|\to\infty$, then \eqref{Eq of kappa 1} implies that $\sqrt{T_2}\tau\to\infty$. By Lemma \ref{Lem of Consistence} and Assumptions \ref{Ap of high dimensionality}, \ref{Ap of noise}, we can show that $\tilde{\sigma}^2\overset{\mbP}{\longrightarrow}\sigma^2>0$ and $\tau+\bbmu'\bbGa\fy_{T_1}>\tau/2$ in probability. Next, define an event
$$\mcE:=\left\{|\tilde{\sigma}^2-\sigma^2|<\epsilon\quad\text{and}\quad\tau+\bbmu'\bbGa\fy_{T_1}>\tau/2\right\},$$
for a fixed $\epsilon>0$, and we have $\lim_{T \to \infty}\mbP(\mcE)=1$. Moreover, by the Minkowski's inequality and Hölder's inequality,
\begin{align*}
	&\mbE[|X_t|^3|\mcF_{t-1}]\leq\mrO(\mbE[|X_t-\tau|^3|\mcF_{t-1}]+|\tau|^3)\leq\mrO(\mbE[|X_t-\tau|^4|\mcF_{t-1}]^{3/4}+|\tau|^3),
\end{align*}
where
\begin{align*}
	&\mbE[|X_t-\tau|^4|\mcF_{t-1}]\leq\mrO(\mbE[\{(\bbmu'\bbGa+\fy_{T_1}'\bbGa'\bbGa)\bby_{t+T_1}\}^4|\mcF_{t-1}]+(\bbmu'\bbGa\fy_{T_1})^4)\\
	&=\mrO(\|\bbmu'\bbGa+\fy_{T_1}'\bbGa'\bbGa\|_2^2+(\bbmu'\bbGa\fy_{T_1})^4),
\end{align*}
it implies that $\mbE[|X_t|^3|\mcF_{t-1},\mcE]\leq\mrO(1)$. Consequently, there exists two constants $C_2,C_3>0$ such that
\begin{align*}
	&\mbE[G_t^3|\mcF_{t-1},\mcE]\leq G_{t-1}^3-\frac{3\tau}{2}G_{t-1}^2+3C_2G_{t-1}+C_3.
\end{align*}
Note that $3C_2G_{t-1}\leq\frac{3\tau}{t}G_{t-1}^2+\frac{3C_2^2}{\tau}$, substituting this bound into above inequality and taking the conditional expectaion on $\mcE$, we have
\begin{align*}
	\mbE[G_t^3|\mcE]\leq \mbE[G_{t-1}^3\mcE]-\frac{3\tau}{4}\mbE[G_{t-1}^2|\mcE]+\frac{3C_2^2}{\tau}+C_3.
\end{align*}
Summing from $t=1$ to $T_2$, it yields that
\begin{align*}
	\mbE[G_{T_2}^3|\mcE]\leq-\frac{3\tau}{4}\sum_{t=1}^{T_2}\mbE[G_{t-1}^2|\mcE]+T_2\left(\frac{3C_2^2}{\tau}+C_3\right),
\end{align*}
where we use the fact that $G_0=0$. Since $\mbE[G_{T_2}^3|\mcE]\geq0$, it gives that
\begin{align*}
	\frac{1}{T_2}\sum_{t=1}^{T_2}\mbE[G_{t-1}^2|\mcE]\leq\frac{4C_2^2}{\tau^2}+\frac{4C_3}{3\tau}\leq \tilde{C}_0\tau^{-2}
\end{align*}
for a constant $\tilde{C}_0>0$. Similarly, by repeating the above arguments, we can derive that
\begin{align*}
	\mbE[G_t^2|\mcE]\leq \mbE[G_{t-1}^2|\mcE]+C_2,
\end{align*}
so $\mbE[G_{T_2}^2|\mcE]-\mbE[G_{T_2-t}^2|\mcE]\leq tC_2$. Let $N=\lfloor \mbE[G_{T_2}^2|\mcE]/\max\{C_2,1\}\rfloor$, then $N<T_2$ and
\begin{align*}
	&\sum_{t=1}^{T_2}\mbE[G_t^2|\mcE]\geq\sum_{t=1}^N(\mbE[G_{T_2}^2|\mcE]-tC_2)=N\left(\mbE[G_{T_2}^2|\mcE]-\frac{(N-1)C_2}{2}\right)\\
	&\geq \frac{N\mbE[G_{T_2}^2|\mcE]}{2}\geq\frac{\mbE[G_{T_2}^2|\mcE]^2}{2\max\{C_2,1\}}-\frac{\mbE[G_{T_2}^2|\mcE]}{2},
\end{align*}
where we use the fact that $\frac{(N-1)C_2}{2}\leq\frac{C_2\mbE[G_{T_2}^2|\mcE]}{2\max\{C_2,1\}}\leq\frac{\mbE[G_{T_2}^2|\mcE]}{2}$. In summary, we obtain that
\begin{align*}
	\frac{\mbE[G_{T_2}^2|\mcE]^2}{2\max\{C_2,1\}}-\frac{\mbE[G_{T_2}^2|\mcE]}{2}\leq \frac{\tilde{C}_0T_2}{\tau^2}\implies\mbE[G_{T_2}^2|\mcE]\leq\mrO(\sqrt{T_2}/\tau).
\end{align*}
Finally, by the Chebyshev's inequality,
\begin{align*}
	\lim_{T\to\infty}\mbP\left(\frac{G_{T_2}}{\sqrt{T_2}}>\delta\Big|\mcE\right)\leq\mrO\left(\lim_{T\to\infty}\frac{\mbE[G_{T_2}^2|\mcE]}{\delta^2T_2}\right)\leq\mrO\left(\lim_{T\to\infty}\frac{1}{\sqrt{T_2}\tau}\right)=0,
\end{align*}
where we use the fact that $\lim_{T\to\infty}\sqrt{T_2}\tau=\infty$.

\subsubsection{\texorpdfstring{$\kappa_{T_2}>0$}{Case 2}}
Recall that $G_t=|G_{t-1}-X_t|$, then
\begin{align*}
	G_t\geq G_{t-1}-X_t\implies G_t-G_{t-1}\geq-X_t,
\end{align*}
and
\begin{align*}
	\frac{G_{T_2}}{\sqrt{T_2}}=\frac{1}{\sqrt{T_2}}\sum_{t=1}^{T_2}G_t-G_{t-1}\geq-\frac{1}{\sqrt{T_2}}\sum_{t=1}^{T_2}X_t.
\end{align*}
Since
\begin{align*}
	X_t=\tau+\bbmu'\bbGa\bby_{t+T_1}+\bbmu'\bbGa\fy_{T_1}+\bby_{t+T_1}'\bbGa'\bbGa\fy_{T_1},
\end{align*}
we have
\begin{align*}
	-\frac{1}{\sqrt{T_2}}\sum_{t=1}^{T_2}X_t=-\sqrt{T_2}(\tau+\bbmu'\bbGa\fy_{T_1})-\frac{1}{\sqrt{T_2}}\sum_{t=1}^{T_2}\bbmu'\bbGa\bby_{t+T_1}-\frac{1}{\sqrt{T_2}}\sum_{t=1}^{T_2}\bby_{t+T_1}'\bbGa'\bbGa\fy_{T_1}.
\end{align*}
Note that
\begin{align*}
	\sqrt{T_2}\bbmu'\bbGa\fy_{T_1}=\sqrt{\frac{T_2}{T_1}}\frac{1}{\sqrt{T_1}}\sum_{s=1}^{T_1}\bbmu'\bbGa\bby_s,
\end{align*}
where $\frac{1}{\sqrt{T_1}}\sum_{s=1}^{T_1}\bbmu'\bbGa\bby_s$ is asymptotically normal, so does $\frac{1}{\sqrt{T_2}}\sum_{t=1}^{T_2}\bbmu'\bbGa\bby_{t+T_1}$. Moreover, since
\begin{align*}
	\frac{1}{\sqrt{T_2}}\sum_{t=1}^{T_2}\bby_{t+T_1}'\bbGa'\bbGa\fy_{T_1}=\frac{1}{T_1\sqrt{T_2}}\sum_{s=1}^{T_1}\sum_{t=1}^{T_2}\bby_{t+T_1}'\bbGa'\bbGa\bby_s,
\end{align*}
which is asymptotically normal by \cite{chen2010two}, combining with $\lim_{T\to\infty}-\sqrt{T_2}\tau=\infty$, we derive that 
$$\frac{G_{T_2}}{\sqrt{T_2}}\longrightarrow\infty$$
in probability.

\newpage

\end{document}